\documentclass[12pt, a4paper]{article} 
\usepackage{setspace}
\usepackage[small, font = onehalfspacing]{caption}
\usepackage[figuresleft]{rotating}
\usepackage[authoryear]{natbib}
\bibpunct{(}{)}{;}{a}{,}{,}
\usepackage[english]{babel}
\usepackage[latin1]{inputenc}
\usepackage[T1]{fontenc}
\usepackage{lmodern}
\usepackage{amssymb}
\usepackage{xcolor,graphicx,xspace,caption}
\usepackage{listings,eso-pic,subfigure}
\usepackage{lscape}
\usepackage{multirow}
\usepackage{enumerate}
\usepackage{amsmath}
\usepackage{mathrsfs}
\usepackage{url}
\usepackage{amssymb}
\usepackage{amsmath}
\usepackage{color}
\usepackage{amsmath,amssymb,amsthm}

\pdfpageattr {/Group << /S /Transparency /I true /CS /DeviceRGB>>}

\newcommand{\by}{\mathbf{y}}

\def\min{\mathbf{min}}
\def\arg{\mathbf{arg}}

\def\me{\mathbb{e}}

 \def\mm{\mathbf{m}}
 \def\mbu{\mathbf{u}}

\def\mx{\mathbf x}

\def\mW{\mathbf{W}}
\def\P{\mathbb{P}}

\def\mmW{\mathbf{W}}

\def\my{\mathbf{y}}

\def\vps{\varepsilon}
\newcommand{\smallO}{\mbox{\tiny $\mathcal{O}$}}

\newtheorem{theorem}{Theorem}
\definecolor{hint}{RGB}{191,63,0}
\definecolor{hellgelb}{rgb}{1,1,0.8}
\definecolor{colKeys}{rgb}{0,0,1}
\definecolor{colIdentifier}{rgb}{0,0,0}
\definecolor{colComments}{rgb}{1,0,0}
\definecolor{colString}{rgb}{0,0.5,0}

\usepackage{hyperref}
\hypersetup{linkcolor={blue}}

\newcommand{\E}{\mathop{\mbox{\sf E}}}

\newcommand{\IF}{\mathbf{I}}

\newcommand{\Cov}{\mathop{\mbox{Cov}}}

\newcommand{\Var}{\mathop{\mbox{\sf Var}}}

\def\ssup{\mbox{sup}_{\theta \in \Theta, \gamma \in \Gamma}}

\newcommand{\CO}{{\mathcal{O}}}
\def\Co{{\scriptstyle{\mathcal{O}}}}

\newcommand{\rit}{\mathbb{R}}

\def\defeq{\stackrel{\mathrm{def}}{=}}  
\newcommand{\II}{\mathbf{I}}

\newtheorem{proposition}{Proposition}
\newtheorem{lemma}{Lemma}

\newtheorem{definition}{Definition}

\usepackage[top = 1in, bottom = 1in, left = 1in, right = 1in]{geometry}
\renewcommand{\arraystretch}{1.3}

\title{Using generalized estimating equations to estimate nonlinear models
with spatial data \thanks{%
This paper is supported by the National Natural Science Foundation of China,
No.71601094 and German Research Foundation.}}
\author{Cuicui Lu\thanks{%
Department of Economics, Nanjing University Business School, Nanjing,
Jiangsu 210093 China; email: lucuicui@nju.edu.cn} ,  Weining Wang \footnote{Department of Economics, City, U of London;
Northampton Square, Clerkenwell, London EC1V 0HB. Humboldt-Universit\"{a}t zu Berlin, C.A.S.E. - Center for Applied Statistics and Economics; email: weining.wang@city.ac.uk}, Jeffrey M. Wooldridge\thanks{%
Department of Economics, Michigan State University, East Lansing, MI 48824
USA; email: wooldri1@msu.edu} }

\date{\;}

\begin{document}
\maketitle

\begin{abstract}

In this paper, we study estimation of nonlinear models with cross sectional data using two-step generalized estimating equations (GEE) in the quasi-maximum likelihood estimation (QMLE) framework.
In the interest of improving efficiency, we propose a grouping estimator to account for the potential spatial correlation in the underlying innovations. We use a Poisson
model and a Negative Binomial II model for count data and a Probit model for binary
response data to demonstrate the GEE procedure. Under mild weak dependency assumptions, results on estimation consistency and asymptotic normality are provided.
Monte Carlo simulations show efficiency gain of our approach in comparison of
different estimation methods for count data and binary response data. Finally we apply the GEE approach to study the determinants of the inflow foreign direct investment (FDI) to China.

\textbf{keywords:}
quasi-maximum likelihood estimation; generalized estimating equations;
nonlinear models; spatial dependence; count data; binary response data; FDI
equation

\end{abstract}

\newpage

\section{ Introduction}

In empirical economic and social studies, there are many examples of discrete
data which exhibit spatial or cross-sectional correlations possibly due to the closeness of geographical locations of individuals or agents. One example is the technology spillover effect. The number of patents a firm received shows correlation with that received by other nearby firms (E.g. \cite{bloom2013identifying}).
Another example is the neighborhood effect. There is a causal effect between
the individual decision whether to own stocks and the average stock market
participation of the individual's community (E.g.\cite{brown2008neighbors}). These two examples involves dealing with discrete data. The first
example is concerned with count data and the second one handles binary response data. Nonlinear models are more
appropriate than linear models for discrete response data. With spatial correlation,
these discrete variables are no longer independent. Both the nonlinearity and the
spatial correlation make the estimation difficult.

In order to estimate nonlinear models, one way is to use maximum likelihood estimation (MLE).
A full MLE specifies the joint distribution of spatial random variables. This includes correctly specifying the marginal and the conditional distributions, which impose very strong assumptions on the data generating processes.  However, given a spatial data set, the dependence structure is generally unknown. If the joint distribution of the
variables is misspecified, MLE is in general not consistent. One of the alternative MLE method is partial-maximum likelihood estimation (PMLE), which only uses marginal distributions. \cite{wang2013partial} use a bivariate Probit partial MLE to improve the estimation efficiency with a
spatial Probit model. Their approach requires to correctly
specify the marginal distribution of the binary response variable
conditional on the covariates and distance measures\footnote{%
A sample of spatial data is collected with a set of geographical locations.
Spatial dependence is usually characterized by distances between
observations. A distance measure is how one defines the distances between
observations. Physical distance or economic distance could be two options.
Information about agents locations is commonly imprecise, e.g. only zip code
is known. \cite{conley2007spatial} deals with the inference problem when
there exist distance errors. In this paper we assume there are no
measurement errors in pairwise distances.} There are two concerns with \cite{wang2013partial}. First the computation is already hard for a bivariate distribution. The multivariate marginal distribution of {a higher dimensional variable}, e.g., trivariate, is more computationally demanding; second it also requires the correct specification of the marginal bivariate distribution to obtain
consistency. The bivariate marginal distribution of a spatial multivariate normal distribution is bivariate normal, thus the bivariate Probit model can be derived. But there are other distributions whose marginal distribution is not the same anymore. For example, the marginal distribution of a multivariate Logit is not logistic. If the partial likelihood is misspecified, the estimation of the mean parameters could be not consistent. With less distributional assumptions, the quasi-maximum likelihood estimation (QMLE) can also be used to estimate nonlinear models. Using a density that belongs to a linear exponential family (LEF), QMLE is consistent if we correctly specify the
conditional mean while other features of the density can be misspecified (\cite{gourieroux1984pseudo}). \cite{lee2004asymptotic} derives asymptotic
distributions of quasi-maximum likelihood estimators for spatial
autoregressive models by allowing not assuming normal distributions. In a
panel data case, pooled or partial QMLE (PQMLE) which ignores serial correlations is
consistent under some regularity conditions (\cite{wooldridge2010econometric}).

We further relax distributional assumptions than those required in
bivariate partial MLE as in \cite{wang2013partial}. Suppose we
only assume correct mean function and one working variance covariance matrix%
\footnote{The true variance covariance matrix is generally unknown. By specifying a
working variance covariance matrix, one can capture some of the correlation
structure between observations.} which may not be correct. Using QMLE in the
LEF, we can consistently estimate the mean parameters as well as the average
partial effects. The generalized estimating equations (GEE) approach is one of the QMLE methods.
It is used in panel data models to account for serial correlation and thus get more efficient
estimators. A generalized estimating equation is used to estimate the
parameters of a generalized linear model with a possible unknown correlation
between outcomes (\cite{liang1986longitudinal}). Parameter estimates from the GEE
are consistent even when the variance and covariance structure is misspecified under mild
regularity conditions. This is quite related to a different terminology,
composite likelihood. \cite{varin2011overview} provide a survey of
developments in the theory and application of composite likelihood. The
motivation for the use of composite likelihood is usually computational, to
avoid computing or modelling the joint distributions of {high dimensional}
random processes. {One can find many related reference in the literature, such as \cite{bhat2010comparison}.} As a special case of composite likelihood methods, one way
is to use partial conditional distribution, and maximize the summand of log
likelihoods for each observation. It assumes a working independence
assumption, which means that the estimators are solved by ignoring
dependence between individual likelihoods. The parameters can be
consistently estimated if the partial log likelihood function satisfies
certain regularity assumptions. However, a consistent variance estimator
should be provided for valid inference\footnote{%
Ignoring dependence in the estimation of parameters will result in wrong
inferences if the variances are calculated in the way that independence is
assumed. Dependence should be accounted for to the extent of how much one
ignores it in the estimation.}. When there exists spatial correlation, the
pooled maximum likelihoods (composite likelihoods) can be considered as
misspecified likelihoods because of the independence assumption.

Generalized least squares (GLS) could be used to improve the estimation efficiency in a linear regression model even if the variance covariance structure is misspecified. { \cite{lu2017quasi} propose a quasi-GLS method to estimate the linear regression model with an spatial error component. By first estimating the spatial parameter in the error variance and then using estimated variance matrix for within group observations, the quasi-GLS is computationally easier and would not loose much efficiency compared to GLS. } Similarly, the multivariate nonlinear weighted least squares estimator (\textrm{MNWLS}), {see Chapter 12.9.2 in \cite{wooldridge2010econometric}}, is essentially a GLS approach applied in nonlinear models to improve the estimation efficiency.

{It is worth noting that the GEE approach discussed in this paper is a two-step method, which is essentially a special MNWLS estimator that uses a \textrm{LEF} variance assumption and a possibly misspecified working correlation matrix in the estimation. The GEE approach was first extended to correlated data by \cite{liang1986longitudinal}, which propose a fully iterated GEE estimator in a panel data setting. In addition, \cite{zeger1986longitudinal} fit the GEE method to discrete dependent variables. The iterated GEE method has solutions which are consistent and asymptotically Gaussian even when the temporal dependence is misspecified. The consistency of mean parameters only depends on the correct specification of the mean, not on the choice of working correlation matrix. GEE used in nonlinear panel data models and system of equations is supposed to obtain more efficient conditional mean parameters with covariance matrix accounting for the dependency structure of the data. In this paper, we apply a similar idea to grouped spatial data. We use the \textrm{PQMLE} as the initial estimator for the two-step GEE and study the efficiency properties of a two-step GEE estimator and expect that GEE can give more efficient estimators compared to PQMLE.}

Moreover, we demonstrate theoretically how to use our GEE approach within the QMLE framework
in a spatial data setting to obtain consistent estimators. We give a
series of assumptions, based on which QMLE estimators are consistent for the
spatial processes. To derive the asymptotics for the GEE estimator we have
to use a uniform law of large numbers (ULLN) and a central limit theorem
(CLT) for spatial data. These limit theorems are the fundamental building
blocks for the asymptotic theory of nonlinear spatial M-estimators, for
example, maximum likelihood estimators (MLE) and generalized method of
moments estimators (GMM) (\cite{Jenish2012}). \cite{conley1999gmm} makes an
important contribution toward developing an asymptotic theory of GMM
estimators for spatial processes. He utilizes \cite{bolthausen1982central} CLT for
stationary random fields. \cite{Jenish2009,Jenish2012} provide ULLNs and a
CLTs for near-epoch dependent spatial processes. Using theorems in \cite{Jenish2009, Jenish2012}, one can analyze more interesting economic phenomena. It should be noted that although GEE can be considered as a special case of M-estimation, we have carefully checked how the near-epoch dependence  property
of the underlying processes is translated to our responses and the partial sum processes involved in proving the asymptotics of the estimation. Our setup is different from the literature as it is with a grouped estimation structure. Finally, we have provided a consistency proof of the proposed semiparametric estimator of the variance covariance matrix.

{We contribute to the literature in three aspects. First, we propose a simple method which uses less distributional assumptions by only specifying the conditional mean for spatial dependent data. The method is computationally easier by dividing data into small groups compared to using all information. We model the spatial correlation as a moving average (MA) type in the underlying innovations instead of the spatial autoregressive (SAR) model in the dependent variable. Second, we proved the theoretical property of our estimator by applying ULLN and CLT in \cite{Jenish2009,Jenish2012} to the GEE estimator with careful checking the hyper assumptions. Third,  we emphasize the possible efficiency gain from making use of spatial correlation from our simulation study, and we demonstrate how to use GEE with two types of data: count and binary response. }

In Section \ref{sec2}, the GEE methodology in a QMLE framework under the spatial data
context is proposed. In Section \ref{sec3}, we
look in detail at a Poisson model and Negative Binomial II model for count data with a multiplicative spatial error term. We further study a Probit model for binary response data with spatial correlation in the latent error term.
In Section \ref{sectheorem}, a series of assumptions are given based on \cite{Jenish2009, Jenish2012} under which GEE-estimators are consistent and have an
asymptotic normal distribution. The asymptotic distributions for GEE for
spatial data are derived. Consistent variance covariance estimators are
provided for the nonlinear estimators.
 Section \ref{sec5} contains Monte Carlo simulation results which compare
efficiency of different estimation methods for the nonlinear models
explored in the previous section.
 Section \ref{sec6}  contains an application to study the determinants of the inflow FDI to China using city level data.  The technical details are delegated to Section \ref{sec7}.

\section{Methodology} \label{sec2}

\subsection{Notation and definition}

Unlike linear models, a very important feature of nonlinear models is that
estimators cannot be obtained in a closed form, which requires new tools for
asymptotic analysis: uniform law of large numbers (\textrm{ULLN}) and a
central limit theorem (\textrm{CLT}). \cite{Jenish2009} develop
\textrm{ULLN} and \textrm{CLT} for $\alpha $-mixing random fields on
unevenly spaced lattices that allow for nonstationary processes with
trending moments. But the mixing property can fail for quite a few reasons, thus
we adopt the notion of near-epoch dependence (\textrm{NED}) as in
\cite{Jenish2012} which refers to a generalized class of random
fields that is "closed with respect to infinite transformations."
We consider spatial processes located on a unevenly spaced lattice $D \subseteq \rit^{d}, d \geq 1$. The space $\rit^{d}$ is endowed with the metric $\rho(i, j) = max_{1 \leq l \leq d}|j_{l} - i_{l}|$ with the corresponding norm $|i|_{\infty} = max_{1 \leq l \leq d}|i_{l}|$, where $i_{l}$ is the $l$-th component of $i$. The distance between any subsets $U, V \in D$ is defined as $\rho(U, V) = \inf \{\rho(i, j): i\in U \text{ and } j\in V \}$. Further, let $|U|$ denote the cardinality of a finite subset $ U\subseteq D$. The setting is illustrated in \citet{Jenish2009,Jenish2012}.

Let $Z = \{Z_{n, i}, i \in D_{n}, n \geq 1\}$ and $\varepsilon = \{\varepsilon_{n, i}, i \in T_{n}, n \geq 1\}$ be triangular arrays of random fields defined on a probability space $(\Omega, \mathscr{F}, P)$ with $D_{n}\subseteq T_{n} \subseteq D$ where $D$ satisfies A.1). The cardinality of $D_{n}$ and $T_{n}$ satisfy $\displaystyle \lim_{n \rightarrow \infty} |D_{n}|\rightarrow \infty, \displaystyle \lim_{n \rightarrow \infty} |T_{n}|\rightarrow \infty$. For any vector $v \in R^p$, $|v|_2$ denotes the $L_2$ norm of $v$. For any $n \times m$ matrix $A$ with element $a_{ij}$, denote $|A|_1 = \displaystyle \max_{1\leq j\leq m} \displaystyle \sum^{n}_{i=1} |a_{ij}|$ and $|A|_{\infty} = \displaystyle \max_{1\leq i\leq n} \displaystyle \sum^{m}_{j=1} |a_{ij}|$, $|A|_2$ denotes the 2-norm. For any random vector $X$, denote $\| X_{n,i} \|_{p} = (\E|X_{n,i}|^{p})^{1/p}$ as its $L_{p}$-norm, where the absolute $p$th moment exists.  We  brief $\| X_{n,i} \|_{2} $ as $\| X_{n,i} \|$. Let $\mathcal{F}_{n,i}(s) = \sigma(\varepsilon_{n,j}: j \in D_{n}, \rho(i, j) \leq s)$ as the $\sigma$- field generated by random vectors $\varepsilon_{n, j}$ located within distance $s$ from $i$.
 Given two sequences of positive numbers $x_n$ and $y_n$, write $x_n\lesssim y_n$ if there exists constant $C>0$ such that $x_n/y_n\leq C$, also we can write $x_n = \CO(y_n)$. A sequence $x_n$ is said to be $\Co(y_n)$ if $x_n/y_n \to 0,$ as $n \to \infty$. In a similar manner,
 The notation, $X_n=\CO_p(a_n)$ means that the set of values $X_n/a_n$ is stochastically bounded.  That is, for any $ \varepsilon > 0$, there exists a finite M > 0 and a finite N > 0 such that,
$P(|X_n/a_n| > M) < \varepsilon, \forall n > N$. $|.|_a$ is the elementwise absolute value of a matrix  $|A|_a$. $a\vee b $ is $ \max(a,b).$

\begin{definition}\label{DefNED}
Let $Z = \{Z_{n, i}, i \in D_{n}, n \geq 1\}$ and $\varepsilon = \{\varepsilon_{n, i}, i \in D_{n}, n \geq 1\}$ be random fields with $\| Z_{n,i} \|_{p} < \infty, p \geq 1 $, where $D_{n} \subseteq D$ and its cardinality $|D_{n}|=n$. Let $\{d_{n, i}, i \in D_{n}, n \geq 1\}$ be an array of finite positive constants. Then the random field $Z$ is said to be $L_{p}$-near-epoch dependent on the random field $\varepsilon$ if
\begin{equation*}
 \| Z_{n, i} - \E(Z_{n, i}|\mathcal{F}_{n, i}(s)) \|_{p} < d_{n, i}\varphi(s)
\end{equation*}
for some sequence $\varphi(s) \geq 0$ with $\displaystyle \lim_{s \rightarrow \infty} \varphi(s) = 0$. $\varphi(s)$ are denoted as the NED coefficients, and $d_{n, i}$ are denoted as NED scaling factors. If $\displaystyle \sup_{n}\sup_{i \in D_{n}} d_{n, i} < \infty $, then $Z$ is called as uniformly $L_{p}$-NED on $\varepsilon$.
\end{definition}

\begin{itemize}\label{Ass1}
\item[A.1)] The lattice $D \subseteq \rit^{d}, d \geq 1$, is infinitely countable. The distance $\rho(i, j)$ between any two different individual units $i$ and $j$ in $D$ is at least larger than a positive constant, i.e., $\forall i, j \in D: \rho(i, j) \geq \rho_{0} $, w.l.o.g. we assume $\rho_{0} >1 $.
\end{itemize}

We will present the $L_{2}$-NED properties of a random field $Z$ on some $\alpha$-mixing random field $\varepsilon$. The definition of the $\alpha$-mixing coefficient employed in the paper are stated as following.

\begin{definition}\label{Def2}
Let $\mathscr{A}$ and $\mathscr{B}$ be two $\sigma$-algebras of $\mathscr{F}$, and let
\begin{equation*}
  \alpha(\mathscr{A}, \mathscr{B}) = \sup(|P(A\cap B) - P(A)P(B)|, A \in \mathscr{A}, B \in \mathscr{B}),
\end{equation*}
For $U \subseteq D_{n}$ and $V \subseteq D_{n}$, let $\sigma_{n}(U) = \sigma(\varepsilon_{n, i}, i \in U)$ ($\sigma_{n}(V) = \sigma(\varepsilon_{n, i}, i \in V)$) and $\alpha_{n}(U, V) = \alpha(\sigma_{n}(U), \sigma_{n}(V))$. Then, the $\alpha$-mixing coefficients for the random field $\varepsilon$ are defined as:
\begin{equation*}
 \overline{\alpha}(u, v, h) = \displaystyle \sup_{n} \displaystyle \sup_{U,V}(\alpha_{n}(U, V), |U| \leq u, |V| \leq v, \rho(U, V)\geq h).
\end{equation*}
\end{definition}

Note that we suppress the dependence on $n$ from now on for the triangular array.
Let $\left \{ \left(\mathbf{x}_{i},y_{i}\right) ,i=1,2,...,n\right \} $, where $\left(
\mathbf{x}_{i},y_{i}\right) $ is the observation at location $s_{i}.$ $%
\mathbf{x}_{i}$ is a row vector of independent variables which can be
continuous, discrete or a combination. The dependent variable $y_{i}$ can be
continuous or discrete. Let $\left( \mathbf{x}_{g},\mathbf{y}%
_{g}\right) $ be the observations in group $g$ and $B_{g}$ is the associated set of locations within the group $g$. We will focus on the case of a discrete dependent variable, a
binary response and a count. Let $\theta \in \mathbf{R}^p$, $\gamma \in \mathbf{R}^q$ and  $\mathbf{\theta \in \Theta, \gamma \in \Gamma}$, where $\mathbf{\Theta\times \Gamma} $ is a compact set, and $(\theta^0, \gamma^0)$ is the true parameter value.

\subsection{The generalized estimating equations methodology}

{The \textrm{GEE} methodology proposed in equations (6) and (7) in \cite{liang1986longitudinal} is an iterated approach to estimate the mean parameters. We simplify the procedure using a two-step method by first estimate the working correlation matrix and then apply MWNLS.} In the following, we write the \textrm{GEE} methodology in the group level notation. Groups are divided according to geographical properties or other researcher defined economic (social) relationships. Our asymptotic analysis is based on large number of groups $g = 1, \cdots, G$. The notation $D_{G}$ indicates the lattice containing group locations, each group location is denoted as vectorizing the elements in $B_g$. Let the total number of groups
be $|D_G| = G,$ while the total number of observations is still $|D_n| = n.$ Let $L_{g}$ be
the number of observations in group $g$. For simplicity assume $L_{g}=L,$
for all $g$. Let $\left \{ \left( \mathbf{x%
}_{g},\mathbf{y}_{g}\right) \right \} $ be the observations for group $g$,
where $\mathbf{x}_{g}$ is an $L\times p$ matrix and $\mathbf{y}_{g}$ is an $%
L\times 1$ vector$.$ There are two extreme cases of the group size. The
first case is when the group size is $1$, the resulting estimator is the
usual \textrm{PQMLE} estimator, which means we ignore all of the pairwise
correlations. The second case is when the group size is $n$, which means we
are using all the pairwise information. If the group size is not equal to $1$ or
$n$, the estimation is actually a "partial" \textrm{QMLE}. By "partial", we
mean that we do not use full information, but only the information within
the same groups. Note that we work with the case with number of groups $G\to \infty$ in our theory,
{ while the groupsize $L$ is assumed to be fixed.}

Assume that we correctly specify conditional mean of $\mathbf{y}_{g},$ that
is, the expectation of $\mathbf{y}_{g}$ conditional on $\mathbf{x}_{g}$ is
\begin{equation}
\E\left( \mathbf{y}_{g}|\mathbf{x}_{g}\right) =%
\mathbf{m}_{g}\left( \mathbf{x}_{g};\mathbf{\theta }^{0}\right) =\mathbf{m}_{g}(\theta^0) .
\label{groupmean}
\end{equation}%
Assume the conditional variance-covariance matrix of $\mathbf{y}_{g}$ is $\mathbf{W}^*_{g}$
which is unknown in most cases, where $\mathbf{W}_g \defeq \Cov(y_g,y_g|\mx_g)= \E(\mathbf{y}_{g}\mathbf{y}_{g}^{\top}|\mx_g) - \E(\mathbf{y}_{g}|\mx_g)\E(\mathbf{y}_{g}|\mx_g)^{\top}$.  Usually we parameterize a corresponding weight matrix $\mathbf{W}_g$ by $\mathbf{W}_g(  \theta,\gamma)$, where $\theta \in \Theta \subset \mathbf{R}^q$ and $\gamma \in \Gamma \subset \mathbf{R}^p$ as a nuisance parameter involved only in the estimation of the variance covariance matrix.
{In practice, we usually preestimate $\gamma$ and thus it is replaced by a consistent estimate of $\hat{\gamma}$, then $\mathbf W_g$ is denoted as $\mathbf{W}( \theta, \hat{\gamma})$.}

The objective function for group $g$ and the whole sample are given as follows:
\begin{eqnarray}
q_{g}(\theta,\gamma) &\defeq&\left( \mathbf{y}_{g}-\mathbf{m}_{g}(\theta)\right) ^{\top }\mathbf{W}%
_{g}^{-1}\left( \mathbf{\theta,\gamma}\right) \left( \mathbf{y}_{g}-\mathbf{m}%
_{g}(\theta)\right) , \\
Q_{G}(\theta,\gamma) &\defeq& {(M_{G} G)}^{-1}\sum_g q_{g}(\theta) ,
\end{eqnarray}
where {$M_{G}$ is a scaling constant defined in A.5) in section \ref{sectheorem}.}

Theoretically, an estimator of $\theta^0, \gamma^0$ is given by%
\begin{equation}
(\mathbf{\hat{\theta},\hat{\gamma}})=\arg \min_{\mathbf{\theta, \gamma }\in
\mathbf{\Theta, \Gamma }}Q_{G}(\theta,\gamma).  \label{gobj}
\end{equation}%

In practice a GEE estimator is obtained by a two-step procedure, where the first step is
to estimate the nuisance parameter $\gamma$ and the second step is to have the parameter $\theta$ estimated
with the plug-in estimator $\hat{\gamma}$ from step 1.
\begin{equation}
\mathbf{\hat{\theta}}_{_{\mathrm{GEE}}}=\arg \min_{\mathbf{\theta }\in
\mathbf{\Theta }}Q_{G}( \theta, \hat{\gamma}).  \label{gobj1}
\end{equation}%

{Because this only uses the groupwise information, it actually is a "quasi" or "pseudo"
\textrm{MWNLS}. }The quasi-score equation, which is the first order condition
for GEE, is defined as follows:
\begin{equation}
\mathbf{S}_{G}\left(  \theta,{\gamma}\right) =\frac{1}{G M_G}%
\sum_{g}\nabla \mathbf{m}_{g}\left( \mathbf{\theta }\right)^{\top }
\mathbf{W}_{g}^{-1}\left(  \theta,{\gamma}\right) \left[ \mathbf{y}_{g}-%
\mathbf{m}_{g}\left( \mathbf{\theta }\right) \right] ,  \label{quasiscore}
\end{equation}%
where $\nabla _{\theta }\mathbf{m}_{g}\left( \mathbf{\theta }\right) $ is
the gradient of $\mathbf{m}_{g}\left( \mathbf{\theta }\right) .$ $M_G$ is defined as the scaling constant in A.5) in section \ref{sectheorem}. The GEE
estimator $(\mathbf{\hat{\theta},\hat{\gamma}}) = \mbox{argzero}_{\theta \in \Theta, \gamma \in \Gamma} \mathbf{S}_{G}\left( \theta, \mathbf{{\gamma}}\right) .$

Denote the population version of loss as
$\mathbf{S}_{\infty}\left( \theta,\gamma\right) = \mbox{lim}_{G\to \infty} \E \mathbf{S}_{G}\left(\theta,\gamma\right),$ and \\$Q_{\infty}(\theta,\gamma) = \lim_{G\to\infty} {(G M_G)}^{-1}\sum_g \E q_{g}(\theta,\gamma) .$  Thus the true parameter $( \theta^0,\gamma^0) \\= \mbox{argzero}_{\theta \in \Theta, \gamma \in \Gamma } \mathbf{S}_{\infty}\left(  \mathbf{\theta}, \mathbf{\gamma} \right) = \mbox{argmin}_{ \theta \in \Theta, \gamma \in \Gamma  }Q_{\infty}( \theta,\gamma).$

Frequently we restrict our attention to the exponential family, which embraces  many frequency encountered distributions, such as Bernoulli, Poisson and Gaussian, etc.

Now we write this estimation in a \textrm{QMLE} framework. We suppress the parameter $\gamma$ for a moment. Assume the
probability density function $f\left( \my_g|\mathbf{x}_g;\mathbf{\theta }\right) $
is in the LEF.({ See details in Appendix \ref{exp} }.)

Without accounting for the spatial covariance, one characterization of \textrm{QMLE} in \textrm{LEF} is that the individual score function has the following form:
\begin{equation}
\mathbf{s}_{i}\left( \mathbf{\theta }\right) =\nabla m_{i}\left(\mathbf{\theta }\right) ^{\top}\{ y_{i}-m_{i}\left(\mathbf{\theta }\right) \}/v_{i}\left( m_{i}\left(
\mathbf{\theta }\right) \right) ,  \label{pqmlescore}
\end{equation}%
where $\nabla m_{i}\left( \mathbf{x}_{i};\mathbf{\theta }%
\right) $ is the $1\times p$ gradient of the mean function and $v_{i}\left(
m_{i}\left( \mathbf{x}_{i},{D}_{n};\mathbf{\theta }\right) \right) $
is the {conditional} variance function associated with the chosen \textrm{LEF} density.
For Bernoulli distribution, $v_{i}\left( m_{i}\left( \mathbf{x}_{i};\mathbf{%
\theta }\right) \right) =m_{i}\left( \mathbf{x}_{i};\mathbf{\theta }\right)
\left( 1-m_{i}\left( \mathbf{x}_{i};\mathbf{\theta }\right) \right) ,$ and
for Poisson distribution, $v_{i}\left( m_{i}\left( \mathbf{x}_{i};\mathbf{%
\theta }\right) \right) =m_{i}\left( \mathbf{x}_{i};\mathbf{\theta }\right)
. $
Note that (\ref{pqmlescore}) gives a consistent estimator but is not likely
to be the most efficient estimator as it ignores the possible spatial
correlations between observations. However, it accounts for possible
heteroscedasticity.

We write the quasi-score function for a group. Let $\mathbf{v}%
_{g}\left( \mathbf{m}_{g}\left( \mathbf{x}_{g};\mathbf{\theta }\right)
\right) $ be the conditional variance covariance matrix for group $g$. Then score involved in the estimation is denoted as
\begin{equation}
\mathbf{S}_{G}\left( \mathbf{\theta}\right) =\frac{1}{M_{G}G}%
\sum_{g}{s}_{g}\left( \mathbf{\theta }\right) =\frac{1}{ M_{G}G}%
\sum_{g}\nabla \mathbf{m}_{g}\left( \mathbf{x}_{g};\mathbf{\theta }%
\right) ^{\top }\mathbf{v}_{g}\left( \mathbf{m}_{g}\left( \mathbf{x}_{g};%
\mathbf{\theta }\right) \right) ^{-1}\left[ \mathbf{y}_{g}-\mathbf{m}%
_{g}\left( \mathbf{x}_{g};\mathbf{\theta }\right) \right] ,
\label{quasiscore2}
\end{equation}%
where
\begin{equation}
{s}_{g}\left( \mathbf{\theta }\right) =\nabla \mathbf{m}_{g}\left(
\mathbf{x}_{g};\mathbf{\theta }\right) ^{\top }\mathbf{v}_{g}\left(
\mathbf{m}_{g}\left( \mathbf{x}_{g};\mathbf{\theta }\right) \right) ^{-1}%
\left[ \mathbf{y}_{g}-\mathbf{m}_{g}\left( \mathbf{x}_{g};\mathbf{\theta }%
\right) \right] .
\end{equation}%

We specify a more general form of variance $\mathbf{v}_{g}\left( \theta \right)$ with the dependency  of the nuisance parameter $\gamma$.
The conditional mean vector is correctly specified for each individual
\textrm{E}$\left( y_{i}|\mathbf{x}_{i}\right) =m_{i}\left( \mathbf{x}_{i};%
\mathbf{\theta }^{0}\right) .$ Thus for each group, $%
\mathbf{m}_{g}\left( \mathbf{x}_{g};\mathbf{\theta }^0\right) =\E%
\left( \mathbf{y}_{g}|\mathbf{x}_{g}\right) .$ Let ${s}_{g}\left(
\theta,\mathbf{{\gamma}}\right) $ denote the $p\times 1$
vector of score for group $g$.
Let $h_{g}\left( \mathbf{
\theta, {\gamma}}\right) $ be the $p\times p$ matrix of Hessian for group $g$.
The score function for $Q_{G}\left( \mathbf{\theta,{\gamma} }\right) $
can be defined as $\mathbf{S}_{G}\left(  \theta, {\gamma}\right)
$ and the Hessian can be defined as $\mathbf{H}_{G}\left( \mathbf{ \theta, {\gamma}
}\right) .$
The score function
for \textrm{GEE} can be written as%
\begin{equation}
\mathbf{S}_{G}\left( \theta, \mathbf{{\gamma} }\right) =\frac{1}{M_{G}G}%
\sum_{g}{s}_{g}\left( \mathbf{
\theta, \gamma }\right) =\frac{1}{ M_{G} G}\sum_{g}\nabla \mathbf{m}_{g}^{\top}\left(
\mathbf{\theta }\right) \mathbf{{W}}_{g}^{-1}(\theta, \gamma)\left[ \mathbf{y}_{g}-%
\mathbf{m}_{g}\left( \mathbf{\theta }\right) \right] .
\end{equation}%
and the Hessian is
\begin{eqnarray}\label{eq:hessian}
\mathbf{H}_{G}\left( \mathbf{\theta ,{\gamma}}\right) &\equiv &\frac{1}{M_{G}G%
}\sum_{g}{h}_{g}\left( \mathbf{\theta ,{%
\gamma}}\right)  \notag \\
&=&-\frac{1}{M_{G}G}\sum_{g}\nabla _{\mathbf{\theta }}\mathbf{m}%
_{g}^{\top }\left( \mathbf{\theta }\right) \mathbf{{W}}_{g}^{-1}(\theta, \gamma)\nabla
_{\mathbf{\theta }}\mathbf{m}_{g}\left( \mathbf{\theta }\right) \nonumber \\&&+\frac{1}{M_{G}G}%
\sum_{g}[  \{(\mathbf{y}_g-\mathbf{m}_g(\theta))^{\top} \mathbf{{W}}_g^{-1}(\theta, \gamma) \otimes  I_q\}]\partial \mbox{Vec}(\nabla \mathbf{m}^{\top}_g(\theta))/\partial \theta \nonumber\\&&+\frac{1}{M_{G}G}%
\sum_{g}\{(\mathbf{y}_g-\mathbf{m}_g(\theta))^{\top} \otimes \nabla \mathbf{m}^{\top}_g (\theta)\}\partial \mbox{Vec}(\mathbf{{W}}_g(\theta, \gamma))/\partial \theta \nonumber\\&
\defeq&\mathbf{H}_{G,1}(\theta,{\gamma})+ \mathbf{H}_{G,2}(\theta,{\gamma})+ \mathbf{H}_{G,3}(\theta, {\gamma}),
\end{eqnarray}
where $\mbox{Vec}$ is denoted as the vectorization of a matrix $A$.

\subsection{The first-step estimation of the weight matrix}

In this subsection, we demonstrate one way to find an estimator for $\gamma$ involved in $\mathbf{W}_{g}(\theta, \gamma).$
$\mathbf{W}_{g}(\theta,\gamma)$ can be written as
\begin{equation}
\mathbf{W}_{g}(\theta,\gamma)=\mathbf{V}_{g}(\mathbf{x}_{g};\mathbf{\theta })^{1/2}\mathbf{R%
}_{g}\left( \mathbf{\gamma },{D}_{G}\right) \mathbf{V}_{g}(\mathbf{x}%
_{g};\mathbf{\theta })^{1/2},
\end{equation}%
where $\mathbf{V}_{g}$ is the $L\times L$ diagonal matrix that only contains
variances of $\my_g - \mm_g(\mx_g, \theta^0)$ and $\mathbf{R}_{g}$ is the $L\times L$ correlation matrix for
group $g$.

Let \begin{equation}
\mathbf{V}_{g}(\mathbf{x}_{g};\mathbf{\theta })=\left(
\begin{array}{cccc}
v_{g1} & 0 & \cdots & 0 \\
0 & v_{g2} &  & \vdots \\
\vdots &  & \ddots & 0 \\
0 & ... & 0 & v_{gL}%
\end{array}%
\right) ,
\end{equation}%
where the $l$th element on the diagonal is $v_{gl}=\mathrm{Var}(\mathbf{y}%
_{gl}\mathbf{|x}_{gl})$ in group $g,$ $\mathbf{y}_{gl}$ is the $l$th element
in the vector $\mathbf{y}_{g}$ and $\mathbf{x}_{gl}$ is the $l$th row in $%
\mathbf{x}_{g}$. And
\begin{equation}
\mathbf{R}_{g}\left( \mathbf{\gamma },{D}_{G}\right) =\left(
\begin{array}{cccc}
1 & \pi _{g12} & \cdots & \pi _{g1L} \\
\pi _{g21} & 1 &  & \vdots \\
\vdots &  & \ddots & \pi _{gL-1,L} \\
\pi _{gL1} & ... & \pi _{gL,L-1} & 1%
\end{array}%
\right) .
\end{equation}%
Let $d_{glm}$ be the distance between the $l$th and the $m$th observations
in group $g$. An example of a parametrization of the correlation i.e. the $l, m$th, $l\neq
m,$ element of $\mathbf{R}_{g},$ as in \cite{cressie1992statistics} is%
\begin{equation}
\pi _{glm}=1-b-c\left[ 1-\exp \left( -d_{glm}/\rho \right) \right] ,
\end{equation}%
where the spatial correlation parameters $\mathbf{\gamma =}\left( b,c,\rho
\right) ,$ $b\geq 0,c\geq 0,\rho \geq 0,$ and $b+c\leq 2.\footnote{%
See \cite{cressie1992statistics} p.61 for more examples.}$Set $b=c=1$ without loss of
generality. Then
\begin{equation}
\pi _{glm}=\left \{
\begin{array}{c}
1\text{ \  \  \  \  \  \  \  \ if }l=m, \\
\exp \left( -d_{glm}/\rho \right) \text{ \  \ otherwise.}%
\end{array}%
\right.
\end{equation}%
Although the above specification does not represent all the possibilities,
it at least provides a way of how to parameterize the spatial correlation,
and therefore the basis for testing spatial correlation.

The following provides a way to estimate $\mathbf{\gamma }$. Let $\mathbf{%
\check{\theta}}$ be the first-step \textrm{PQMLE} estimator. $\check{u}%
_{i}=y_{i}-m_{i}\left( x_{i};\mathbf{\check{\theta}}\right) $ are the
first-step residuals. $\check{v}_{i}=v\left( m_{i}\left( \mathbf{x}_{i};%
\mathbf{\check{\theta}}\right) \right) $ is the fitted variance of
individual $i$ corresponding to the chosen \textrm{LEF} density. Let $\check{r%
}_{i}=\check{u}_{i}/\sqrt{\check{v}_{i}}$ be the standardized residual. Let $%
\mathbf{\check{r}}_{g}=$ $\left( \check{r}_{g1},\check{r}_{g2},...,\check{r}%
_{gL}\right) ^{\top }.$ Then $\mathbf{\mathbf{\check{r}}_{g}\mathbf{\check{r%
}}_{g}}^{\top }$ is the estimated sample correlation matrix for group $g$.
Let $\mathbf{e}_{g}(\check{\theta})$ be a vector containing $L(L-1)/2$ different elements of
the lower (or upper) triangle of $\mathbf{\mathbf{\check{r}}_{g}\mathbf{%
\check{r}}_{g}}^{\top },$ excluding the diagonal elements. Let $\mathbf{z}%
_{g}(\gamma)$ be the vector containing the elements in $\mathbf{R}_{g}$
corresponding to the same entries of elements in $\mathbf{\mathbf{\check{r}}%
_{g}\mathbf{\check{r}}_{g}}^{\top }$. We can follow \cite{prentice1988applications}, who provides one way to
find a consistent estimator for $\mathbf{\gamma }$ by solving:%
\begin{equation}\label{pregamma}
\mathbf{\hat{\gamma}}=\arg \min_{\gamma \in \Gamma} \sum_{g}(\mathbf{e}_{g}(\check{\theta})-\mathbf{z}%
_{g}(\gamma))^{\top }(\mathbf{e}_{g}(\check{\theta})-\mathbf{z}_{g}(\gamma)).
\end{equation}

\section{Estimating nonlinear models with spatial error: two examples} \label{sec3}

The setup of nonlinear models with spatial data varies with different
models. For each model, we need to incorporate the spatial correlated term
in an appropriate way. In this Section, we will demonstrate how we
incorporate the spatial correlated error term in two types of discrete data
and how to use a \textrm{GEE} procedure to estimate the nonlinear models.
The first example is for count data and the second one is for binary response data.

\subsection{Example 1 \ Count data with a multiplicative spatial error}

A count variable is a variable that takes on nonnegative integer values,
such as the number of patents applied for by a firm during a year. \cite{bloom2013identifying} studies spillover effects of R\&D between firms in terms of firm
patents. Other examples include the number of times someone being arrested
during a given year. Count data examples with upper bound include the number
of children in a family who are high school graduates, in which the upper
bound is number of children in the family (\cite{wooldridge2010econometric}).

\subsubsection{Poisson model}

We first model the count data with a conditional Poisson density, $f\left( y|%
\mathbf{x}\right) =\exp \left[ -\mu \right] \mu ^{y}/y!,$ where $y!=1\cdot
2\cdot ...\cdot \left( y-1\right) \cdot y$ and $0!=1.$ $\mu $ is the
conditional mean of $y.$ The Poisson QMLE requires us only to correctly specify
the conditional mean. A default assumption for the Poisson distribution is
that the mean is equal to the variance. Note that even if $y_{i}$ does not
follow the Poisson distribution, the Q\textrm{MLE} approach will give a
consistent estimator if you use the Poisson density function and a correctly
specified conditional mean (\cite{gourieroux1984pseudo}). Moreover, $%
y_{i}$ even need not to be a count variable. The most common mean function
in applications is the exponential form:
\begin{equation}
\E\left( y_{i}|\mathbf{x}_{i}\right) =\exp \left( \mathbf{x}_{i}%
\mathbf{\beta }_{0}\right) .
\end{equation}%
When spatial correlation exists, we can characterize count data model with a
multiplicative spatial error. \cite{silva2006log} use the Poisson
pseudo-maximum-likelihood (\textrm{PPML}), which is the Poisson QMLE in this
paper, to estimate the gravity model for trade. They argue that constant
elasticity models should be estimated in their multiplicative form, because
using a log linear model can cause bias in coefficient estimates under
heteroskedasticity. Now we further consider the Poisson regression model
with spatial correlation in the multiplicative error,
\begin{equation}
\E\left( y_{i}|\mathbf{x}_{i},v_{i}\right) =v_{i}\exp \left( \mathbf{%
x}_{i}\mathbf{\beta }_{0}\right) ,
\end{equation}%
where $v_{i}$ is the multiplicative spatial error term. Let $\mathbf{v}$
equal $\left( v_{1},v_{2},...,v_{n}\right) ^{\top }.$  (Note that for this example we treat location $i$ as an one dimensional object.) This model is
characterized by the following assumptions:

(1) $\{(\mathbf{x}_{i},v_{i}),i=1,2,...,n\}$ is a mixing sequence on the
sampling space $D_n$, with mixing coefficient $\alpha $.

(2) $\E\left( y_{i}|\mathbf{x}_{i},v_{i}\right) =v_{i}\exp \left(
\mathbf{x}_{i}\mathbf{\beta }_{0}\right) .$

(3) $y_{i},y_{j}$ are independent conditional on $\mathbf{x}_{i},\mathbf{x}%
_{j},v_{i},v_{j},i\neq j.$

(4) $v_{i}$ has a conditional multivariate distribution, $\E\left(
v_{i}|\mathbf{x}_{i}\right) =1$. $\mathrm{Var}\left( v_{i}|\mathbf{x}%
_{i}\right) =\tau ^{2},$ $\mathrm{Cov}\left( v_{i},v_{j}|\mathbf{x}_{i},%
\mathbf{x}_{j}\right) =\tau ^{2}\cdot c\left( d_{ij},\rho \right) ,$ where $%
c\left( d_{ij},\rho \right) $ is the correlation function of $v_{i}$ and $%
v_{j}.$

Under the above assumptions, and again conditional on $D_n$ is
suppressed, we can integrate out $v_{i}$ by using the law of iterated
expectations.%
\begin{equation}
\E\left( y_{i}|\mathbf{x}_{i},D_n\right) =\E%
\left( \E\left( y_{i}|\mathbf{x}_{i},v_{i}\right) |\mathbf{x}_{i},%
D_n\right) =\exp \left( \mathbf{x}_{i}\mathbf{\beta }_{0}\right) .
\label{pmean}
\end{equation}%
If $x_{j}$ is continuous, the partial effects on $\E\left( y_{i}|%
\mathbf{x}_{i},D_n\right) $ is $\exp \left( \mathbf{x}_{i}\mathbf{%
\beta }_{0}\right) \beta _{j}.$ If $x_{j}$ is discrete the partial effects
is the change in $\E\left( y_{i}|\mathbf{x}_{i},\mathbf{D}%
_{n}\right) $ when, say, $x_{K}$ goes from $a_{K}$ to $a_{K}+1$ which is
\begin{equation}
\exp \left( \beta _{1}+x_{2}\beta _{2}+...+\beta _{K}\left( a_{K}+1\right)
\right) -\exp \left( \beta _{1}+x_{2}\beta _{2}+...+\beta _{K}a_{K}\right) .
\end{equation}%
The pooled QMLE gives a consistent estimator for the mean parameters, which
solves:
\begin{equation}
\mathbf{\hat{\beta}}_{PQMLE}=\arg \max_{\theta \in \Theta
}\sum_{i=1}^{n}l_{i}\left( \beta \right) =\sum_{i=1}^{n}y_{i}\mathbf{x}_{i}%
\mathbf{\beta }-\sum_{i=1}^{n}\exp \left( \mathbf{x}_{i}\mathbf{\beta }%
\right) -\sum_{i=1}^{n}\log \left( y_{i}!\right) .
\end{equation}%
Its score function is
\begin{equation}
\sum_{i=1}^{n}\mathbf{x}_{i}^{\top }\left[ y_{i}-\exp \left( \mathbf{x}_{i}%
\mathbf{\check{\beta}}_{\mathrm{QMLE}}\right) \right] =\mathbf{0}.
\end{equation}%
Since this estimator does not account for any heteroskedasticity or spatial
correlation, a robust estimator for the asymptotic variance of partial \textrm{QMLE} estimator is
provided as follows,
\begin{eqnarray}
\widehat{\mathrm{Avar}}\left( \mathbf{\check{\beta}}_{\mathrm{QMLE}}\right) &=&%
\left[ \sum_{i=1}^{n}\exp \left( -\mathbf{x}_{i}\mathbf{\check{\beta}}_{%
\mathrm{QMLE}}\right) \mathbf{x}_{i}^{\top }\mathbf{x}_{i}\right]
^{-1} \\
&&\sum_{i=1}^{n}\sum_{j=1}^{n}k\left( d_{ij}\right) \mathbf{x}_{i}^{\top }\hat{u}_{i}\hat{u}_{j}\mathbf{x}_{j}\left[ \sum_{i=1}^{n}\exp\left( -\mathbf{x}_{i}\mathbf{\check{\beta}}_{\mathrm{QMLE}}\right) \mathbf{x}_{i}^{\top }\mathbf{x}_{i}\right] ^{-1}, \notag
\end{eqnarray}
where $k\left( d_{ij}\right) $ is a kernel function depending on the
distance between observations $i$ and $j$.

Moreover, a very specific nature of the Poisson distribution is that we can
write down the conditional variances and covariances of $y$:%
\begin{equation}
\mathrm{Var}\left( y_{i}|\mathbf{x}_{i},D_n\right) =\exp \left(
\mathbf{x}_{i}\mathbf{\beta }_{0}\right) +\exp \left( 2\mathbf{x}_{i}\mathbf{%
\beta }_{0}\right) \cdot \tau ^{2}.  \label{geevar}
\end{equation}%
The conditional variance of $y_{i}$ given $\mathbf{x}_{i}$ is a function of
both the level and the quadratic of the conditional mean. The traditional
Poisson variance assumption is that the conditional variance should equal the conditional mean. That is, $\mathrm{Var}\left( y_{i}|\mathbf{x}_{i}\right)=\exp \left( \mathbf{x}_{i}\mathbf{\beta }_{0}\right) .$
The Poisson \textrm{GLM} variance assumption is $\mathrm{Var}\left( y_{i}|%
\mathbf{x}_{i}\right) =\sigma ^{2}\exp \left( \mathbf{x}_{i}\mathbf{\beta }%
_{0}\right) $ with an overdispersion or underdispersion parameter $\sigma
^{2}$, which is a constant. Obviously, there is over-dispersion in (\ref%
{geevar}) since $\exp \left( 2\mathbf{x}_{i}\mathbf{\beta }_{0}\right) \cdot
\tau ^{2}\geq 0,$ and the over-dispersion parameter is $1+\exp \left(
\mathbf{x}_{i}\mathbf{\beta }_{0}\right) \cdot \tau ^{2}$, which is changing
with $\mathbf{x}_{i}$. This does not coincide with Poisson variance assumption
and the \textrm{GLM} variance assumption. What is more, the conditional
covariances can be written in the following form,
\begin{equation}
\mathrm{Cov}\left( y_{i},y_{j}|\mathbf{x}_{i},\mathbf{x}_{j},{D}%
_{n}\right) =\exp \left( \mathbf{x}_{i}\mathbf{\beta }_{0}\right) \exp
\left( \mathbf{x}_{j}\mathbf{\beta }_{0}\right) \cdot \tau ^{2}\cdot c\left(
d_{ij},\rho \right) .  \label{geecov}
\end{equation}%
In the group level notation,
\begin{equation}
\E\left( \mathbf{y}_{g}|\mathbf{x}_{g},{D}_{G}\right) =\exp
\left( \mx_{g}\mathbf{\beta }_{0}\right) .
\end{equation}%
Let $\mathbf{W}_{g}$ be the variance-covariance matrix for group $g$ evaluated at the true value $\beta_0,\rho_0$. The
variance of the $l$th element in group $g$ is
\begin{equation}
v_{gl}=\exp \left( \mathbf{x}_{gl}\mathbf{\beta }_{0}\right) \left( 1+\exp
\left( \mathbf{x}_{gl}\mathbf{\beta }_{0}\right) \cdot \tau ^{2}\right) ,
\label{countvar}
\end{equation}%
and the covariance of the $l$th and $m$th elements in group $g$ is
\begin{equation}
r_{glm}=\exp \left( \mathbf{x}_{gl}\mathbf{\beta }_{0}\right) \exp \left(
\mathbf{x}_{gm}\mathbf{\beta }_{0}\right) \cdot \tau ^{2}\cdot c\left(
d_{glm},\rho \right) .  \label{countcov}
\end{equation}%
Here $\mathbf{\gamma =}\left( \tau ^{2},\rho \right) ^{\top }$ and $%
\mathbf{\hat{\gamma}=}\left( \hat{\tau}^{2},\hat{\rho}\right) ^{\top }$ is
an estimator for $\mathbf{\gamma }$. Let $\mathbf{\check{\beta}}_{\mathrm{%
PQMLE}}$ be the partial QMLE estimator in the first step. Then the elements
in $\mathbf{W}_{g}$ can be estimated as

\begin{equation}
\hat{v}_{gl}=\exp \left( \mathbf{x}_{gl}\mathbf{\check{\beta}}_{\mathrm{PQMLE%
}}\right) +\exp \left( 2\mathbf{x}_{gl}\mathbf{\check{\beta}}_{\mathrm{PQMLE}%
}\right) \cdot \hat{\tau}^{2},
\end{equation}%
\begin{equation}
\hat{r}_{glm}=\exp \left( \mathbf{x}_{gl}\mathbf{\check{\beta}}_{\mathrm{%
PQMLE}}\right) \exp \left( \mathbf{x}_{gm}\mathbf{\check{\beta}}_{\mathrm{%
PQMLE}}\right) \cdot \hat{\tau}^{2}\cdot c\left( d_{ij},\hat{\rho}\right) .
\end{equation}%
Based on the conditional distribution, the first order conditions for
\textrm{GEE }is:%
\begin{equation}
\sum_{g}\mx_{g}^{\top }\mathbf{W}_{g}^{-1}\left( \mathbf{\hat{%
\gamma},\hat{\theta}}\right) \left[ \mathbf{y}_{g}-\exp \left( \mx_{g}\mathbf{\hat{%
\beta}}_{\mathrm{GEE}}\right) \right] =0.
\end{equation}%
$\mathbf{\hat{\beta}}_{\mathrm{GEE}}$ is consistent and follows a normal
distribution asymptotically by Theorem \ref{th:consistency} and \ref{th:normality}. We will brief $\mathbf{W}_{g}^{-1}\left( \mathbf{\hat{%
\gamma},\hat{\theta}}\right)$  as $\hat{\mathbf{W}}_{g}^{-1}$ in the following text. The variance
estimator for the asymptotic variance that is robust to misspecification of spatial correlation is:
\begin{eqnarray}
\widehat{\mathrm{Avar}}\left( \mathbf{\hat{\beta}}_{\mathrm{GEE}}\right)
&=& \left( \sum_{g}\exp \left( 2\mx_g^{\top }\mathbf{\hat{%
\beta}}_{\mathrm{GEE}}\right) _{g}\mx_g^{\top }\mathbf{\hat{W}}%
_{g}^{-1}\mx_g\right) ^{-1} \\
&& \left( \sum_{g}\sum_{h (\neq g) }k(d_{gh})\exp \left( \mathbf{x}%
_{g}^{\top }\mathbf{\hat{\beta}}_{\mathrm{GEE}}+\mathbf{x}_{h}\mathbf{\hat{%
\beta}}_{\mathrm{GEE}}\right) \mx_g^{\top }\mathbf{\hat{W}}%
_{g}^{-1}\mathbf{\hat{u}}_{g}\mathbf{\hat{u}}_{h}^{\top }\mathbf{\hat{W}}%
_{h}^{-1}\mathbf{x}_{h}^{\top }\right)  \notag \\
&&\left( \sum_{g}\exp \left( 2\mx_g^{\top }\mathbf{\hat{%
\beta}}_{\mathrm{GEE}}\right) \mx_g^{\top }\mathbf{\hat{W}}%
_{g}^{-1}\mx_g\right) ^{-1}  \notag
\end{eqnarray}%
where $k(d_{gh})$ is a kernel function depending on the distances between
groups. The distances could be the smallest distance between two
observations belonging to different groups.

The pivotal parameters, $\tau ^{2}$ and $\rho ,$ can be estimated using the
Poisson QMLE residuals. Let $\check{u}_{i}^{2}=\left[ y_{i}-\exp \left(
\mathbf{x}_{i}\mathbf{\check{\beta}}_{\mathrm{QMLE}}\right) \right] ^{2}$ be
the squared residuals from the Poisson QMLE. Based on equation (\ref%
{countvar}), $\tau ^{2}$ can be estimated as the coefficient by regressing $%
\check{u}_{i}^{2}-\exp \left( \mathbf{x}_{i}\mathbf{\check{\beta}}_{\mathrm{%
QMLE}}\right) $ on $\exp \left( 2\mathbf{x}_{i}\mathbf{\check{\beta}}_{%
\mathrm{QMLE}}\right) .$ The situation to estimate $\rho $ depends on the
specific form of $c\left( d_{ij},\rho \right) $. We would like to assume a
structure, though it might be wrong, to approximate the true covariance. For
example, suppose the covariance structure of $e_{i}$ and $e_{j}$ is $\exp
\left( \frac{\rho }{d_{ij}}\right) -1,$ and the correlation structure is $%
c\left( d_{ij},\rho \right) =\frac{\exp \left( \frac{\rho }{d_{ij}}\right) -1%
}{\mathrm{e}-1},$ then an estimator for $\rho $ is:

\begin{equation}
\hat{\rho}=\mbox{argmin}_{\rho}\sum_{i=1}^{n}\sum_{j\neq i}^{n}\left \{ \frac{\check{u}_{i}%
\check{u}_{j}}{\exp \left( \mathbf{x}_{i}\mathbf{\check{\beta}}\right) \exp
\left( \mathbf{x}_{j}\mathbf{\check{\beta}}\right) }-\left[ \exp \left(
\frac{\rho }{d_{ij}}\right) -1\right] \right \} ^{2}.
\end{equation}%
Then $\mathbf{\hat{W}}_{g}$ is obtained by plugging $\hat{\tau}^{2}$ and $%
\hat{\rho}$ back in the variance-covariance matrix. We can also directly
calculate $\hat{\rho}$ as
\begin{equation}
\hat{\rho}=\frac{1}{n\cdot \left( n-1\right) }\sum_{i=1}^{n}\sum_{j\neq
i}^{n}\left[ \log \left( \frac{\check{u}_{i}\check{u}_{j}}{\exp \left(
\mathbf{x}_{i}\mathbf{\check{\beta}}\right) \exp \left( \mathbf{x}_{j}%
\mathbf{\check{\beta}}\right) }+1\right) \cdot d_{ij}\right] .
\end{equation}

\subsubsection{The negative binomial model}

Since the conditional variances and covariances can be written in a specific
form, we would consider NegBin II model of \cite{cameron1986econometric} as a
more appropriate model. The NegBin II model can be derived from a model of
multiplicative error in a Poisson model. With an exponential mean, $y_{i}|%
\mathbf{x}_{i},v_{i},D_n\sim $Poisson$\left[ v_{i}\exp \left(
\mathbf{x}_{i}\mathbf{\beta }_{0}\right) \right] $. Under the above
assumptions for Poisson distribution, with the conditional mean (\ref{pmean}%
) and conditional variance (\ref{geevar}), $y_{i}|\mathbf{x}_{i}$ is shown
to follow a negative binomial II distribution. It implies overdispersion,
but where the amount of overdispersion increases with the conditional mean,
\begin{equation}
\mathrm{Var}\left( y_{i}|\mathbf{x}_{i},D_n\right) =\exp \left(
\mathbf{x}_{i}\mathbf{\beta }_{0}\right) \left( 1+\exp \left( \mathbf{x}_{i}%
\mathbf{\beta }_{0}\right) \cdot \tau ^{2}\right) .
\end{equation}%
Now the log-likelihood function for observation $i$ is%
\begin{eqnarray}
l_{i}(\beta, \tau) &=&\left( \tau ^{2}\right) ^{-2}\log \left[ \frac{\left( \tau
^{2}\right) ^{-2}}{\left( \tau ^{2}\right) ^{-2}+\exp \left( \mathbf{x}_{i}%
\mathbf{\beta }\right) }\right] +y_{i}\log \left[ \frac{\exp \left( \mathbf{x%
}_{i}\mathbf{\beta }\right) }{\left( \tau ^{2}\right) ^{-2}+\exp \left(
\mathbf{x}_{i}\mathbf{\beta }\right) }\right]  \label{nblikelihood} \\
&&+\log \left[ \Gamma \left( y_{i}+\left( \tau ^{2}\right) ^{-2}\right)
/\Gamma \left( \left( \tau ^{2}\right) ^{-2}\right) \right] ,  \notag
\end{eqnarray}%
where $\Gamma \left( \cdot \right) $ is the gamma function defined for $r>0$
by $\Gamma \left( r\right) =\int_{0}^{\infty }z^{r-1}\exp \left(
-z\right) dz$. For fixed $\tau ^{2}$, the log likelihood equation in (\ref%
{nblikelihood}) is in the exponential family; see \cite{gourieroux1984pseudo}. Thus the
negative binomial QMLE using (\ref{nblikelihood}) is consistent under
conditional mean assumption only, which is the same as the Poisson QMLE.
Since the negative binomial II likelihood captures the nature of the variance
function, it should deliver more efficient estimation when the data generating process is correctly specified, although the spatial correlation is not accounted. Again, we can use a GEE working correlation
matrix to account for the spatial correlation.

\subsection{Example 2. Binary response data with spatial correlation in the
latent error}

The Probit model is one of the popular binary response models. The dependent
variable $y$ has conditional Bernoulli distribution and takes on the values
zero and one, which indicates whether or not a certain event has occurred.
For example, $y=1$ if a firm adopts a new technology, and $y=0$ otherwise.
The value of the latent variable $y^{\ast }$ determines the outcome of $y$.

Assume the Probit model is%
\begin{eqnarray}
y_{i} &=&1\left[ y_{i}^{\ast }>0\right] , \\
y_{i}^{\ast } &=&\mathbf{x}_{i}\mathbf{\beta }+e_{i}.
\end{eqnarray}%
We do not observe $y_{i}^{\ast }$; we only observe $y_{i}.$ Let $\Phi \left(
\cdot \right) $ be the standard normal cumulative density function (\textrm{%
CDF}), and $\phi $ be the standard normal probability density function (%
\textrm{PDF}). Assume that the mean function $m_{i}\left( \mathbf{x}_{i};%
\mathbf{\beta }\right) \equiv $ $\E\left( y_{i}|\mathbf{x}_{i},%
D_n\right) =\Phi \left( \mathbf{x}_{i}\mathbf{\beta }\right) $ is
correctly specified. $e$ is the spatial correlated latent error. Let $%
\mathbf{e}=\left( e_{1},e_{2},...,e_{n}\right) ^{\top }$. For example,
\cite{pinkse1998contracting} use the following assumption of $\mathbf{e}$:
\begin{equation}
\mathbf{e}=\rho W\mathbf{e}+\mathbf{\varepsilon ,}  \label{spatialcorr}
\end{equation}%
where $\mathbf{\varepsilon =}\left( \varepsilon _{1},\varepsilon
_{2},...,\varepsilon _{n}\right) $ which has a standard normal distribution.
$W$ is a $n\times n$ weight matrix with zeroes on the diagonal and inverse
of distances off diagonal. $\rho $ is a correlation parameter. We can see $e$
can be written as a function of $\varepsilon ,$%
\begin{equation}
\mathbf{e}=\left( I-\rho W\right) ^{-1}\mathbf{\varepsilon }.
\end{equation}%
Thus the conditional expectation of $\me$ is zero. The variance covariance
matrix of $\me$ is
\begin{equation}
\mathrm{Var}\left( \mathbf{e}|\mathbf{x},D_n\right) =\left( I-\rho
W\right) ^{-1}\left( I-\rho W\right) ^{-1\top }.  \label{probitvar}
\end{equation}%
If we assume that $e|x$ has a multivariate normal distribution with mean zero
and variance matrix specified in (\ref{probitvar}). Thus a much simpler
specification is to directly model $e|x$ as a multivariate distribution.
Different from the usual multivariate distribution\footnote{%
A multivariate normal distribution usually specifies the mean vector and
correlation matrix. The correlations do not depend on the pairwise distance
between two variables.}, the covariances of $e$ should depend on the
pairwise distances $d_{ij}$. We also let the covariances depend on a
parameter $\rho $. The above equation can be written in a conditional mean
form:%
\begin{equation}
\E\left( y_{i}|\mathbf{x}_{i},D_n\right) =\Phi \left(
\mathbf{x}_{i}\mathbf{\beta }\right) .  \label{probitmean}
\end{equation}%
It is very natural to write the variance function for a Bernoulli
distribution,
\begin{equation}
\mathrm{Var}\left( y_{i}|\mathbf{x}_{i},D_n\right) =\Phi \left(
\mathbf{x}_{i}\mathbf{\beta }\right) \left[ 1-\Phi \left( \mathbf{x}_{i}%
\mathbf{\beta }\right) \right] .
\end{equation}%
We are interested in the partial effects of $x$ to $y$. For a continuous $%
x_{K}$ the partial effect is
\begin{equation}
\frac{\partial \E\left( y_{i}|\mathbf{x}_{i},D_n\right) }{%
\partial x_{K}}=\Phi \left( \mathbf{x}_{i}\mathbf{\beta }\right) \beta _{K}.
\end{equation}%
For a discrete $x_{K}$, the partial effects when $x_{K}$ changes from $a_{K}$
to $a_{K}+1$ is
\begin{equation}
\Phi \left( \beta _{1}+x_{2}\beta _{2}+...+\beta _{K}\left( a_{K}+1\right)
\right) -\Phi \left( \beta _{1}+x_{2}\beta _{2}+...+\beta _{K}a_{K}\right) .
\end{equation}%
A simple one-step estimation is the pooled Bernoulli quasi-MLE (\textrm{QMLE}%
), which is obtained by maximizing the pooled Probit log-likelihood. The log
likelihood function for each observation is%
\begin{equation}
l_{i}\left( \mathbf{\beta }\right) \mathbf{=}y_{i}\log \Phi \left( \mathbf{x}%
_{i}\mathbf{\beta }\right) \mathbf{+}\left( 1-y_{i}\right) \log \left[
1-\Phi \left( \mathbf{x}_{i}\mathbf{\beta }\right) \right] .
\end{equation}
Let $\check{u}_{i}=y_{i}-\Phi \left( \mathbf{x}_{i}\mathbf{\check{\beta}}%
\right) ,i=1,2,...,n$ be the residuals from the partial \textrm{QMLE}
estimation. At this stage, a robust estimator for the asymptotic variance of $%
\mathbf{\check{\beta}}_{\mathrm{PQMLE}}$ can be computed as follows:%
\begin{eqnarray}
\widehat{\mathrm{Avar}}\left( \mathbf{\check{\beta}}_{\mathrm{PQMLE}}\right)
&=&\left( \sum_{i=1}^{n}\frac{\phi ^{2}\left( \mathbf{x}_{i}\mathbf{\check{%
\beta}}_{\mathrm{PQMLE}}\right) \mathbf{x}_{i}^{\top }\mathbf{x}_{i}}{\Phi
\left( \mathbf{x}_{i}\mathbf{\check{\beta}}\right) \left[ 1-\Phi \left(
\mathbf{x}_{i}\mathbf{\check{\beta}}_{\mathrm{PQMLE}}\right) \right] }%
\right) ^{-1} \\
&&\left( \sum_{i=1}^{n}\sum_{j\neq i}^{n}k\left( d_{ij}\right) \frac{\phi
\left( \mathbf{x}_{i}\mathbf{\check{\beta}}_{\mathrm{PQMLE}}\right) \phi
\left( \mathbf{x}_{j}\mathbf{\check{\beta}}_{\mathrm{PQMLE}}\right) \mathbf{x%
}_{i}^{\top }\check{u}_{i}\check{u}_{j}\mathbf{x}_{j}}{\Phi \left( \mathbf{%
x}_{i}\mathbf{\check{\beta}}_{\mathrm{PQMLE}}\right) \left[ 1-\Phi \left(
\mathbf{x}_{i}\mathbf{\check{\beta}}_{\mathrm{PQMLE}}\right) \right] }\right)
\notag \\
&&\left( \sum_{i=1}^{n}\frac{\phi ^{2}\left( \mathbf{x}_{i}\mathbf{\check{%
\beta}}_{\mathrm{PQMLE}}\right) \mathbf{x}_{i}^{\top }\mathbf{x}_{i}}{\Phi
\left( \mathbf{x}_{i}\mathbf{\check{\beta}}_{\mathrm{PQMLE}}\right) \left[
1-\Phi \left( \mathbf{x}_{i}\mathbf{\check{\beta}}_{\mathrm{PQMLE}}\right) %
\right] }\right) ^{-1},  \notag
\end{eqnarray}%
where $k\left( d_{ij}\right) $ is the kernel weight function that depends on
pairwise distances. This partial QMLE and its robust variance-covariance
estimator provides a legitimate way of the estimation of the spatial Probit
model.

We use partial QMLE as a first-step estimator. An estimator for the working
variance matrix for each group is%
\begin{equation}
\check{v}_{gl}=\Phi \left( \mathbf{x}_{gl}\mathbf{\check{\beta}}_{\mathrm{%
PQMLE}}\right) \left[ 1-\Phi \left( \mathbf{x}_{gl}\mathbf{\check{\beta}}_{%
\mathrm{PQMLE}}\right) \right] .
\end{equation}%
And assume the working correlation function for $l$th and $m$th elements in
group $g$ is
\begin{equation}
r_{glm}=\mathbf{C}\left( d_{glm},\rho \right) .
\end{equation}%
For example, suppose that %
\begin{equation}
\mathbf{C}\left( d_{glm},\rho \right) =\frac{\rho }{d_{glm}}\text{ or }\exp
\left( -\frac{d_{glm}}{\rho }\right) .
\end{equation}%
Let $\check{u}_{i}$ be the partial QMLE residual and $\hat{r}_{i}=\check{u}%
_{i}/\sqrt{\check{v}_{i}}$, for $i=1,2,...,n,$ be the standardized
residuals. $\mathbf{\hat{C}}_{ij}$ equals the sample correlation of $\check{u%
}_{i}/\sqrt{\check{v}_{i}}$ and $\check{u}_{j}/\sqrt{\check{v}_{j}}$. Using
the correlations within groups, one estimator of $\rho $ is
\begin{equation}
\hat{\rho}=\arg \min_{\rho} \sum_{g}\sum_{l=1}^L\sum_{m < l}\left[
\hat{r}_{gl}\hat{r}_{gm}-C\left( d_{glm},\rho \right) \right] ^{2},
\end{equation}%
for $l<m.$

The second-step \textrm{GEE} estimator for $\mathbf{\beta }$ is
\begin{equation}
\mathbf{\hat{\beta}}_{\mathrm{GEE}}=\arg \min_{\beta }\sum_{g}\left(
\mathbf{y}_{g}-\Phi \left( \mx_g\mathbf{\beta }\right) \right)
^{\top }\mathbf{\hat{W}}_{g}^{-1}\left( \mathbf{y}_{g}-\Phi \left( \mathbf{%
x}_{g}\mathbf{\beta }\right) \right) .
\end{equation}%
The first order condition is
\begin{equation}
\sum_{g}\phi \left( \mx_g\mathbf{\hat{\beta}}_{\mathrm{GEE}%
}\right) ^{\top }\mathbf{\hat{W}}_{g}^{-1}\left( \mathbf{y}_{g}-\Phi
\left( \mx_g\mathbf{\hat{\beta}}_{\mathrm{GEE}}\right) \right) =%
\mathbf{0.}
\end{equation}%
$\mathbf{\hat{\beta}}_{\mathrm{GEE}}$ is consistent and follows a normal
distribution asymptotically by Theorem \ref{th:normality}. $%
\mathbf{\hat{\beta}}$ is consistent even for misspecified spatial
correlation structure $\mathbf{\hat{W}}_{g}$. The asymptotic variance estimator that is robust
to misspecification of spatial correlation is:
\begin{eqnarray}
\widehat{\mathrm{Avar}}\left( \mathbf{\hat{\beta}}_{\mathrm{GEE}}\right)&=&\left(
\sum_{g}\phi ^{2}\left( \mx_g\mathbf{\hat{\beta}}_{\mathrm{GEE%
}}\right) \mx_g^{\top }\mathbf{\hat{W}}_{g}^{-1}\mathbf{x}%
_{g}\right) ^{-1} \\
&&\left( \sum_{g}\sum_{h (\neq g) }k(d_{gh})\phi \left( \mathbf{x}_{g}
\mathbf{\hat{\beta}}_{\mathrm{GEE}}\right) \phi \left( \mathbf{x}_{h}\mathbf{%
\hat{\beta}}_{\mathrm{GEE}}\right) \mx_g^{\top }\mathbf{\hat{W}}%
_{g}^{-1}\mathbf{\hat{u}}_{g}\mathbf{\hat{u}}_{h}^{\top }\mathbf{\hat{W}}%
_{h}^{-1}\mathbf{x}_{h}\right)  \notag \\
&&\left( \sum_{g}\phi ^{2}\left( \mx_g\mathbf{\hat{\beta}}_{%
\mathrm{GEE}}\right) \mx_g^{\top }\mathbf{\hat{W}}_{g}^{-1}%
\mx_g\right) ^{-1},  \notag
\end{eqnarray}%
where $k(d_{gh})$ is a kernel function which depends on the distances
between groups.

An alternative approach is to specify the specific distributions of the
multivariate normal distribution of the latent error, and then find the
estimator for the spatial correlation parameter for the latent error within
a MLE framework. For example, see \cite{wang2013partial}.

\section{Theorems}\label{sectheorem}
In this section, we provide the assumptions and results on the theoretical properties our GEE estimation.
\subsection{Consistency and Normality}
\begin{itemize}
\item[A.2)] $\{y_{i}\}$ is $L_4-$ uniformly NED on the $\alpha-$ mixing random field $\vps = \{\vps_{i}, i\in D_n\},$ where $\vps_i = (x_i, \epsilon_i)$($\epsilon_i$s are some underlying innovation processes). With the $\alpha-$ mixing coefficient $\overline{\alpha}(u,v,r) \leq (u+v)^\tau \hat{\alpha}(r),$ and $\hat{\alpha}(r) \to 0$ as $r\to \infty.$ Assume that $\sum^{\infty}_{r = 1} r^{d-1}\hat{\alpha}(r)< \infty .$ The NED  constant is $d_{n,i}$, ($\sup_{n, i\in T_n} d_{n,i} < \infty$) and the NED coefficient is $\psi(s)$ with $\psi(s) \to 0$, where recall that $L$ is the group size, and $\sum^{\infty}_{r=0} r^{d-1} \psi(r) \to 0$.

    \textbf{Remark:}
    See section \ref{ass2} for a detailed verification of the special cases.
   It should be noted that by the
Lyapunov inequality, if $\{y_{i}\}$ is $L_k$-NED, then it is also $L_l$-NED with the same coefficients
$d_{n,i}$ and $\psi(s)$ for any $l \leq k$.

\item[A.3)] The parameter space $\mathbf{\Theta }\times \mathbf{\Gamma}$ is a compact
subset on $\mathcal{R}^{p+q}$ with metric $\nu(.,.)$.

\item[A.4)]$q_{g}\left( \mathbf{\theta,\gamma }\right) $, ($s_g(\mathbf{\theta,\gamma})$), ($h_g(\mathbf{\theta,\gamma})$) are $\mathbf{R}^{p_w}\times \Theta \times \Gamma \to \mathbf{R}^{1} (\mathbf{R}^{p}), (\mathbf{R}^{p^2}) $ measurable for each $\theta \in \Theta, \gamma \in \Gamma$, and Lipschitz continuous on $\mathbf{\Theta }\times \Gamma$.

\item[A.5)] $\E \sup_{\theta \in \Theta} |m_{g,i}|^{r}\leq C_1$, $\E \ssup |w_{g,i,j}|^r \leq  C_2$, $\E |y_{g,i}|^{r} \leq C_3$\\ $\E \sup_{\theta \in \Theta} |\nabla_{\theta}m_{g,i}|^{r}\leq C_4$,  where $C_1, C_2, C_3, C_4$ are constants, where $w_{g,i,j}, y_{g,i}, m_{g,i}$ is the elementwise component for $\mathbf{W}_g^{-1}(\theta,\gamma)$, $\mathbf{y}_g$, $\mm_g(\theta,\gamma).$ $r > 4p'' \vee 4p'.$ $m_{g,i}, w_{g,i,j}$ are continuously differentiable up to the third order derivatives, and its $r$th moment (the supreme over the parameter space) is bounded up to the second order derivatives.
    Define $d_g = \max_{i \in B_g} d_{n,i}$, $M_G \defeq \max_g d_g \vee c_{g,q} \vee c_{g,s}\vee c_{g,h}.$ Also assume that $\sup_G\sup_g (c_{g,q} \vee c_{g,s}\vee c_{g,h})/ d_g \leq C_5$, where $C_5$ is a constant.

    \textbf{Remark:}
Condition A.5) guarantees that there exists non random positive constants such that $c_{g,q},c_{g,s},c_{g,h}, g \in D_G, n\geq 1$ such that  $\E |q_{g}/c_{g,q}|^{p''} <
     \infty$, $\E |s_{g}/c_{g,s}|_2^{p''}< \infty$, $\E |h_{g}/c_{g,h}|_1^{p''} <\infty$ .

\end{itemize}

From now on we work with group level asymptotics.
Define the field $\tilde{\vps} = \{\vps_g: g \in 1, \cdots, G\}$ with grouped observations.
First of all suppose that $D_n$ is divided by $G$ blocks with $\cup^G_1 B_g = D_n \subset T_n$,  and the group level lattice is denoted as $D_G$.
Define the distance between two groups $g,h$ as $\rho(g,h) = \min_{i \in B_g, j \in B_h} \rho(i,j).$
 And the $\alpha-$ mixing coefficient between two union of groups for $U = \{g_1,\cdots, g_L\}$, $V = \{h_1, \cdots, h_M\}$,  $\rho(U,V) = \min_{l \in 1 \cdots L,m \in 1,\cdots, M} \rho(g_l,h_m)$ is thus $\tilde{\alpha}(u, v, r)  = \tilde{\alpha}(L\leq u, M \leq v, \rho(U,V) \geq r) = \sup_{L\leq u, M \leq v, \rho(U,V) \geq r} \alpha(\sigma(U), \sigma(V)) $.
If the group size are the same, i.e. $L$, then the mixing coefficients of the grouped observations have the following relationship with respect to it in the original field $\tilde{\alpha}(u, v, r) = \alpha(uL, vL, r).$
We can assume $\tilde{\alpha}(u, v, r)  = (uL+ vL)^{\tau} \hat{\alpha}(r)$.

Assume that  $L^{\tau}\hat{\alpha}(r) \to 0$ as $r \to \infty,$ and $\tilde{\vps}$ would maintain the $\alpha-$ mixing property.
Define the ball around group $g$ with radius $s$ to be $\mathcal{F}_g(s) = \sigma\{\cup_{h: \rho(g,h)\leq s} B_h\}.$

\begin{itemize}

\item[A.6)]The $\alpha-$ mixing coefficients of the input field $\tilde{\vps}$ satisfy $\tilde{\alpha}(u,v,r) \leq \phi(uL,vL) \hat{\alpha}(r),$ with $\phi(uL,vL) = (u+v)^{\tau}L^{\tau}$ and for some $\hat{\alpha}(r)$,
 $\sum^{\infty}_{r=1} L^{\tau} r^{d-1} \hat{\alpha}(r)< \infty.$

\item[A.7)] We assume moment conditions on the objects involved to prove the NED property of $\mathbf{H}_G(\theta,\gamma)$.
$b_{ij} \defeq e_i^{\top}(\mathbf{1}^{\top} \mathbf{W}_g(\theta,\gamma) \otimes I_g )|\partial{\mbox{Vec}(\nabla\mm_g(\theta))}/\partial \theta|_a e_j.$
$c_{ij} =  e_i^{\top}(\mathbf{1}^{\top}\otimes \nabla \mm_g^{\top}(\theta))|\partial{\mbox{Vec}(\nabla\mm_g(\theta))}/\partial \theta|_a e_j$.
$\|\ssup b_{ij}\|$ and $\|\ssup c_{ij}\|$ are finite.

\item[A.8)](Identifiability)Let $\overline{Q}_{G}\left( \mathbf{\theta }, \mathbf{\gamma}\right) \defeq \frac{1}{|M_{G}||D_G|}%
\sum_{g}\mathrm{\E}\left( q_{g}\left( \mathbf{\theta }, \mathbf{\gamma}\right) \right) .$
 Recall that $Q_{\infty}( \theta, \gamma)\defeq \lim_{G\rightarrow \infty }\bar{Q}_{G}\left( \mathbf{
\theta, \gamma }\right) .$  Assume that $\theta^0, \gamma^0$ are identified unique in a sense that \\ $\liminf_{G\to \infty}\mathbf{inf}_{\theta \in \Theta: \nu(\theta, \theta^0) \geq \vps }Q_{G}\left( \mathbf{\theta }, \mathbf{\gamma}\right) > c_0> 0$, for any $\gamma$ and a positive constant $c_0$.
\end{itemize}

\textbf{Remark} A.8) can be implied from positive definiteness of $\mW_g(\theta, \gamma)$ and the same identification assumption $\liminf_{G\to \infty}\mbox{inf}_{\theta \in \Theta: \nu(\theta, \theta_0) \geq \vps }Q'_{G} (\theta, \gamma)> c_0>0$ on $Q'_{G} (\theta, \gamma)\defeq \frac{1}{M_{G}|D_G|}%
\sum_{g\in |D_G|} \E\left[ \mathbf{y}_{g}-\mathbf{m}%
_{g}\left( \mathbf{x}_{g};\mathbf{\theta }\right) \right]^{\top }\left[ \mathbf{y}_{g}-\mathbf{m}%
_{g}\left( \mathbf{x}_{g};\mathbf{\theta }\right) \right]$.
As it can be seen that with probability $1- \Co_p(1)$ \\$\liminf_{G\to \infty}\mbox{inf}_{\theta \in \Theta: \nu(\theta, \theta_0) \geq \vps }Q_{G}\left( \mathbf{\theta }, \mathbf{\gamma}\right)> \liminf_{G\to \infty}\mbox{inf}_{\theta \in \Theta: \nu(\theta, \theta_0) \geq \vps }\lambda_{min} \{\mW_g(\theta, \gamma)\} Q'_{\infty} (\theta, \gamma),$ where $\lambda_{min} \{\mW_g(\theta, \gamma)\}$ is the minimum eigenvalue of the matrix $\lambda_{min} \{\mW_g(\theta, \gamma)\}$. If we assume that with probability $1-\Co_p(1),$ $\lambda_{min} \{\mW_g(\theta, \gamma)\}>c$ where $c$ is a positive constant.
We now comment on assumptions, Condition A.2) is concerning the $L_2$ NED property of our data generating processes.
A.3) and A.4) are the standard regularities assumptions. A.5) is a few moment assumptions on the statistical objects involved in the estimation. A.6) is the mixing coefficients restrictions after grouping observations. A.7) is again moment conditions on the elementwise Hessian matrices. A.8) is a condition on identification of our estimator.
Given the assumptions, we can provide the consistency property of our estimation.

\begin{theorem} \label{th:consistency}
(Consistency) Under A.1)-A.8) the GEE-estimator in (\ref{gobj}) is
consistent, that is, $\nu(\mathbf{\hat{\theta}}, \mathbf{%
\theta }^{0}) \to _{p}0$ as $G\rightarrow \infty .$
\end{theorem}
Theorem \ref{th:consistency} indicates the consistency of the estimation as long as the number of groups tends to infinity.
The proof is in the Appendix. To prove further the asymptotic normality of the estimation we need in addition the following assumptions.

\begin{itemize}
\item[A.9)] The true point $\theta^0, \gamma^0$ lies in the interior point of $\Theta, \Gamma$. $\hat{\gamma}$ is estimated with $|\hat{\gamma} - \gamma^0|_2 = \Co_p(G^{-1/2}).$\\
\textbf{Remark}
Verification of this assumption is in Proposition \ref{pre} and its proof in the Appendix.

\item[A.10)] $c'<\lambda_{min}(M_G^{-2}\E\left(\nabla \mathbf{m}_{g}^{\top }(\theta^0)\mathbf{W}_{g}^{
-1}(\theta^0,\gamma^0)\nabla \mathbf{m}_{g}(\theta^0)\right))\\< \lambda_{max} (M_G^{-2}\E\left(\nabla \mathbf{m}_{g}^{\top }(\theta^0)\mathbf{W}_{g}^{
-1}(\theta^0,\gamma^0)\nabla \mathbf{m}_{g}(\theta^0)\right) < C'$ is positive definite, and $c'$ and $C'$ are two positive constants.\\

Define $\mathbf{u}_g = \my_g - \mm_g(\theta^0)$ and $\hat{\mathbf{u}}_g = \my_g- \mm_g(\hat{\theta}) $
\begin{equation}
\mathbf{S}_{G}\left( \mathbf{\theta ,\hat{\gamma}}\right) =\frac{1}{M_{G}|D_G|}%
\sum_{g}\nabla \mathbf{m}_{g}^{\top}\left( \mathbf{\theta }\right)
\mathbf{W}_{g}^{-1}\left( \mathbf{\theta,\hat{\gamma}}\right) \left[ \mathbf{y}_{g}-%
\mathbf{m}_{g}\left( \mathbf{\theta }\right) \right] . \label{quasiscore}
\end{equation}
Define
\begin{eqnarray}\label{as}
AS_G
&=&\frac{1}{G}\sum_{g}\E\left[ \nabla \mathbf{m}_{g}^{\top
}\left( \mathbf{\theta }^{0}\right) \mathbf{W}_{g}^{-1}\left( \mathbf{\theta }^{0},\mathbf{\gamma
}^{0 }\right) \mathbf{u}_{g}\mathbf{u}_{g}^{\top }\mathbf{W}%
_{g}^{-1}\left( \mathbf{\theta }^{0},\mathbf{\gamma }^{0 }\right) \nabla \mathbf{m}_{g}\left(
\mathbf{\theta }^{0}\right) \right]  \\
&&+\frac{1}{G}\sum_{g}\sum_{h, h\neq g}\E\left[ \nabla
\mathbf{m}_{g}^{\top }\left( \mathbf{\theta }^{0}\right) \mathbf{W}%
_{g}^{-1}\left( \mathbf{\theta }^{0}, \mathbf{\gamma }^{0 }\right) \mathbf{u}_{g}\mathbf{u}^{\top}_{h}%
\mathbf{W}_{h}^{-1}\left(\mathbf{\theta }^{0}, \mathbf{\gamma }^{0 }\right) \nabla \mathbf{m}%
_{h}\left( \mathbf{\theta }^{0}\right) \right] , \nonumber
\end{eqnarray}
 and $AS_{\infty} = \lim_{G\to \infty} AS_G$.

\item[A.11)] $\mathbf{S}_{G}\left( \mathbf{\hat{\gamma}, \hat{\theta}} \right) = \Co_p(1)$.
$\inf_G |D_G|^{-1} M_G^{-2}  \lambda_{min}(\mathbf{AS}_{\infty})> 0,$ where $\mathbf{AS}_{\infty}$ is defined in equation (\ref{as}).
The mixing coefficients satisfy $\sum^{\infty}_{r =1}r^{(d \tau^*+d)-1}L^{\tau^*} \hat{\alpha}^{\delta/(2+\delta)}(r)< \infty.$ ($\tau^* = \delta \tau /(4+2\delta)$).
\end{itemize}

A.9) is concerning the the pre-estimation of the nuisance parameter $\gamma$, and A.10), A.11) are two standard assumptions on the regularities of the estimation. Note that $\mathbf{S}_{G}\left( \mathbf{ \hat{\theta},\hat{\gamma}} \right) = \Co_p(1)  = 0$ if $\hat{\theta}, \hat{\gamma}$ lies in the interior point of the parameter space.
In the following, we verify that with our proposal of estimating $\hat{\gamma}$ in (\ref{pregamma}) in Section \ref{sec2} , we will achieve A.9).

\begin{proposition} \label{pre}
Under A.1)-A.3), A.5), A.6) and A.8)', A.9)', A.11)', ( A.8)', A.9)', A.11)'are defined in the Appendix), the estimator solving equation (\ref{pregamma}) satisfies,
\begin{equation}
|\hat{\gamma} - \gamma^0|_2 = \Co_p(1/{\sqrt{G}}).
\end{equation}
\end{proposition}

 $\mathbf{H}_{\infty} \defeq \lim_{G\to \infty} \E\mathbf{H}_G(\theta^0,\gamma^0),$ where $\mathbf{H}_G(\theta^0,\gamma^0)$ is defined in equation (\ref{eq:hessian}).
 It is not surprising to see that our estimation will be asymptotically normally distributed, with a variance covariance matrix of a sandwich form $AV(\theta^0)$, which involves the Hessian.  The rate of convergence is shown to be $\sqrt{G}$.

\begin{theorem} \label{th:normality}
Under A.1) - A.11), we have $AV(\theta^0) \defeq \mathbf{H}_{\infty}^{\top} \mathbf{AS}_{\infty} \mathbf{H}_{\infty}$.
\begin{equation}
\sqrt{G}AV(\theta^0)^{-1/2}(\hat{\theta} - \theta^0) \Rightarrow \mathbb{N}(0,I_p).
\end{equation}
\end{theorem}

\subsection{Consistency of variance covariance matrix estimation}
In this subsection, we propose a semiparametric estimator of the asymptotic variance in Theorem \ref{th:normality}, and prove its consistency. The estimation is tailored to account for the spatial dependency of the underlying process. This facilitates us to create a confidence interval for our estimation.

 First let%
\begin{eqnarray}
\mathbf{\hat{A}} &\mathbf{=}&\frac{1}{|D_G|}\sum_{g}\nabla \mathbf{\hat{m}}%
_{g}^{\top }\mathbf{\hat{W}}_{g}^{-1}\nabla \mathbf{\hat{m}}_{g}, \\
\mathbf{\hat{B}} &=&\frac{1}{|D_G|}\sum_{g}\sum_{h\neq g}k(d_{gh})\nabla
\mathbf{\hat{m}}_{g}^{\top }\mathbf{\hat{W}}_{g}^{-1}\mathbf{\hat{u}}_{g}%
\mathbf{\hat{u}}_{h}^{\top }\mathbf{\hat{W}}_{h}^{-1}\nabla \mathbf{\hat{m}%
}_{h}^{\top },
\end{eqnarray}%
where $\nabla \mathbf{\hat{m}}_{g}\equiv \nabla \mathbf{\hat{m}}_{g}\left(\mathbf{\hat{\theta}}\right) ,$ $\mathbf{\hat{W}}_{g}\equiv \mathbf{\hat{W}}_{g}(\mathbf{\hat{\gamma}}, \hat{\theta})$.


 The estimator of $\mathrm{AV}\left( \mathbf{\theta}^0\right)$ which is robust to misspecification of the variance covariance
matrix is
\begin{eqnarray}
\widehat{\mathrm{AV}}\left( \mathbf{\hat{\theta}}\right)
&=&|D_G| \left( \sum_{g}\nabla \mathbf{\hat{m}}_{g}^{\top}\mathbf{\hat{W}}%
_{g}^{-1}\nabla \mathbf{\hat{m}}_{g}\right) ^{-1}
\notag \\
&&\left(
\sum_{g}\sum_{h(\neq g)}\nabla \mathbf{\hat{m}}_{g}^{\top}\mathbf{%
\hat{W}}_{g}^{-1}k\left( d_{gh}\right) \mathbf{\hat{u}}_{g}\mathbf{\hat{u}}%
_{h}^{\top }\mathbf{\hat{W}}_{h}^{-1}\nabla \mathbf{\hat{m}}_{h}\right)\nonumber \\
&&\left( \sum_{g}\nabla \mathbf{\hat{m}}_{g}^{\top}\mathbf{\hat{W}}%
_{g}^{-1}\nabla \mathbf{\hat{m}}_{g}\right) ^{-1},  \label{Avar}\\
&=& \hat{\mathbf{A}}^{-1}\hat{\mathbf{B}} \hat{\mathbf{A}}^{-1}
\end{eqnarray}%
where $k\left( d_{gh}\right) $ is the kernel function depending on the distance
between group $g$ and $h$, i.e. $\rho(g,h)$, and a bandwidth parameter $h_g$. As noted in \cite{kelejian2007hac}, there are many choices
 for the kernel functions, such as  rectangular kernel, Bartlett or triangular kernel, etc.
 In particular, without loss of generality, we can choose the Bartlett kernel function $k\left( d_{gh}\right)  = 1- \rho(g,h)/h_g$, for $\rho(g,h) < h_g$ and $k\left(g,h\right) =0$ for $ \rho(g,h) \geq h_g$. Further, we can obtain the average partial
effects (\textrm{APE}) of interest and carry on valid inference.

We now list the assumptions needed for the consistency of estimator of $AV(\theta^0)$.
\begin{itemize}
\item[B.1)] $\hat{\mbu}_g- \mbu_g = C_g \Delta_g,$ where $C_g$ is a $L \times p$, and $\Delta_g$ is a $p\times 1$ dimensional vector, with
the condition that $ |C_g|_2 = \CO_p(1),$ and $|\Delta_g|_2 = \CO_p((p G)^{-1/2}).$

\item[B.2)] The moment is bounded by a constant $ \max_{h:  \rho(h,g)\leq h_g}\E |Z_h|^{q'} \leq M L^2,$ $q'\geq 1,$ and $M$ is a constant, where $Z_{h}\defeq  \nabla \mathbf{m}^{\top}_{h}(\theta^0)\mathbf{W}_{h}^{-1}(\theta^0,\gamma^0)\mathbf{u}_{h}$.
\item[B.3)]
$|k(d_{gh})-1|\leq C_k |d_{gh}/h_g|^{\rho_K}$ for $d_{gh}\leq 1$ for some constant $\rho_k\geq 1$ and $0<C_k<\infty$
$ M_G^{-2}|D_G|^{-1}\sum_g \sum_h  |\rho(g,h)/h_g|^{\rho_k}\|e_i^{\top}Z_g^{\top}\|\|Z_he_j\| = \Co(1).$
\item[B.4)]
 Assume that $h_g^{d/q'} |D_G|^{-1}L^{d/q'}L^2 = \Co(1)$ , $h_g^{2d} L^{2d}\sum^{\infty}_{r=1} r^{(d\tau^*+d)-1}\hat{\alpha}^{\delta/(2+\delta)}(r)= \CO(G)$, and $h_g^{2d} \sum^{\infty}_{r=1}  L^{2d} r^{d-1}\psi((r-h_g)_{+}) = \CO(G)$, ($(r-h_g)_{+} = \max(r-h_g,0)$) where $\delta$ is a constant and
 $\delta^* = \delta \tau/(2+\delta)$.

\end{itemize}

B.1) is an assumption for decomposing the difference between the residuals and the true error, as in \cite{kelejian2007hac}. B.2) is about the moment bound and B.3) is on property of the kernel function. B.4) constrains on the spatial dependence coefficients and the bandwidth length. We provide in the following theorem the consistency of the $\widehat{\mathrm{AV}}\left( \mathbf{\hat{\theta}}\right)$. It is worth noting that we prove an elementwise version of the consistency, and the results below can be verified equivalently in any matrix norm, as we consider fixed dimension parameter.

\begin{theorem}\label{th:variance}
Under assumption B.1)- B.4) and A.1) - A.8). The variance-covariance estimator in (\ref{Avar}) is consistent. ${\widehat{\mathrm{AV}}}\left( \mathbf{\hat{\theta}}\right)\to_p \mathrm{AV}\left( \mathbf{\theta}^0\right).$
\end{theorem}

\section{Monte Carlo Simulations}\label{sec5}

In this section, we use Monte Carlo simulations to investigate the finite
sample performances of our proposed GEE approach with groupwise data compared to the partial QMLE. We
simulated count data and binary response data separately. We show that our GEE method is very critical for improving the efficiency of our estimation.The simulation mechanism is described as follows.

\subsection{Sampling Space}

We use sample sizes of 400 or 1600. We sample observations on a lattice.
For example, for sample size of 400, the sample space is a $20\times 20$
square lattice. Each observation resides on the intersections of this
lattice. The locations for the data are $\{(r,s):r,s=1,2,...,20\}$. The
distance $d_{ij}$ between location $i$ and $j$ is chosen to be the Euclidean distance.
Suppose $A(a_{i},a_{j})$ and $B(b_{i},b_{j})$ are the two points on the
lattice; their distance $d_{ij}$ is $\sqrt{%
(a_{i}-b_{i})^{2}+(a_{j}-b_{j})^{2}}$. The spatial correlation is based
on a given parameter $\rho $ and $d_{ij}$. The data are divided into groups
of 4 and the number of groups are set to be 100 for sample size 400. Similarly, for the
sample size of 1600, we use a $40\times 40$ lattice. We still use sample
size of 4 in each group and there are 400 groups in total. For simplicity,
we keep the pairwise distances in different groups the same.

\subsection{Count data}
\subsubsection{Data generating process}

In the count data case, for a Poisson distribution the variances and covariances of the count dependent
variable can be written in closed forms given the spatial correlation in the
underlying spatial error term. That is, by knowing the correlations in the spatial
error term, we can derive the correlations in the count dependent variable as shown in (\ref{geevar}) and (\ref{geecov}).
Consider the following spatial count data generating process: 1. $%
v_{i}$ is simulated as a multivariate lognormal variable with \textrm{E}$%
\left( v_{i}\right) =1$, exponentiating an underlying multivariate normal
distribution using with correlation matrix $W$. Let $a_{i}$ be the underlying multivariate normal
distributed variable. Then $v_{i}=\exp \left( a_{i}\right) $ follows a multivariate lognormal
distribution. We describe the underlying spatial process in Case 1, 2, and 3 as three
special cases to demonstrate different spatial correlations. 2. The coefficient parameters and explanatory variables are set as follows:
 $\beta _{1}=0.5,\beta _{2}=1,\beta _{3}=1,\beta _{4}=1.$; $ x_{2}\sim $\textrm{N}$\left(
0,0.25\right), x_{3}\sim\mathrm{Uniform}\left( 0,1\right), x_{5}\sim $\textrm{N}$\left(
0,1\right), x_{4}=1[x_{5}>0].$ 3. The mean function for individual i is $m_{i}=v_{i}\exp \left( \beta
_{1}+\beta _{2}x_{2}+\beta _{3}x_{3}+\beta _{4}x_{4}\right) ;$ 4. Finally we draw the dependent variable from the Poisson distribution with mean $m_{i}$: $y_{i}\sim\mathrm{Poisson}\left( m_{i}\right) .$ Specifically, the underlying spatial error $a_{i}$ has the
following three cases.

\textbf{Case 1}. $a_{i}=\left( I-\rho W\right) ^{-1}e_{i},e_{i}\sim $\textrm{N}$\left( 0,1\right);$ $W$ is the matrix with $W_{g}$ on the diagonal,
 $g=1,2,...,G.$ Other elements in $W$ are equal to zero. For group size equal to four,%
\begin{equation}
W_{g}=\frac{1}{3}\left(
\begin{array}{cccc}
0 & 1 & 1 & 1 \\
1 & 0 & 1 & 1 \\
1 & 1 & 0 & 1 \\
1 & 1 & 1 & 0%
\end{array}%
\right) .  \label{wg1}
\end{equation}

\textbf{Case 2}. $a_{i}=\left( I-\rho W\right) ^{-1}e_{i},e_{i}\sim $\textrm{N}$\left( 0,1\right);$ $W$ is the matrix with $W_{g}$ on the diagonal,
$g=1,2,...,G.$ The $(l,m)$th element in $W_{g}$, $W_{g\_lm}=$\ \ $\frac{\rho }{6\ast
d_{g_{\_}lm}},$ $\rho =0,0.5,1,1.5,l\neq m;$ $W_{g_{\_}lm}=0,l=m$ for group $g$.
Correlations are zero if observations are in different groups. For group size equal to four,%
\begin{equation}
W_{g}=\frac{1}{6}\left(
\begin{array}{cccc}
0 & \frac{\rho}{d_{g\_12}}\ & \frac{\rho}{d_{g\_13}}\ & \frac{\rho}{d_{g\_14}}\ \\
\frac{\rho}{d_{g\_21}}\ & 0 & \frac{\rho}{d_{g\_23}}\ & \frac{\rho}{d_{g\_24}}\ \\
\frac{\rho}{d_{g\_31}}\ & \frac{\rho}{d_{g\_32}}\ & 0 & \frac{\rho}{d_{g\_34}}\ \\
\frac{\rho}{d_{g\_41}}\ & \frac{\rho}{d_{g\_42}}\ & \frac{\rho}{d_{g\_43}}\ & 0%
\end{array}%
\right) .  \label{wg2}
\end{equation}

\textbf{Case 3}. In this case, the DGP has the following differences from Case 1 and Case 2.
$a_{i}$ is simulated as a multivariate lognormal variable by
exponentiating an underlying multivariate normal distribution \textrm{N}$%
\left( -\frac{1}{2},1\right) $ using with correlation matrix $W$. $W_{ij}=$\
\ $\frac{\rho }{d_{ij}},$ $\rho =0,0.2,0.4,0.6,i\neq j;$ $W_{ii}=1;$ $%
i,j=1,2,...,N.$ The underlying normal distribution implies that $v_{i}$
follows a multivariate lognormal distribution with \textrm{E}$\left(
v_{i}\right) =1$. We set $\beta _{1}=-1,\beta _{2}=1,\beta _{3}=1,\beta _{4}=1.$ $%
x_{2}$ follows a multivariate normal distribution \textrm{N}$\left(0,W\right) ;$
In this case, the data has general spatial correlations for
each pair of observations if $\rho \neq 0.$

\begin{equation}
\mathbf{W}=\left(
\begin{array}{ccccc}
1 & \frac{\rho }{d_{12}} & \frac{\rho }{d_{13}} & \cdots  & \frac{\rho }{%
d_{1N}} \\
\frac{\rho }{d_{21}} & 1 &  & \vdots  & \frac{\rho }{d_{2N}} \\
\frac{\rho }{d_{31}} &  & 1 &  & \vdots  \\
\vdots  & ... &  & \ddots  & \frac{\rho }{d_{N-1,N}} \\
\frac{\rho }{d_{N1}} & \frac{\rho }{d_{N2}} & \cdots  & \frac{\rho }{%
d_{N,N-1}} & 1%
\end{array}%
\right)   \label{W}
\end{equation}%
\begin{equation*}
\label{count3}
\end{equation*}

\subsubsection{Simulation results}

Table \ref{count1}, Table \ref{count2} and Table \ref{count3} show three cases of simulation results with
1000 replications with two different samples and group sizes: (1) $N=400,$ $%
G=100,$ $L=4$ (2) $N=1600,$ $G=400,$ $L=4$. There are four estimators,
Poisson partial \textrm{QMLE} estimator, \ Poisson GEE, Negative Binomial II (NB II) partial QMLE, and  NB II GEE. For simplicity, we use an exchangeable working correlation matrix for GEE estimators. We can see that, first as spatial correlation increases the GEE methods has smaller standard
deviations than QMLE. Second, when there is little spatial correlation, GEE
does not increase much finite sample bias due to accounting for possible
spatial correlation.

In Case 1, when there is no spatial correlation, the Poisson QMLE should be
as efficient as GEE asymptotically. We can see that when $\rho =0,$ the
coefficient estimates and their standard deviations of Poisson QMLE and GEE
are pretty close, which means that there is little finite sample bias due to
accounting for possible spatial correlation when there is actually no
spatial correlation. The standard deviations for the estimated coefficients of Poisson QMLE and GEE
are almost the same. The standard deviation of $\hat{\beta}_{2}$ equals $%
0.259$ for Poisson QMLE and $0.260$ for Poisson GEE when $\rho =0$ for a
sample size of 400. As $\rho $ grows larger. the GEE estimator shows more
and more efficiency improvement over the partial QMLE. For example, for a
sample size of 400, when $\rho =1,$ the standard deviation of $\hat{\beta}%
_{2}$ equals $0.267$ for Poisson QMLE and $0.259$ for Poisson GEE. When $%
\rho =1.5,$ the standard deviation of $\hat{\beta}_{2}$ equals $0.320$ for
Poisson GEE and $0.302$ for Poisson PQMLE. The NB II GEE also has some
improvement over NB II PQMLE. When $\rho =1,$ the standard deviation of $%
\hat{\beta}_{2}$ equals $0.234$ for NB II PQMLE and $0.226$ for NB II GEE.
When $\rho =1.5,$ the standard deviation of $\hat{\beta}_{2}$ equals $0.276$
for NB II PQMLE and $0.261$ for NB II GEE. When sample size increases from
400 to 1600, we see the similar scenarios. Case 2 and Case 3 have shown
similar efficiency results for the GEE estimators.

{\begin{table}[tbp] \centering%
\caption{ Means and Standard Deviations for Count Case 1, averaged over
$1000$ samples.}%
\renewcommand{\arraystretch}{0.7}%
\scriptsize
\begin{tabular}{llcclllllll}
\hline\hline
&  & \multicolumn{4}{c}{N=400,G=100,L=4} &  & \multicolumn{4}{l}{
N=1600,G=400,L=4} \\
&  & Poisson & GEE-poisson & NB II & GEE-nb2 &  & Poisson & GEE-poisson & NB II & GEE-nb2
\\ \cline{1-2}\cline{3-11}
$\rho =0$ & $\hat{\beta}_{2}$ & $1.000$ & $0.999$ & \multicolumn{1}{c}{$1.002
$} & $1.002$ &  & 0.994 & 0.994 & 0.997 & 0.997 \\
& s.d.$\left( \hat{\beta}_{2}\right) $ & $\mathbf{0.259}$ & ${0.260}$
& \multicolumn{1}{c}{$\mathbf{0.227}$} & ${0.228}$ & ${}$ & $\mathbf{%
0.160}$ & $\mathbf{0.160}$ & $\mathbf{0.136}$ & $\mathbf{0.136}$ \\
\cline{2-2}\cline{3-11}
& $\hat{\beta}_{3}$ & $1.000$ & $0.999$ & \multicolumn{1}{c}{$1.002$} & $%
1.002$ &  & 0.999 & 1.000 & 0.998 & 0.998 \\
& s.d.$\left( \hat{\beta}_{3}\right) $ & $\mathbf{0.259}$ & $0.260$ &
\multicolumn{1}{c}{$\mathbf{0.227}$} & $0.228$ &  & $\mathbf{0.137}$ & $\mathbf{0.137}$ & $\mathbf{0.121}$ & $%
\mathbf{0.121}$ \\ \cline{2-2}\cline{3-11}
& $\hat{\beta}_{4}$ & $0.998$ & $0.998$ & \multicolumn{1}{c}{$0.996$} & $%
0.996$ &  & 1.003 & 1.003 & 1.003 & 1.003 \\
& s.d.$\left( \hat{\beta}_{4}\right) $ & $\mathbf{0.146}$ & $0.147$ &
\multicolumn{1}{c}{$\mathbf{0.137}$} & $\mathbf{0.137}$ &  & $\mathbf{0.071}$ & $\mathbf{0.071}$ & $\mathbf{0.067}$ & $%
\mathbf{0.067}$ \\ \cline{1-2}\cline{3-11}
$\rho =0.5$ & $\hat{\beta}_{2}$ & $0.985$ & $0.985$ & \multicolumn{1}{c}{$%
0.993$} & $0.994$ &  & 1.000 & 1.000 & 1.000 & 0.999 \\
& s.d.$\left( \hat{\beta}_{2}\right) $ & $0.256$ & $\mathbf{0.255}$ &
\multicolumn{1}{c}{$0.216$} & $\mathbf{0.215}$ &  & $\mathbf{0.127}$ & $\mathbf{0.127}$ & $0.110$ & $%
\mathbf{0.109}$ \\ \cline{2-2}\cline{3-11}
& $\hat{\beta}_{3}$ & $1.006$ & $1.006$ & \multicolumn{1}{c}{$1.004$} & $%
1.005$ &  & 1.005 & 1.004 & 1.003 & 1.003 \\
& s.d.$\left( \hat{\beta}_{3}\right) $ & $0.211$ & $\mathbf{0.210}$ &
\multicolumn{1}{c}{$0.180$} & $\mathbf{0.179}$ &  & $\mathbf{0.106}$ & $\mathbf{0.106}$ & $\mathbf{0.092}$ & $%
\mathbf{0.092}$ \\ \cline{2-2}\cline{3-11}
& $\hat{\beta}_{4}$ & $1.002$ & $1.002$ & \multicolumn{1}{c}{$1.003$} & $%
1.003$ &  & 1.003 & 1.003 & 1.002 & 1.002 \\
& s.d.$\left( \hat{\beta}_{4}\right) $ & $\mathbf{0.117}$ & $\mathbf{0.117}$ &
\multicolumn{1}{c}{$0.111$} & $\mathbf{0.110}$ &  & $\mathbf{0.058}$ & $\mathbf{0.058}$ & $\mathbf{0.054}$ & $%
\mathbf{0.054}$ \\ \cline{1-2}\cline{3-11}
$\rho =1$ & $\hat{\beta}_{2}$ & $0.987$ & $0.988$ & \multicolumn{1}{c}{$0.991
$} & $0.991$ &  & $0.998$ & 0.997 & 0.997 & 0.997 \\
& s.d.$\left( \hat{\beta}_{2}\right) $ & ${0.267}$ & $\mathbf{0.259}$
& \multicolumn{1}{c}{${0.234}$} & $\mathbf{0.226}$ &  & ${0.130%
}$ & $\mathbf{0.127}$ & ${0.130}$ & $\mathbf{0.128}$ \\
\cline{2-2}\cline{3-11}
& $\hat{\beta}_{3}$ & $1.003$ & $1.003$ & \multicolumn{1}{c}{$1.004$} & $%
1.004$ &  & 1.000 & 0.999 & 1.000 & 0.999 \\
& s.d.$\left( \hat{\beta}_{3}\right) $ & $0.220$ & $\mathbf{0.214}$ &
\multicolumn{1}{c}{$0.195$} & $\mathbf{0.190}$ &  & $0.105$ & $\mathbf{0.102}$ & $0.094$ & $%
\mathbf{0.091}$ \\ \cline{2-2}\cline{3-11}
& $\hat{\beta}_{4}$ & $0.995$ & $0.996$ & \multicolumn{1}{c}{$0.996$} & $%
0.997$ &  & $1.000$ & 1.000 & 1.000 & 1.000 \\
& s.d.$\left( \hat{\beta}_{4}\right) $ & $0.120$ & $\mathbf{0.119}$ &
\multicolumn{1}{c}{$0.113$} & $\mathbf{0.111}$ &  & $0.060$ & $\mathbf{0.058}$ & ${0.056}$ & $%
\mathbf{0.054}$ \\ \cline{1-2}\cline{3-11}
$\rho =1.5$ & $\hat{\beta}_{2}$ & \multicolumn{1}{l}{\ 0.980} &
\multicolumn{1}{c}{\ 0.982} & \multicolumn{1}{c}{0.995} & $0.998$ &  &
0.988 & 0.988 & 0.995 & 0.997 \\
& s.d.$\left( \hat{\beta}_{2}\right) $ & ${0.320}$ & $\mathbf{0.302}$
& \multicolumn{1}{c}{${0.276}$} & $\mathbf{0.261}$ &  & ${0.183%
}$ & $\mathbf{0.173}$ & ${0.154}$ & $\mathbf{0.145}$ \\
\cline{2-2}\cline{3-11}
& $\hat{\beta}_{3}$ & 0.997 & 0.995 & \multicolumn{1}{c}{0.992} & 0.992 &  &
0.992 & 0.994 & 0.997 & 0.999 \\
& s.d.$\left( \hat{\beta}_{3}\right) $ & $0.288$ & $\mathbf{0.271}$ &
\multicolumn{1}{c}{$0.250$} & $\mathbf{0.234}$ &  & $0.143$ & $\mathbf{0.136}$ & $0.126$ & $%
\mathbf{0.120}$ \\ \cline{2-2}\cline{3-11}
& $\hat{\beta}_{4}$ & 0.997 & 0.999 & \multicolumn{1}{c}{1.000} & 0.998 &  &
1.002 & 1.001 & 1.003 & 1.003 \\
& s.d.$\left( \hat{\beta}_{4}\right) $ & $0.146$ & $\mathbf{0.139}$ &
\multicolumn{1}{c}{$0.139$} & $\mathbf{0.131}$ &  & $0.077$ & $\mathbf{0.072}$ & $0.073$ & $%
\mathbf{0.068}$ \\ \hline\hline
\multicolumn{11}{l}{Note: The estimates with smaller standard deviations are marked with bold.}
\end{tabular}%
\label{count1}
\end{table}%

\begin{table}[tbp] \centering%
\caption{Means and Standard Deviations for Count Case 2, averaged over
$1000$ samples}
\renewcommand{\arraystretch}{0.7}%
\scriptsize
\begin{tabular}{llcclllllll}
\hline\hline
&  & \multicolumn{4}{c}{N=400,G=100,L=4} &  & \multicolumn{4}{l}{
N=1600,G=400,L=4} \\
&  & Poisson & GEE-poisson & NB II & GEE-nb2 &  & Poisson & GEE-poisson & NB II & GEE-nb2
\\ \cline{1-2}\cline{3-11}
$\rho =0$ & $\hat{\beta}_{2}$ & $0.990$ & $0.9990$ & \multicolumn{1}{c}{$%
0.992$} & $0.992$ &  & $0.994$ & $0.994$ & $0.999$ & $0.999$ \\
& s.d.$\left( \hat{\beta}_{2}\right) $ & $\mathbf{0.322}$ & $0.323$
& \multicolumn{1}{c}{$\mathbf{0.267}$} & $0.268$ &  & $\mathbf{0.162%
}$ & $\mathbf{0.162}$ & $\mathbf{0.139}$ & $\mathbf{0.139}$ \\
\cline{2-2}\cline{3-11}
& $\hat{\beta}_{3}$ & $0.986$ & $0.987$ & \multicolumn{1}{c}{$0.992$} & $%
0.992$ &  & $0.997$ & $0.997$ & $0.999$ & $0.999$ \\
& s.d.$\left( \hat{\beta}_{3}\right) $ & $\mathbf{0.281}$ & $\mathbf{0.281}$ &
\multicolumn{1}{c}{$\mathbf{0.244}$} & $\mathbf{0.244}$ &  & $\mathbf{0.137}$ & $\mathbf{0.137}%
$ & $\mathbf{0.119}$ & $\mathbf{0.119}$ \\ \cline{2-2}\cline{3-11}
& $\hat{\beta}_{4}$ & $0.999$ & $1.000$ & \multicolumn{1}{c}{$0.998$} & $%
0.998$ &  & $0.998$ & $0.998$ & $0.998$ & $0.998$ \\
& s.d.$\left( \hat{\beta}_{4}\right) $ & $\mathbf{0.140}$ & $0.141$ &
\multicolumn{1}{c}{$\mathbf{0.133}$} & $\mathbf{0.133}$ &  & $\mathbf{0.076}$ & $\mathbf{0.076}%
$ & $\mathbf{0.071}$ & $\mathbf{0.071}$ \\ \cline{1-2}\cline{3-11}
$\rho =0.5$ & $\hat{\beta}_{2}$ & $0.972$ & $0.971$ & \multicolumn{1}{c}{$%
0.981$} & $0.980$ &  & $1.000$ & $1.000$ & $1.001$ & $1.002$ \\
& s.d.$\left( \hat{\beta}_{2}\right) $ & $\mathbf{0.330}$ & $0.331$ &
\multicolumn{1}{c}{$\mathbf{0.285}$} & $0.286$ &  & $\mathbf{0.164}$ & $0.165$ & $\mathbf{0.136}$ & $\mathbf{0.136}$ \\ \cline{2-2}\cline{3-11}
& $\hat{\beta}_{3}$ & $0.992$ & $0.991$ & \multicolumn{1}{c}{$0.995$} & $%
0.994$ &  & $0.999$ & $0.999$ & $0.999$ & $0.999$ \\
& s.d.$\left( \hat{\beta}_{3}\right) $ & $\mathbf{0.276}$ & $\mathbf{0.276}$ &
\multicolumn{1}{c}{$\mathbf{0.243}$} & $\mathbf{0.243}$ &  & $\mathbf{0.141}$ & $\mathbf{0.141}%
$ & $\mathbf{0.120}$ & $\mathbf{0.120}$ \\ \cline{2-2}\cline{3-11}
& $\hat{\beta}_{4}$ & $0.995$ & $0.995$ & \multicolumn{1}{c}{$0.996$} & $%
0.995$ &  & $0.998$ & $0.998$ & $0.998$ & $0.998$ \\
& s.d.$\left( \hat{\beta}_{4}\right) $ & $\mathbf{0.151}$ & $\mathbf{0.151}$ &
\multicolumn{1}{c}{$0.142$} & $\mathbf{0.141}$ &  & $\mathbf{0.077}$ & $\mathbf{0.077}%
$ & $\mathbf{0.073}$ & $\mathbf{0.073}$ \\ \cline{1-2}\cline{3-11}
$\rho =1$ & $\hat{\beta}_{2}$ & $1.017$ & $1.014$ & \multicolumn{1}{c}{$1.016
$} & $1.014$ &  & $0.998$ & $0.997$ & $0.998$ & $0.997$ \\
& s.d.$\left( \hat{\beta}_{2}\right) $ & $0.400$ & $\mathbf{0.396}$
& \multicolumn{1}{c}{$0.319$} & $\mathbf{0.316}$ &  & $0.193%
$ & $\mathbf{0.191}$ & $0.161$ & $\mathbf{0.159}$ \\
\cline{2-2}\cline{3-11}
& $\hat{\beta}_{3}$ & $0.975$ & $0.976$ & \multicolumn{1}{c}{$0.978$} & $%
0.979$ &  & $1.005$ & $1.004$ & $1.004$ & $1.003$ \\
& s.d.$\left( \hat{\beta}_{3}\right) $ & $\mathbf{0.331}$ & $\mathbf{0.331}$ &
\multicolumn{1}{c}{$0.278$} & $\mathbf{0.276}$ &  & $0.158$ & $\mathbf{0.157}$ & $0.135$ & $%
\mathbf{0.134}$ \\ \cline{2-2}\cline{3-11}
& $\hat{\beta}_{4}$ & $0.998$ & $0.996$ & \multicolumn{1}{c}{$0.995$} & $%
0.994$ &  & $1.000$ & $1.000$ & $1.000$ & $1.000$ \\
& s.d.$\left( \hat{\beta}_{4}\right) $ & $0.185$ & $0.182$ &
\multicolumn{1}{c}{$0.173$} & $0.169$ &  & $0.088$ & $0.087$ & $0.083$ & $%
0.081$ \\ \cline{1-2}\cline{3-11}
$\rho =1.5$ & $\hat{\beta}_{2}$ & \multicolumn{1}{c}{$\ 0.970$} &
\multicolumn{1}{c}{$\ 0.973$} & \multicolumn{1}{c}{$1.013$} & $1.015$ &  & $%
1.004$ & $1.001$ & $1.008$ & $1.004$ \\
& s.d.$\left( \hat{\beta}_{2}\right) $ & $0.677$ & $\mathbf{0.662}$
& \multicolumn{1}{c}{$0.577$} & $\mathbf{0.570}$ &  & $0.311%
$ & $\mathbf{0.302}$ & $0.262$ & $\mathbf{0.255}$ \\
\cline{2-2}\cline{3-11}
& $\hat{\beta}_{3}$ & $0.972$ & $0.972$ & \multicolumn{1}{c}{$0.974$} & $%
0.976$ &  & $0.999$ & $0.997$ & $1.001$ & $1.000$ \\
& s.d.$\left( \hat{\beta}_{3}\right) $ & $0.627$ & $\mathbf{0.611}$ &
\multicolumn{1}{c}{$0.524$} & $\mathbf{0.504}$ &  & $0.293$ & $\mathbf{0.286}$ & $0.239$ & $%
\mathbf{0.233}$ \\ \cline{2-2}\cline{3-11}
& $\hat{\beta}_{4}$ & $1.002$ & $1.000$ & \multicolumn{1}{c}{$1.000$} & $%
0.998$ &  & $0.999$ & $1.000$ & $0.999$ & $1.000$ \\
& s.d.$\left( \hat{\beta}_{4}\right) $ & $0.326$ & $\mathbf{0.318}$ &
\multicolumn{1}{c}{$0.293$} & $\mathbf{0.284}$ &  & $0.160$ & $\mathbf{0.156}$ & $0.144$ & $%
\mathbf{0.141}$ \\ \hline\hline
\multicolumn{11}{l}{Note: The estimates with smaller standard deviations are marked with bold.}
\end{tabular}%
\label{count2}
\end{table}%
}

{\tiny
\begin{table}[tbp] \centering%
\caption{Means and Standard Deviations for Count Case 3, averaged over
$1000$ samples}%
\scriptsize
\renewcommand{\arraystretch}{0.7}%
\begin{tabular}{llcclllllll}
\hline\hline
&  & \multicolumn{4}{c}{N=400, G=100, L=4} &  & \multicolumn{4}{l}{N=1600,
G=400, L=4} \\
&  & Poisson & GEE-poisson & NB II & GEE-nb2 &  & Poisson & GEE-poisson & NB II & GEE-nb2
\\ \cline{1-2}\cline{3-11}
$\rho =0$ & $\hat{\beta}_{2}$ & $0.998$ & $0.998$ & \multicolumn{1}{c}{$1.000
$} & $1.000$ &  & $0.998$ & $0.998$ & $0.999$ & $0.999$ \\
& s.d.$\left( \hat{\beta}_{2}\right) $ & $\mathbf{0.330}$ & $\mathbf{0.330}$
& \multicolumn{1}{c}{$\mathbf{0.266}$} & $0.267$ &  & $\mathbf{0.165}$%
& $\mathbf{0.165}$ & $\mathbf{0.144}$ & $\mathbf{0.144}$ \\ \cline{2-2}\cline{3-11}
& $\hat{\beta}_{3}$ & $1.002$ & $1.002$ & \multicolumn{1}{c}{$1.002$} & $%
1.002$ &  & $0.995$ & $0.995$ & $0.995$ & $0.995$ \\
& s.d.$\left( \hat{\beta}_{3}\right) $ & $\mathbf{0.273}$ & $0.274$ &
\multicolumn{1}{c}{$\mathbf{0.240}$} & $0.241$ &  & $\mathbf{0.138}$ & $\mathbf{0.138}$ & $\mathbf{0.126}$ & $\mathbf{0.126}$ \\
\cline{2-2}\cline{3-11}
& $\hat{\beta}_{4}$ & $0.998$ & $0.999$ & \multicolumn{1}{c}{$0.998$} & $%
0.998$ &  & 0.997 & 0.997 & 0.997 & 0.997 \\
& s.d.$\left( \hat{\beta}_{4}\right) $ & $\mathbf{0.152}$ & $0.153$ &
\multicolumn{1}{c}{$\mathbf{0.142}$} & $0.143$ &  & $\mathbf{0.073}$ & $\mathbf{0.073}$ & $\mathbf{0.069}$ & $\mathbf{0.069}$ \\
\cline{1-2}\cline{3-11}
$\rho =0.2$ & $\hat{\beta}_{2}$ & $0.991$ & $0.911$ & \multicolumn{1}{c}{$%
0.997$} & $0.996$ &  & 0.998 & $0.999$ & $0.999$ & 1.000 \\
& s.d.$\left( \hat{\beta}_{2}\right) $ & $\mathbf{0.312}$ & $\mathbf{0.312}$ &
\multicolumn{1}{c}{$0.272$} & $\mathbf{0.271}$ &  & 0.158 & $\mathbf{0.157}$ & $\mathbf{0.137}$ & $\mathbf{0.137}$ \\
\cline{2-2}\cline{3-11}
& $\hat{\beta}_{3}$ & $0.991$ & $0.991$ & \multicolumn{1}{c}{$0.996$} & $%
0.996$ &  & 1.000 & 1.000 & $0.999$ & $0.999$ \\
& s.d.$\left( \hat{\beta}_{3}\right) $ & $\mathbf{0.265}$ & $0.266$ &
\multicolumn{1}{c}{$\mathbf{0.234}$} & $0.235$ &  & $\mathbf{0.130}$ & $\mathbf{0.130}$ & $\mathbf{0.116}$ & $\mathbf{0.116}$ \\
\cline{2-2}\cline{3-11}
& $\hat{\beta}_{4}$ & $1.002$ & $1.002$ & \multicolumn{1}{c}{$1.003$} & $%
1.004$ &  & $0.999$ & $0.999$ & 0.998 & 0.998 \\
& s.d.$\left( \hat{\beta}_{4}\right) $ & $\mathbf{0.143}$ & $\mathbf{0.143}$ &
\multicolumn{1}{c}{$\mathbf{0.137}$} & $\mathbf{0.137}$ &  & $\mathbf{0.073}$ & $\mathbf{0.073}$ & $\mathbf{0.069}$ & $\mathbf{0.069}$ \\
\cline{1-2}\cline{3-11}
$\rho =0.4$ & $\hat{\beta}_{2}$ & $0.988$ & $0.989$ & \multicolumn{1}{c}{$%
0.992$} & $0.993$ &  & 0.999 & 0.999 & 0.997 & 0.998 \\
& s.d.$\left( \hat{\beta}_{2}\right) $ & $0.305$ & $\mathbf{0.303}$
& \multicolumn{1}{c}{$0.261$} & $\mathbf{0.260}$ &  & 0.162
& $\mathbf{0.160}$ & 0.140 & $\mathbf{0.138}$ \\ \cline{2-2}\cline{3-11}
& $\hat{\beta}_{3}$ & $1.002$ & $1.003$ & \multicolumn{1}{c}{$1.006$} & $%
1.006$ &  & 0.997 & 0.996 & 0.999 & 0.998 \\
& s.d.$\left( \hat{\beta}_{3}\right) $ & $0.267$ & $\mathbf{0.265}$ &
\multicolumn{1}{c}{$0.238$} & $\mathbf{0.237}$ &  & 0.129 & $\mathbf{0.128}$ & 0.117 & $\mathbf{0.116}$ \\
\cline{2-2}\cline{3-11}
& $\hat{\beta}_{4}$ & $0.998$ & $0.998$ & \multicolumn{1}{c}{$0.997$} & $%
0.998$ &  & 1.003 & 1.003 & 1.003 & 1.003 \\
& s.d.$\left( \hat{\beta}_{4}\right) $ & $0.139$ & $\mathbf{0.138}$ &
\multicolumn{1}{c}{$0.131$} & $\mathbf{0.130}$ &  & 0.074 & $\mathbf{0.073}$ & 0.070 & $\mathbf{0.069}$ \\
\cline{1-2}\cline{3-11}
$\rho =0.6$ & $\hat{\beta}_{2}$ & \multicolumn{1}{c}{ 0.995} &
\multicolumn{1}{c}{$ 0.995$} & \multicolumn{1}{c}{$0.999$} & $0.999$ &  &
1.005 & 1.006 & 1.003 & 1.004 \\
& s.d.$\left( \hat{\beta}_{2}\right) $ & $0.300$ & $\mathbf{0.292}$
& \multicolumn{1}{c}{$0.260$} & $\mathbf{0.251}$ &  & 0.161
& $\mathbf{0.156}$ & 0.135 & $\mathbf{0.130}$ \\ \cline{2-2}\cline{3-11}
& $\hat{\beta}_{3}$ & 1.004 & 1.003 & \multicolumn{1}{c}{1.002} & 1.000 &  &
1.002 & 1.001 & 1.006 & 1.005 \\
& s.d.$\left( \hat{\beta}_{3}\right) $ & $0.253$ & $\mathbf{0.247}$ &
\multicolumn{1}{c}{$0.219$} & $\mathbf{0.213}$ &  & 0.131 & $\mathbf{0.128}$ & 0.114 & $\mathbf{0.110}$ \\
\cline{2-2}\cline{3-11}
& $\hat{\beta}_{4}$ & 1.004 & 1.003 & \multicolumn{1}{c}{1.003} & 1.002 &  &
1.000 & 1.000 & 1.000 & 1.000 \\
& s.d.$\left( \hat{\beta}_{4}\right) $ & $0.140$ & $\mathbf{0.138}$ &
\multicolumn{1}{c}{$0.132$} & $\mathbf{0.130}$ &  & 0.072 & $\mathbf{0.071}$ & $\mathbf{0.068}$ & $\mathbf{0.068}$ \\
\hline\hline
\multicolumn{11}{l}{Note: The estimates with smaller standard deviations are marked with bold.}
\end{tabular}%
\label{count3}
\end{table}%
 }

\subsection{Binary response data}

\subsubsection{Data generating process}

For the Probit model, the correlations of latent normal errors result in
correlations of binary response variables, but we cannot easily find the
specific form of the conditional variances and covariances for the binary dependent
variables. The correlations in latent error do not reflect the exact
correlations in the binary dependent variables. Consider the following cases
of data generating process. 1. The latent variable $y^{\ast }=\beta
_{1}+\beta _{2}x_{2}+\beta _{3}x_{3}+\beta _{4}x_{4}+e_{4},$ where $e_{4}$
is the latent spatial error term, and the parameters are set to be $\beta _{1}=\beta
_{2}=\beta _{3}=\beta _{4}=1.$ Then the binary dependent variable is
generated as $y_{i}=1$ if $y_{i}^{\ast }\geq 1.5$ and $y_{i}=0$ if $%
y_{i}^{\ast }<1.5.$ The explanatory variables are set as follows: $\
x_{1}=1; $ $x_{2}\sim \mathrm{N}\left( 1,1\right) ;$ $\
x_{3}=0.2x_{2}-1.2e_{1},e_{1}\sim \mathrm{N}\left( 0,1\right) ;$ $\
x_{5}=0.2x_{2}+0.2x_{3}+e_{2},e_{2}\sim \mathrm{N}\left( 0,1\right) ;$ $%
x_{4}=1\left[ x_{5}>0\right] .$ We consider two cases of latent spatial
error terms and the corresponding binary response variables are generated as
follows.

\textbf{Case 1}. The vector of spatial error $\mathbf{e}_{4}=\left( I-\rho W\right)
^{-1}\mathbf{e}_{3},e_{3}\sim \mathrm{N}\left( 0,1\right) ,$ where $\rho
=0,0.5,1,1.5$ respectively. $W$ is the matrix with $W_{g}$ on the diagonal, $%
g=1,2,...,G.$ Other elements in $W$ are equal to zero. In this case, only
individuals within a group are correlated. For group size equal to four, $%
W_{g}$ is the same as in (\ref{wg1}) in Case 1 for count data.

\textbf{Case 2}. The latent spatial error $e_{4}\sim \mathrm{MVN}(0,$W$ ),$ that
is, $e_{4}$ follows a standard multivariate normal distribution with
expectation zero and $W$  is $N\times N$ correlation matrix. %
$W_{ij}=\frac{\rho }{d_{ij}},$ $\rho =0,0.2,0.4,0.6,i\neq j;$ $W_{ii}=1;$ $%
i,j=1,2,...,N.$
$W$ is the same as in (\ref{W}) in Case 3 for count data. Therefore, the data has general spatial correlations for each pair of observations if $\rho \neq 0.$

\subsubsection{Simulation results}

In the simulation, two estimators are compared, the Probit partial QMLE
estimator, and the Probit GEE estimator with an exchangeable working
correlation matrix. We show two cases of the simulation: (1) N=400, G=100,
L=4; 2) N=1600, G=400, L=4. The replication times are 1000.
The simulation results for Case 1 and Case 2 are in Table \ref{probit1} and Table \ref{probit2}
separately. We find the following results.

First, in both cases, the GEE estimator is less biased than the partial
QMLE estimator. For example, for N=400, in Case 1 when $\rho =1,\hat{\beta}%
_{2}$ equals 1.252 for QMLE and 1.203 for GEE. In Case 2 when $\rho =0.6,%
\hat{\beta}_{2}$ equals 1.148 for QMLE and 1.090 for GEE. Second, the GEE
estimator has some obvious efficiency improvement over partial QMLE. For
example, in case 1 when $\rho =1$, the standard deviation of $\hat{\beta}%
_{2} $ equals $0.280$ for QMLE and $0.173$ for GEE. In Case 2 when $\rho
=0.6,$ the standard deviation of $\hat{\beta}_{2}$ equals 0.270 for QMLE and
0.164 for GEE for a sample size of 400. Third, when we increase the sample size to
1600 and number of groups to 400 correspondingly, the same scenario applies.
What is more, the bias and especially standard deviations for both the
Probit QMLE and GEE reduces. For example, for N=1600, in Case 1 when $\rho
=1,$ the standard deviations of $\hat{\beta}_{2}$ reduce to 0.121 for QMLE
and 0.081 for GEE.

\begin{table}[tbp] \centering
\caption{Means and Standard Deviations for Probit Case 1, averaged over
$1000$ samples}%
\scriptsize
\begin{tabular}{llcccccc}
\hline\hline
&  &  & \multicolumn{2}{c}{ N=400, G=100, L=4} &  &
\multicolumn{2}{c}{ N=1600, G=400, L=4} \\
&  &  & { Probit} & { GEE-probit} &  & { Probit}
& { GEE-probit} \\ \cline{1-3}\cline{3-8}
${ \rho =0}$ & ${ \hat{\beta}}_{2}$ &  & {%
1.076} & {1.033} &  & {1.016} & {1.007}
\\
& { s.d.}$\left( \hat{\beta}_{2}\right) $ &  & ${0.230%
}$ & ${\mathbf {0.142}}$ &  & ${0.103}$ & ${\mathbf {0.069}}$
\\ \cline{2-3}\cline{3-8}
& ${ \hat{\beta}}_{3}$ &  & {1.070} & {%
1.031} &  & {1.016} & {1.018} \\
& { s.d.}$\left(\hat{\beta}_{3}\right) $ &  & ${0.205%
}$ & ${\mathbf{0.127}}$ &  & ${0.084}$ & ${\mathbf{0.059}}$
\\ \cline{2-3}\cline{3-8}
& ${\hat{\beta}}_{4}$ &  & {1.069} & {%
1.021} &  & {1.019} & {1.011} \\
& {s.d.}$\left( \hat{\beta}_{4}\right) $ &  & ${0.304%
}$ & ${\mathbf{0.200}}$ &  & ${0.136}$ & ${\mathbf{0.103}}$
\\ \cline{1-3}\cline{3-8}
${ \rho =0.5}$ & ${ \hat{\beta}}_{2}$ &  &
{1.310} & {1.252} &  & {1.229} &
{ 1.213} \\
& { s.d.}$\left( \hat{\beta}_{2}\right) $ &  & ${0.293%
}$ & ${\mathbf{0.169}}$ &  & ${0.124}$ & ${\mathbf{0.077}}$
\\ \cline{2-3}\cline{3-8}
& ${ \hat{\beta}}_{3}$ &  & {1.310} & {%
1.256} &  & {1.229} & {1.214} \\
& { s.d.}$\left(\hat{\beta}_{3}\right) $ &  & ${0.259%
}$ & ${\mathbf{0.156}}$ &  & ${0.111}$ & ${\mathbf{0.720}}$
\\ \cline{2-3}\cline{3-8}
& ${\hat{\beta}}_{4}$ &  & {1.297} & {%
1.243} &  & {1.227} & {1.213} \\
& { s.d.}$\left(\hat{\beta}_{4}\right) $ &  & ${0.364%
}$ & ${\mathbf{0.238}}$ &  & ${0.165}$ & ${\mathbf{0.112}}$
\\ \cline{1-3}\cline{3-8}
${\rho =1}$ & ${\hat{\beta}}_{2}$ &  & {%
1.254} & {1.203} &  & {1.180} & {1.167}
\\
& {s.d.}$\left(\hat{\beta}_{2}\right) $ &  & ${0.280%
}$ & ${ \mathbf{0.173}}$ &  & ${ 0.121}$ & ${ \mathbf{0.081}}$
\\ \cline{2-3}\cline{3-8}
& ${\hat{\beta}}_{3}$ &  & {1.236} & {%
1.192} &  & {1.176} & {1.164} \\
& { s.d.}$\left( \hat{\beta}_{3}\right) $ &  & ${0.236%
}$ & ${\mathbf{0.152}}$ &  & ${0.105}$ & ${\mathbf{0.072}}$
\\ \cline{2-3}\cline{3-8}
& ${\hat{\beta}}_{4}$ &  & \multicolumn{1}{c}{1.238} & \multicolumn{1}{c}{1.196} &  & { 1.175}
& {1.165} \\
& { s.d.}$\left( \hat{\beta}_{4}\right) $ &  & \multicolumn{1}{c}{%
${ 0.356}$} & \multicolumn{1}{c}{${\mathbf{0.241}}$}
&  & ${ 0.156}$ & ${ \mathbf{0.109}}$ \\ \cline{1-3}\cline{3-8}
${ \rho =1.5}$ & ${ \hat{\beta}}_{2}$ &  &
{ 1.022} & { 0.982} &  & { 0.963} &
{ 0.949} \\
& { s.d.}$\left( \hat{\beta}_{2}\right) $ &  & ${ 0.230%
}$ & ${ \mathbf{0.149}}$ &  & ${ 0.102}$ & ${\mathbf{ 0.070}}$
\\ \cline{2-3}\cline{3-8}
& ${ \hat{\beta}}_{3}$ &  & { 1.013} & { %
0.979} &  & { 0.966} & { 0.953} \\
& { s.d.}$\left( \hat{\beta}_{3}\right) $ &  & ${ 0.196%
}$ & ${ \mathbf{0.132}}$ &  & ${ 0.086}$ & ${ \mathbf{0.063}}$
\\ \cline{2-3}\cline{3-8}
& ${ \hat{\beta}}_{4}$ &  & { 1.003} & { %
0.963} &  & { 0.968} & { 0.953} \\
& { s.d.}$\left( \hat{\beta}_{4}\right) $ &  & ${ 0.309%
}$ & ${ \mathbf{0.209}}$ &  & ${ 0.139}$ & ${\mathbf{ 0.101}}$
\\ \hline\hline
\multicolumn{8}{l}{Note: The estimates with smaller standard deviations are marked with bold.}
\end{tabular}%
\label{probit1}
\end{table}
\begin{table}[tbp] \centering%
\caption{Means and Standard Deviations for Probit Case 2, averaged over
$1000$ samples}%
\renewcommand{\arraystretch}{0.7}%
\scriptsize
\begin{tabular}{llcccccc}
\hline\hline
&  &  & \multicolumn{2}{l}{\scriptsize N=400, G=100, L=4} &  &
\multicolumn{2}{c}{N=1600, G=400, L=4} \\
&  &  & { Probit} & {GEE-probit} &  & {Probit}
& {GEE-probit} \\ \cline{1-3}\cline{3-8}
${\rho =0}$ & ${\hat{\beta}}_{2}$ &  & {%
1.068} & {1.033} &  & {1.009} & {1.004}
\\
& {s.d.}$\left( \hat{\beta}_{2}\right) $ &  & ${0.226%
}$ & ${\mathbf{0.143}}$ &  & ${0.101}$ & ${\mathbf{0.067}}$
\\ \cline{2-3}\cline{3-8}
& ${\hat{\beta}}_{3}$ &  & {1.070} & {%
1.036} &  & {1.011} & {1.006} \\
& {s.d.}$\left( \hat{\beta}_{3}\right) $ &  & ${0.196%
}$ & ${\mathbf{0.127}}$ &  & ${0.085}$ & ${\mathbf{0.061}}$
\\ \cline{2-3}\cline{3-8}
& ${\hat{\beta}}_{4}$ &  & {1.056} & {%
1.019} &  & {1.006} & {1.002} \\
& {s.d.}$\left( \hat{\beta}_{4}\right) $ &  & ${0.301%
}$ & ${\mathbf{0.214}}$ &  & ${0.142}$ & ${\mathbf{0.099}}$
\\ \cline{1-3}\cline{3-8}
${\rho =0.2}$ & ${\hat{\beta}}_{2}$ &  &
{1.100} & {1.046} &  & {1.023} &
{1.013} \\
& {s.d.}$\left(\hat{\beta}_{2}\right) $ &  & ${0.252%
}$ & ${\mathbf{0.139}}$ &  & ${0.102}$ & ${\mathbf{0.070}}$
\\ \cline{2-3}\cline{3-8}
& ${\hat{\beta}}_{3}$ &  & {1.087} & {%
1.040} &  & {1.020} & {1.012} \\
& {s.d.}$\left( \hat{\beta}_{3}\right) $ &  & ${0.212%
}$ & ${\mathbf{0.124}}$ &  & ${0.085}$ & ${\mathbf{0.060}}$
\\ \cline{2-3}\cline{3-8}
& ${\hat{\beta}}_{4}$ &  & {1.096} & {%
1.043} &  & {1.021} & {1.012} \\
& {s.d.}$\left( \hat{\beta}_{4}\right) $ &  & ${0.343%
}$ & ${\mathbf{0.210}}$ &  & ${0.138}$ & ${\mathbf{0.103}}$
\\ \cline{1-3}\cline{3-8}
${\rho =0.4}$ & ${\hat{\beta}}_{2}$ &  &
{1.106} & {1.059} &  & {1.036} &
{1.024} \\
& {s.d.}$\left( \hat{\beta}_{2}\right) $ &  & ${0.257%
}$ & ${\mathbf{0.153}}$ &  & ${0.106}$ & ${\mathbf{0.071}}$
\\ \cline{2-3}\cline{3-8}
& ${\hat{\beta}}_{3}$ &  & {1.099} & {%
1.058} &  & {1.034} & {1.022} \\
& {s.d.}$\left( \hat{\beta}_{3}\right) $ &  & ${ 0.207%
}$ & ${\mathbf{0.133}}$ &  & ${0.091}$ & ${\mathbf{0.065}}$
\\ \cline{2-3}\cline{3-8}
& ${\hat{\beta}}_{4}$ &  & \multicolumn{1}{c}{1.104} & \multicolumn{1}{c}{1.059} &  & {1.031}
& {1.020} \\
& {s.d.}$\left( \hat{\beta}_{4}\right) $ &  & \multicolumn{1}{c}{%
${0.326}$} & \multicolumn{1}{c}{${\mathbf{0.213}}$}
&  & ${0.145}$ & ${\mathbf{0.103}}$ \\ \cline{1-3}\cline{3-8}
${\rho =0.6}$ & ${\hat{\beta}}_{2}$ &  &
{1.148} & {1.090} &  & {1.041} &
{1.035} \\
& {s.d.}$\left( \hat{\beta}_{2}\right) $ &  & ${0.270%
}$ & ${\mathbf{0.164}}$ &  & ${0.109}$ & ${\mathbf{0.077}}$
\\ \cline{2-3}\cline{3-8}
& ${\hat{\beta}}_{3}$ &  & {1.140} & {%
1.089} &  & {1.037} & {1.030} \\
& {s.d.}$\left( \hat{\beta}_{3}\right) $ &  & ${0.238%
}$ & ${\mathbf{0.157}}$ &  & ${0.096}$ & ${\mathbf{0.072}}$
\\ \cline{2-3}\cline{3-8}
& ${\hat{\beta}}_{4}$ &  & {1.131} & { %
1.074} &  & {1.039} & {1.034} \\
& {s.d.}$\left( \hat{\beta}_{4}\right) $ &  & ${0.346%
}$ & ${\mathbf{0.232}}$ &  & ${0.151}$ & ${\mathbf{0.106}}$
\\ \hline\hline
\multicolumn{8}{l}{Note: The estimates with smaller standard deviations are marked with bold.}
\end{tabular}%

\label{probit2}
\end{table}

\section{An empirical application of the inflow FDI to China}\label{sec6}

In the empirical FDI literature, the gravity equation specification was
initially adopted from the empirical literature on trade flows. The gravity
equation has been widely used and extended in international trade since
\cite{tinbergen1962analysis}. \cite{anderson2003gravity} specify the gravity
equation as
\begin{equation}
T_{ij}=\alpha _{0}Y_{i}^{\alpha _{1}}Y_{j}^{\alpha _{2}}D_{ij}^{\alpha
_{3}}\eta _{ij}
\end{equation}%
where $T_{ij}$ is the trade flows between country $i$ and country $j$. $%
T_{ij}$ is proportional to the product of the two countries' GDPs, denoted
by $Y_{i}$ and $Y_{j}$, and inversely proportional to their distance. $%
D_{ij} $ broadly represents trade resistance. Let $\eta _{ij}$ be a
stochastic error that represents deviations from the theory. As a tradition
in the existing literature, by taking the natural logarithms of both sides
and adding other control variables represented by $Z_{ij}$, the
log-linearized equation is:%
\begin{equation}
\ln T_{ij}=\ln \alpha _{0}+\alpha _{1}\ln Y_{i}+\alpha _{2}\ln Y_{j}+\alpha
_{3}\ln D_{ij}+\beta Z_{ij}+\ln \eta _{ij}
\end{equation}

For the above equation, a traditional estimation approach is to use ordinary
least squares (OLS). However, there are two problems with the OLS estimation
of the log linearized model. First, $T_{ij}$ must be positive in order to
take the logarithm. A transformation of $\log(T_{ij}+1)$ can solve the
problem of logarithm but it is not clear how to interpret the estimation
results with respect to the original values. Second, the estimation heavily
depends on the independence assumption of $\eta _{ij}$ and explanatory
variables, which means the variance of $\eta _{ij}$ cannot depend on the
explanatory variables. Because of taking the logarithm, only under very
specific conditions on $\eta _{ij}$ is the log linear representation of the
constant-elasticity model useful as a device to estimate the parameters of
interest (\cite{silva2006log}). Jensen's inequality implies that ($ \E\log
Y $) is smaller than $\log \E(Y)$, thus log-linearized models estimated by OLS
as elasticities can be highly misleading in the presence of
heteroscedasticity. If the variance of $\eta _{ij}$ is dependent on the
explanatory variables, ordinary least squares is not consistent any more.

We adopt this specification and augment it to the inflow FDI to cities of
China. and use nonlinear estimation method, the GEE estimation. The
estimating equation is specified as follows
\begin{eqnarray}
\mathrm{\E}\left( FDI_{i}|X_{i}\right) &=&\exp [\beta _{0}+\beta _{1}\ln
(GDP_{i})+\beta _{2}\ln \left( GDPPC_{i}\right) +\beta _{3}\ln (WAGE_{i})
\nonumber \\
&&+\beta _{4}\ln \left( SCIEXP_{i}\right) +\beta _{5}BORDER_{i}],
\end{eqnarray}%
where $FDI_{i}$ is the inflow FDI in actual use for city $i$, $X_{i}$
represents all explanatory variables. The control variables includes city
level GDP, GDP per capita, the average wage, the government expenditure to
science, and whether the city is on the border. We collect
data of inflow FDI to 287 cities in 31 provincial administrative regions in
2007 in mainland China from the website of Development Research Center of the State Council of P. R. China
\footnote{The website of Development Research Center of the State Council of P. R. China is www.drcnet.com.cn}. Three cities, Jiayuguan (Gansu Province), Dingxi (Gansu Province) and
Karamay (Xinjiang Province), are dropped because of missing data on FDI. Thus we are
using 284 cities in total. We collect the latitudes and longitudes of the
center of each city using Google map and calculated the geographical distance matrix
between cities. The city center is defined as the location of the city
government. We use provinces as natural grouping so there are 31 groups.
Each group has one to twenty cities. The descriptive statistics are in Table \ref{summary}.
The grouping information is in Table \ref{grouping}.

For comparison, we also provide the OLS estimates of the log-linearized
model:
\begin{eqnarray}
\ln \left( FDI_{i}\right) &=&\beta _{0}+\beta _{1}\ln (GDP_{i})+\beta
_{2}\ln \left( GDPPC_{i}\right) +\beta _{3}\ln (WAGE_{i})  \nonumber \\
&&+\beta _{4}\ln \left( SCIEXP_{i}\right) +\beta _{5}BORDER_{i}+u_{i}.
\end{eqnarray}%
The log linearized model suffers from two main problems, first the dependent
variable cannot take log if it is zero; second as mentioned in Silva and
Tenreyro (2006) the log linearization can cause bias in parameter estimates
if there exists heteroskedasticity in the error term $u_{i}$.

To estimate the equation for FDI, we use OLS, Poisson QMLE, Poisson GEE with
the exchangeable working matrix, NB QMLE, NB GEE with the exchangeable
working matrix. In Table \ref{fdiequ} the results show advantage of Poisson GEE estimation. All
estimation results verifies the positive effect of GDP and GDP per capita in
the gravity equation for FDI. These estimates are all significant at the 1\%
level. What is more, the standard error of GDP and GDP per capita for
Poisson GEE is smaller than that for Poisson QMLE, which is smaller than
that for OLS. The Poisson regression has significant results on the
explanatory variables, log(wage), log(sciexp) and border, which are not
significant in the OLS regression. The local average wage has a negative
effect on inflow FDI to this city. Compared to other estimation methods, the
Poisson GEE estimates on log(wage) is the most significant, at 1\% level. It
means that when the average wage increase by 1\%, the inflow FDI would
decrease by about 1\%, which could due to the inhabiting effect of labor
cost. Similarly, when local government increase science expenditure by 1\%,
the inflow FDI would increase by about 0.3\%, which is shown by Poisson QMLE
and Poisson GEE, and in which case the Poisson GEE estimate has smaller
standard error than Poisson QMLE, which are 0.102 and 0.110 respectively.

\begin{table}[tbp] \centering%
\caption{Descriptive statistics}

$%
\begin{tabular}{lllllll}
\hline\hline
Variables & Obs & $\text{Average}$ & Std.Dev. & Min & Max & Variable
description \\ \hline
FDI & 284 & 43571.94 & 99369.96 & 0 & 791954 & 10,000 dollars \\
$\ln $FDI & 275 & 9.28 & 1.81 & 3.14 & 13.58 &  \\
GDP & 287 & 9451788 & 1.31e+07 & 618352 & 1.20e+08 & 10,000 yuan \\
$\ln $GDP & 287 & 15.58 & 0.92 & 13.34 & 18.60 &  \\
GDPPC & 287 & 21566.76 & 16506.67 & 3398 & 98938 & yuan \\
$\ln $GDPPC & 287 & 9.76 & 0.65 & 8.13 & 11.50 &  \\
WAGE & 287 & 21228.01 & 5800.10 & 9523.21 & 49311.1 & yearly, yuan. \\
$\ln $WAGE & 287 & 9.93 & 0.25 & 9.16 & 10.81 &  \\
SCIEXP & 287 & 23513.22 & 91766.74 & 469 & 1100000 & 10,000 yuan \\
$\ln $SCIEXP & 287 & 8.86 & 1.25 & 6.15 & 13.91 &  \\
BORDER & 287 & 0.06 & 0.24 & 0 & 1 & =1 if on the border \\ \hline\hline
\end{tabular}%
$%
\label{summary}
\end{table}%

\begin{table}[tbp] \centering%
\caption{Grouping information}

$%
\begin{tabular}{lllllllll}
\hline\hline
Group & Province & Freq. & Percent &  & Group & Province & Freq. & Percent
\\ \hline
1 & Beijing & 1 & 0.35 &  & 17 & Henan & 17 & 5.92 \\
2 & Tianjin & 1 & 0.35 &  & 18 & Hubei & 12 & 4.18 \\
3 & Hebei & 11 & 3.83 &  & 19 & Hunan & 13 & 4.53 \\
4 & Shanxi & 11 & 3.83 &  & 20 & Guangdong & 21 & 7.32 \\
5 & Guangxi & 14 & 4.88 &  & 21 & Hainan & 2 & 0.70 \\
6 & Inner Mongolia & 9 & 3.14 &  & 22 & Chongqing & 1 & 0.35 \\
7 & Liaoning & 14 & 4.88 &  & 23 & Sichuan & 18 & 6.27 \\
8 & Jilin & 8 & 2.79 &  & 24 & Guizhou & 4 & 2.07 \\
9 & Heilongjiang & 12 & 4.18 &  & 25 & Yunnan & 8 & 2.76 \\
10 & Shanghai & 1 & 0.35 &  & 26 & Shaanxi & 10 & 3.45 \\
11 & Jiangsu & 13 & 4.53 &  & 27 & Gansu & 12 & 4.14 \\
12 & Zhejiang & 11 & 3.83 &  & 28 & Qinghai & 1 & 0.34 \\
13 & Anhui & 17 & 5.92 &  & 29 & Ningxia & 5 & 1.72 \\
14 & Fujian & 9 & 3.14 &  & 30 & Xinjiang & 2 & 0.69 \\
15 & Jiangxi & 11 & 3.83 &  & 31 & Tibet & 1 & 0.34 \\
16 & Shandong & 17 & 5.92 &  & Total &  & 287 & 1.00 \\ \hline\hline
\end{tabular}%
$%
\label{grouping}
\end{table}%

\begin{table}[tbp] \centering%
\caption{Estimating the FDI equation}

$%
\begin{tabular}{llllll}
\hline\hline
& OLS & Poisson & $%
\begin{array}{c}
\text{GEE \_poisson}%
\end{array}%
$ & NB & $%
\begin{array}{c}
\text{GEE \_nb2}%
\end{array}%
$ \\ \hline
$\ln $GDP & 1.099*** & 0.705*** & 0.746*** & 1.071*** & 0.982*** \\
& (0.188) & (0.151) & (0.132) & (0.205) & (0.176) \\
$\ln $GDPPC & 0.570*** & 0.747*** & 0.687*** & 0.610*** & 0.533*** \\
& (0.219) & (0.134) & (0.122) & (0.157) & (0.172) \\
$\ln $WAGE & -0.123 & -0.726* & -1.013*** & -0.146 & -0.111 \\
& (0.393) & (0.384) & (0.390) & (0.400) & (0.331) \\
$\ln $SCIEXP & 0.186 & 0.289*** & 0.311*** & 0.094 & 0.137 \\
& (0.142) & (0.110) & (0.102) & (0.111) & (0.106) \\
BORDER & -0.192 & -0.593*** & -0.197* & -0.556** & -0.037 \\
& (0.187) & (0.166) & (0.128) & (0.185) & (0.273) \\
\_cons & -13.894*** & -3.884 & -1.238 & -12.360*** & -11.021*** \\
& (3.670) & (3.094) & (3.011) & (3.130) & (2.863) \\ \hline
Observations & 275 & 284 & 284 & 284 & 284 \\
F(5, 269) & 152.03 &  &  &  &  \\
Wald Chi2(5) &  & 701.24 & 269.58 & 602.66 & 495.67 \\
p value & 0.000 & 0.000 & 0.000 & 0.000 & 0.000 \\ \hline\hline
\multicolumn{6}{l}{Note: Robust standard errors are in parentheses.} \\
\multicolumn{6}{l}{$^{\ast \ast \ast }$, $^{\ast \ast }$ and $^{\ast }$
indicate significance at the 1\%, 5\%, and 10\% level separately.}%
\end{tabular}%
$%
\label{fdiequ}
\end{table}%

\newpage
\section{Appendix}\label{sec7}

\subsection{Some Useful Lemmas}
We verify the $L_1$ NED property of $q_g(\theta,\gamma)$, $h_g(\theta,\gamma)$ accordingly via the $L_4$ NED property of $\by_g$, and the $L_2$ NED property of $s_g(\theta,\gamma)$ for central limit theorem.
\begin{lemma}\label{NED} \label{lem1:consistency}
Under condition A.1)- A.8), $q_{g}\left( \mathbf{\theta}, \mathbf{\gamma}\right) $ is
$L_1$ \textrm{NED }on $\tilde{\vps}$, with the NED constant as $d_g \defeq \max_{i\in B_g} d_{n,i},$ and with the NED coefficients $\psi(s).$
Moreover, we have ULLN for the partial sum $ \{M_{G}G\}^{-1}\sum_{g} q_g (\theta, \gamma)$, namely \\$\ssup {(M_{G} G)}^{-1}\sum_{g} \{q_g (\theta, \gamma) - \E  [\sum_{g } q_g (\theta^0, \gamma^0)]\}\to_p 0$.
\end{lemma}

\begin{proof}
We verify $q_g(\theta,\gamma)$ is $L_1$ NED on $\tilde{\vps}.$
From A.1), we work with increasing domain asymptotics, which essentially assume that the growth of the sample size is achieved by an unbounded expansion of the sample region. Namely $|D_G|=G \to \infty$.


The groupwise vector $\by_g$ satisfies $\|\by_{g} - \E (\by_{g}|\mathcal{F}_g(s))\|_{2} \leq \sum_{i\in B_g} d_{i,n} \psi(s)\leq d_g L \psi(s)$ ($d_g = \max_{i\in B_g} d_{i,n}$) for $s \to \infty $ and $\psi(s) \to 0$ when $s\to \infty$. We abbreviate $W_{g,ij}$ as an element of $\mathbf{W}_{g}(\gamma, \theta)$.
Thus $y_{g,i}$ is $L_2$ NED on $\tilde{\vps}$ by A.2).
As $\E|y_{gi} W_{g,ij} y_{gj} - \E\{y_{gi}| \mathcal{F}_g(s)\} W_{g,ij} \E\{y_{gj}| \mathcal{F}_g(s)\}| \leq  \|y_{gi}- \E\{y_{gi}| \mathcal{F}_g(s)\}\|_2 \|W_{g,ij}y_{gj}\|_2
\\+ \|y_{gj}- \E\{y_{gj}| \mathcal{F}_g(s)\}\|_2 \|{W}_{g,ij}y_{gi}\|_2 \leq C (d_{n,i}\vee d_{n,j})\psi(s),
$ by the fact that $\mathcal{F}_i(s) \subset \mathcal{F}_g(s)$ should hold for any $i \in B_g$.
Therefore we have
$\E| (\my_g - \mm_g)^{\top} \mathbf{W}_g (\my_g- \mm_g) - \E\{(\my_g - \mm_g)^{\top} \mathbf{W}_g (\my_g- \mm_g)| \mathcal{F}_g(s)\}| \leq \sum_i \sum_j\E |(y_g- m_g)_i W_{g,ij} (y_{g}- m_g)_j- \E\{(y_g- m_g)_i W_{g,ij} (y_{g}- m_g)_j| \mathcal{F}_g(s)\}|\leq \sum_i \sum_j \E|(y_g- m_g)_i W_{g,ij} (y_{g}- m_g)_j- \E\{(y_g- m_g)_i| \mathcal{F}_g(s)\} W_{g,ij} \E\{(y_{g}- m_g)_j| \mathcal{F}_g(s)\}| \leq C L^{2} d_g \psi(s) ,$ where $d_g = \max_{i\in B_g} d_{n,i}$ with $d_{n,i} = \CO(L).$

Given the $L_1-$ NED property of $q_g(\theta, \gamma)$ regarding the ULLN, we first look at a pointwise convergence of the function $q_g(.,.)$.
%
 We need to verify the following assumptions:
 \begin{itemize}
 \item[i)] There exists non random positive constants  $c_{g}, g \in D_n, n\geq 1$  such that for any $\theta, \gamma$, such that  $\E |q_{g}/c_{g}|^{p'} <
     \infty$, where $p'>1$.
 \item[ii)]The $\alpha-$ mixing coefficients of the input field $\vps$ satisfy $\tilde{\alpha}(u,v,r) \leq \psi(uL,vL) \hat{\alpha}(r),$ and for some $\hat{\alpha}(r)$, $\sum^{\infty}_{r=1} r^{d-1}L^{\tau} \hat{\alpha}(r)< \infty.$
 \end{itemize}

 Condition i) is implied by A.5) with the moment assumptions on objects involved in $q_g(\gamma, \theta)$ with $c_{g,q} = \CO(L^2)$. The reason is that $\E|q_g(\gamma, \theta)|^{p'} \leq \E\ssup |q_g(\gamma, \theta)|^{p'}.$
 For ii) we see that it is implied from A.6).

Moreover the uniform convergence needs in addition two assumptions:
\begin{itemize}
\item[i)] $p'-$ dominance assumption. There exists an array of positive real constants $%
\left \{ c_{g,q}\right \} $ such that $p  \geq 1$.%
\begin{equation}
\lim \sup_{G}\frac{1}{\left \vert D_{G}\right \vert }\sum_{g}\E\left( \mathbf{q}_{g}^{p'}1\left( \mathbf{q}_{g}>k\right) \right) \rightarrow 0\text{
as }k\rightarrow \infty,
\end{equation}%
where $\mathbf{q}_{g}=\sup_{\mathbf{\gamma \in \Gamma, \theta \in \Theta}}\left \vert q_{g}\left(
 \gamma, \theta \right)\right \vert /c_{g,q} .$ This is a revision form of
the domination condition as Assumption 6 in \cite{Jenish2009}.
Uniform boundedness of $q_{g}\left( \mathbf{\gamma, \theta }\right) $ is
covered by setting $c_{g,q}=\CO(L).$

\item[ii)] Stochastic equicontinuity. We assume that $q_{g}(\theta, \gamma)$ to be $L_0$ stochastic equicontinuity on $\Gamma\times \Theta$ iff $\lim_{G\to \infty} 1/ |D_G| \sum_{g\in D_G} \P(\sup_{(\gamma'\in \Gamma, \theta' \in \Theta ) \in B(\theta', \gamma', \delta)}|q_{g}(\gamma, \theta) - q_{g}(\gamma', \theta')|> \vps) \to 0 ,$ where $B(\theta', \gamma', \delta)$ is a $\delta-$ ball around the point $\gamma', \theta'$ with $\nu(\theta,\theta')\leq \delta$ and $\nu(\gamma,\gamma')\leq \delta$.
\end{itemize}
i) is implied by condition A.5). Namely we would like to prove the condition i), which is implied by the L $-s$ for any constant $s > p'$ boundedness of $\mathbf{q}_g$. Then we need to verify $\mbox{sup}_g\|\mathbf{q}_g\|_s< C$. As $\|\mathbf{q}_g\|_s = (\E|\mbox{sup}_{\theta \in \Theta, \gamma \in \Gamma}q_g(\theta,\gamma)|^s)^{1/s} \leq \sum_l\sum_g \E| \tilde{\vps}^s_{g,l}w^s_{g,l,m} \tilde{\vps}^s_{g,m}|\\\leq \sum_l\sum_g(\E \tilde{\vps}_{g,l}^{2s} \tilde{\vps}_{g,m}^{2s})^{1/(2s)} (\E w^{2s}_{g,l,m})^{1/(2s)},$ where $ \tilde{\vps}_{g,l} \defeq \sup_{\theta\in \Theta} \nabla \mathbf{m}_{g,l}(\theta)$, and $w_{g,l,m} \defeq \ssup w_{g,l,m}(\gamma, \theta).$
Therefore it can be seen that this will be implied by A.5) with $s<r/4$, with $c_{g,q} = \CO(L^2).$

The stochastic equicontinuity can be guaranteed by $q_g(\theta,\gamma)$ to be Lipschitz in parameter.
Namely for any $(\gamma, \theta)  \in (\Gamma, \Theta)$ and $ (\gamma', \theta' )\in (\Gamma, \Theta)$
\begin{equation}
|q_{g}(\gamma, \theta) - q_{g}(\gamma', \theta')| \leq B_{g1} g( \nu(\gamma, \gamma'))+ B_{g2}g(\nu(\theta,\theta')),
\end{equation}
where $g(s) \to 0$ when $s \to \infty$, and $B_{g1}, B_{g2}$ are random variables that do not depend on $\theta, \gamma$.
And $p'>0$,
\begin{equation}
\mbox{limsup}_{n \to \infty} (|D_G| |M_{G}|)^{-1}\sum_{g} \E |B_{gl}|_a^{p'}< \infty.
\end{equation}
To verify this
\begin{eqnarray}
&&|c_{g,q}^{-1}(\my_g -\mm_g(\theta))^{\top} \mathbf{W}_g^{-1}(\theta,\gamma) (\my_g- \mm_g(\theta)) - c_{g,q}^{-1}(\my_g - \mm_g(\theta'))^{\top}\mathbf{W}_g^{-1}(\theta',\gamma')(\my_g- \mm_g(\theta'))|\nonumber\\&&\leq
|c_{g,q}^{-1}\ssup s_g(\theta, \gamma)|_2|\theta-\theta'|_2+ |c_{g,q}^{-1}\ssup s_{g,\gamma}(\theta, \gamma)|_2|\gamma-\gamma'|_2,
\end{eqnarray}
where $s_{g,\gamma}(\theta, \gamma) \defeq  (\my_g-\mm_g(\theta))^{\top} \bigotimes (\my_g-\mm_g(\theta))^{\top} \partial (\mathbf{W}_g^{-1}(\theta, \gamma))/\partial \gamma$.

By A.5) we have $c_{g,q}^{-1}\||\ssup s_g(\theta, \gamma)|_2\|_{p'} = \CO(p),  c_{g,q}^{-1}\||\ssup s_{g,\gamma}(\theta, \gamma)|_2\|_{p'} = \CO(q) $.
So we have ii) and the desired results $\sup_{\theta \in \Theta, \gamma \in \Gamma}|\mathbf{Q}_{G}\left( \mathbf{\theta,\gamma}\right)_{i,j} - \overline{\mathbf{Q}}_{\infty,i,j}(\theta^0,\gamma^0)| \to \Co_p(1)$.

\end{proof}

\begin{lemma}\label{NED} \label{lem2:consistency}
Under condition A.1)- A.8), $s_{g}\left( \mathbf{\theta },\mathbf{\gamma}\right) $ is
$L_2$ \textrm{NED }on $\tilde{\vps}$, with the NED constant as $d_g = \max_{i\in B_g} d_{n,i},$ and with the NED coefficients $\psi(s).$
Moreover, we have a ULLN for the partial sums ${(M_{G}D_G)}^{-1}\sum_{g} s_g\left( \mathbf{\theta },\mathbf{\gamma}\right).$
\end{lemma}

\begin{proof}This proof is similarly proved as in lemma \ref{lem1:consistency}.
It can be seen that
$\| |\nabla \mm_g^{\top} \mathbf{W}_g(\theta,\gamma) (\my_g- \mm_g(\theta)) - \E\{\nabla \mm_g^{\top} \mathbf{W}_g(\theta,\gamma) (\my_g- \mm_g(\theta))|\mathcal{F}_g(s)\}|_2\|
 \leq \sum_i \sum_j\||\nabla m_{gi}^{\top} W_{ij}(\theta,\gamma) (y_{g}- m_g)_j\\- \E\{\nabla m_{gi}^{\top} W_{ij}(\theta,\gamma) (y_{g}- m_g)_j| \mathcal{F}_g(s)\}|_2\|\leq \sum_i \sum_j \||\nabla m_{gi}^{\top}
W_{ij} (\theta,\gamma)(y_{g}- m_g)_j- \nabla m_{gi}^{\top}W_{ij} (\theta,\gamma)\E\{(y_{g}- m_g)_j| \mathcal{F}_g(s)\}|_2\| \leq C' L^{2} d_g \psi(s) ,$ where $d_g = \max_{i\in B_g} d_{n,i}, $ and $ \mbox{max}_{i,j}\||\nabla m_{gi}^{\top}W_{ij} |_2\|_4 \lesssim C'pL$ according to A.5).
The $p'$ dominance assumption will be following from A.5) given the fact that $\mbox{sup}_g\E |\ssup s_g|^r< C,$ for $ p'< s< r/4.$  This would imply the uniform integrability.

Regarding the Lipschitz condition needed for the stochastic equicontinuity property $M_{G}^{-1}|\nabla \mm_g^{\top} \mathbf{W}_g(\theta,\gamma) (\my_g- \mm_g) - \nabla \mm_g^{\top}(\theta') \mathbf{W}_g(\theta',\gamma') (\my_g- \mm_g(\theta') )|\leq M_{G}^{-1}|h_g|_2 |\theta-\theta'|_2 + M_{G}^{-1}|H_{g,\gamma}|_2|\gamma-\gamma'|_2$, where $h_g \defeq \ssup h_g(\theta, \gamma)$ and $H_{g,\gamma} \defeq \ssup \partial s_g(\gamma, \theta)/\partial \gamma.$ The finiteness of $\sup_g\E(|H_g|_2^{p'}), \sup_g\E(|H_{g,\gamma}|_2^{p'}) $ will be implied by A.5).

\end{proof}

\begin{lemma}\label{NED} \label{lem3:consistency} \label{ULLNh}
Under condition A.1)- A.8), $h_{g}\left( \mathbf{\theta}, \mathbf{\gamma} \right)$ is
$L_1$ \textrm{NED }on $\tilde{\vps}$, with the NED constant as $d_g = \max_{i\in B_g} d_{n,i},$ and with the NED coefficient $\psi(s).$
Moreover, we have a ULLN for the partial sums ${(M_{G}D_G)}^{-1}\sum_{g} h_g\left( \mathbf{\theta },\mathbf{\gamma}\right).$
\end{lemma}

\begin{proof}
Now we verify the component involved in the partial sums in  $\mathbf{H}_{G}\left( \mathbf{\theta ,\hat{\gamma}}\right)$ are also $L_1$ NED on $\tilde{\vps}.$

Namely, $h_{1g} \defeq \nabla _{\mathbf{\theta }}\mathbf{m}_{g}^{\top }\left( \mathbf{\theta }\right) \mathbf{W}_{g}(\gamma, \theta)^{-1}\nabla
_{\mathbf{\theta }}\mathbf{m}_{g}\left( \mathbf{\theta }\right)$, $h_{2g} \defeq  [(\mathbf{y}_g-\mathbf{m}_g(\theta) )^{\top} \mathbf{W}_g(\gamma, \theta)^{-1} \otimes  I_q]\partial \mbox{Vec}(\nabla \mathbf{m}^{\top}_g)/\partial \theta $,
$h_{3g} \defeq \{(\mathbf{y}_g-\mathbf{m}_g(\theta))^{\top} \otimes \nabla \mathbf{m}^{\top}_g \}\partial \mbox{Vec}\{\mathbf{W}_g(\gamma,\theta)\}/\partial \theta.$
It is obvious that $h_{1g}$ is NED on $\tilde{\vps}$ as a measurable function of $\mx_g$. Define $e_i$ as a $p \times 1$ vector with only the $i-$th component as $1$, $|.|_a$ is taking the elementwise absolute value. And $b_{ij} \defeq e_i^{\top}(\mathbf{1}^{\top} W_g \otimes I_g )|\partial{\mbox{Vec}(\nabla\mm_g)}/\partial \theta|_a e_j.$ We verify now $h_{2g}$ for any fixed point $\gamma$ and $\theta$, it can seen that
$\E|h_{2g,i,j} -\E\{h_{2g,i,j}|\mathcal{F}_g(s)\}| \leq \E|e_i^{\top}([\{\by_{g,i} - \E (\by_{g,i} |\mathcal{F}_g(s))\}^{\top}   \mathbf{W}_g(\gamma, \theta)^{-1} \otimes  I_q]\partial \mbox{Vec}(\nabla \mathbf{m}^{\top}_g)/\partial \theta)e_j| \leq \E (\max_{i\in B_g}|y_{g,i} - \E [y_{g,i} |\mathcal{F}_g(s)]| b_{ij} |) \leq {L}^{1/2} \|\ssup b_{ij}\| d_{g} \psi(s)$, where for sufficiently large $s$ and $d_g = \CO({L}^{1/2} \|b_{ij}\| d_{g})$.  Therefore we proved the $L_1$ NED of $H_{2g}$.
Similarly for $H_{3g}$, define $c_{ij} =  e_i^{\top}(\mathbf{1}^{\top}\otimes \nabla \mm_g^{\top}(\theta) )|\partial{\mbox{Vec}(\nabla\mm_g(\theta) )}/\partial \theta|_a e_j$.
Then $\E|H_{2g,i,j} -\E\{H_{2g,i,j}|\mathcal{F}_l(s)\}| \leq {L_g}^{1/2} \|\ssup c_{ij}\| d_{g} \psi(s), $ where for sufficiently large $s$ and assume that ${L}^{1/2} \|\ssup c_{ij}\| d_{g} \psi(s)\to 0$.  We proved thus the $L_1$ NED of $H_{3g}.$
Then we would have the pointwise convergence of $\mathbf{H}_{G,1}(\theta, \gamma)$, $\mathbf{H}_{G,2}(\theta, \gamma)$, $\mathbf{H}_{G,3}(\theta,\gamma)$ any fixed point $\theta \in \Theta , \gamma \in \Gamma$.
To ensure that with probability $1-\Co_p(1),$ $|\mathbf{H}_{G}\left( \mathbf{\theta ,\hat{\gamma}}\right) - \mathbf{H}_{\infty}\left( \mathbf{\theta^0,\gamma^0}\right)|\leq \mbox{sup}_{\theta \in \Theta, \gamma \in \Gamma}|\mathbf{H}_{G}\left( \mathbf{\theta ,\gamma}\right) - \mathbf{H}_{\infty}\left( \mathbf{\theta^0,\gamma^0}\right)|\to 0$, therefore we need a ULLN.

Moreover the uniform convergence needs in addition two assumptions:
\begin{itemize}
\item[i)] There exists an array of positive real constants $%
\left \{ c_{g,h}\right \} $ such that  for constant $\delta> 0:$%
\begin{equation}
\mbox{limsup}_{G}\frac{1}{|D_{G}|}\sum_{g}E\left( \mathbf{H}_{l,g,i,j}^{2+\delta}\IF\left( \mathbf{H}_{l,g,i,j}>k\right) \right) \rightarrow 0\text{
as }k\rightarrow \infty ,
\end{equation}%
where $\mathbf{H}_{l,g,i,j}=\sup_{\mathbf{\theta \in \Theta, \gamma \in \Gamma }}\left \vert h_{gl,i,j}\left(
\mathbf{\theta }, \gamma\right) /  c_{g,h}\right \vert_1
 .$ This is a again revision form of
the domination condition as Assumption 6 in \cite{Jenish2009}.
Uniform boundedness of $\mathbf{H}_{g}\left( \mathbf{\theta }\right) $ is
covered by setting $c_{g,h}=\CO(L^2).$ $l = 1,2,3.$
\item[ii)] Stochastic equicontinuity. We assume that $H_{l,g}(\theta, \gamma)$ to be $L_0$ stochastic equicontinuity on $\Gamma$ iff $\lim_{G\to \infty} 1/ |D_G| \sum_{g} \P(\sup_{(\gamma'\in \Gamma, \theta'\in \Theta)  \in B(\gamma', \theta',\delta)}|H_{g,i,j}(\gamma, \theta) \\- H_{g,i,j}(\gamma',\theta')|> \vps) \to 0 .$
\end{itemize}
The stochastic equicontinuity can be guaranteed by $h_{g,i,j}(\theta,\gamma)$ to be Lipschitz in parameter, which is ensured by A.5).

Then we have $\sup_{\gamma \in \Gamma, \theta \in \Theta}|\mathbf{H}_{G,i,j}\left( \mathbf{\theta,\gamma}\right) - \mathbf{H}_{\infty,i,j}(\gamma, \theta)]| \to \Co_p(1)$.
\end{proof}

\subsection{Proof of Theorem \ref{th:consistency}}

Two sufficient conditions for consistent estimators are
i) identification implied by A.8) and ii) the objective
function $Q_G(\theta,\gamma)$ satisfies the uniform law of large numbers (ULLN).
By Lemma \ref{lem1:consistency}, we have the uniform LLN of $Q_G(\theta, \gamma)$.
%

Namely  $\mathbf{\theta \in \Theta , \gamma \in \Gamma}$, $\mbox{sup}_{\theta \in \Theta, \gamma \in \Gamma}\frac{1}{M_{G}\left \vert D_{G}\right \vert }%
\left[ Q_{G}\left( \mathbf{\theta }, \mathbf{\gamma}\right) -Q_{\infty}\left( \mathbf{\theta^0
}, \mathbf{\gamma}^0\right) \right] $ $\overset{p}{\rightarrow }0,$ as $G\rightarrow \infty
.$

 Thus we conclude that under A.1)-A.8), the GEE estimator is consistent.

\subsection{Proof of Theorem \ref{th:normality}}

\subsubsection{Step 1 : Main expansion step} \label{expansion}

%
Recall $\mathbf{\mu}_g = \my_g - \mm_g(\theta^0)$ and $\mathbf{\hat{\mu}}_g = \my_g- \mm_g(\hat{\theta}) $
\begin{equation}
\mathbf{S}_{G}\left( \mathbf{\theta ,\hat{\gamma}}\right) =\frac{1}{M_{G}G}%
\sum_{g}\nabla \mathbf{m}_{g}^{\top }\left( \mathbf{\theta}\right)
\mathbf{W}_{g}^{-1}\left( \mathbf{\hat{\gamma}},\theta\right) \left[ \mathbf{y}_{g}-%
\mathbf{m}_{g}\left( \mathbf{\theta }\right) \right] .  \label{quasiscore}
\end{equation}%

From the first order condition A.11).
\begin{equation*}
\mathbf{S}_{G}\left( \mathbf{\hat{\gamma}, \hat{\theta}} \right) = \Co_p(1).
\end{equation*}

To expand $\mathbf{S}_{G}\left( \mathbf{\hat{\gamma}, \hat{\theta}} \right) $ around the point
$\gamma^0, \theta^0$, we have,
\begin{eqnarray*}
&&\mathbf{S}_{G}\left( \mathbf{\hat{\gamma}, \hat{\theta}} \right) =\mathbf{S}_{G}\left( \mathbf{\gamma^0, \theta^0} \right) + \mathbf{H}_G(\tilde{\theta}, \tilde{\gamma}) (\hat{\theta}-\theta^0)+
\nabla_{\gamma} \mathbf{S}_G(\tilde{\gamma},\tilde{\theta})(\hat{\gamma} - \gamma^0)\\
&&=\mathbf{S}_{G}\left( \mathbf{\gamma^0, \theta^0} \right) + \mathbf{H}_{\infty}(\theta^0, \gamma^0) (\hat{\theta}-\theta^0)+
\mathbf{F}_0(\hat{\gamma} - \gamma^0)\\&&+ \{\mathbf{H}_G(\tilde{\theta}, \tilde{\gamma}) - \mathbf{H}_{\infty}(\theta^0, \gamma^0)\}(\hat{\theta}-\theta^0)+ \{\nabla_{\gamma} \mathbf{S}_G(\tilde{\theta}, \tilde{\gamma})- \mathbf{F}_0\}(\hat{\gamma} - \gamma^0)
\end{eqnarray*}
where $\tilde{\theta}, \tilde{\gamma}$ lie in the line segment between $\theta^0,\gamma^0$ to $\hat{\theta},\hat{\gamma}$, $\mathbf{F}_{0}$ is a $L\times q$ matrix, $\mathbf{F}%
_{0}=\lim_{G\rightarrow \infty }\left \{ \frac{1}{M_{G}|D_G|}\sum_{g}\E%
\left[ \nabla _{\gamma }s_{g}\left( \mathbf{\theta }%
^{0};\mathbf{\gamma }^{0}\right) \right] \right \} .$  From the derivation below we see that $\mathbf{F}_{0}=%
\mathbf{0},$ the asymptotic distribution of the average score does not
depend on the distribution of $\mathbf{\hat{\gamma},}$ and the first-step
estimation of $\mathbf{\hat{\gamma}}$ will not affect the second-step
estimation in terms of asymptotic variance.

%

$\mathbf{F}_0$ is the the limit of orthogonal score by construction. To identify this, we can see that
$\nabla_{\gamma} \{\nabla \mathbf{m}_{g}^{\top }\left(
\mathbf{\theta}^0\right) \mathbf{W}_{g}^{-1}(\theta^0,\gamma^0)\left[ \mathbf{y}_{g}-%
\mathbf{m}_{g}\left( \mathbf{\theta}^0\right) \right]\} \\= \{\mathbf{y}_{g}-%
\mathbf{m}_{g}\left( \mathbf{\theta}^0\right) \}^{\top} \otimes \nabla \mathbf{m}^{\top}_{g} (\theta^0) \nabla_{\gamma} \mbox{Vec}\{\mathbf{W}_{g}^{-1}(\theta^0, \gamma^0)\}
$.

$\left[ \mathbf{y}_{g}-%
\mathbf{m}_{g}\left( \mathbf{\theta}^0\right) \right] ^{\top} \otimes \nabla \mathbf{m}^{\top}_{g} (\theta^0) \nabla_{\gamma} \mbox{Vec}\{\mathbf{W}_{g}^{-1}(\theta^0, \gamma^0)\} \\= \E[\E [\{\mathbf{y}_{g}-%
\mathbf{m}_{g}\left( \mathbf{\theta}^0\right) \} ^{\top}|\mx_g] \otimes \nabla \mathbf{m}^{\top}_{g} (\theta
^0)  \nabla_{\gamma} \mbox{Vec}\{\mathbf{W}_{g}^{-1}(\theta^0, \gamma^0)\}] = 0$.

To handle the term $ \{\mathbf{H}_G(\tilde{\theta}, \tilde{\gamma}) - \mathbf{H}_{\infty}(\theta^0, \gamma^0)\}(\hat{\theta}-\theta^0)+ \{\nabla_{\gamma} \mathbf{S}^{\top}_G(\tilde{\theta}, \tilde{\gamma})- \mathbf{F}^{\top}_0\}(\hat{\gamma} - \gamma^0)$, we need the ULLN for $\mathbf{H}_G(\theta^0, \gamma^0)$ to derive $|e_i^{\top}(\mathbf{H}_G( \tilde{\theta}, \tilde{\gamma}) - \mathbf{H}_{\infty}(\theta^0, \gamma^0))e_j|\leq  \mbox{sup}_{\theta,\gamma}|e_i^{\top}(\mathbf{H}_G(\theta, \gamma) - \mathbf{H}_{\infty}(\theta^0, \gamma^0))e_j|\to_p 0.$ Also for
$\nabla_{\gamma} \mathbf{S}_G( \tilde{\theta}, \gamma)$ to derive\\ $|e_i^{\top}(\{\nabla_{\gamma} \mathbf{S}^{\top}_G( \tilde{\theta}, \gamma)- \mathbf{F}^{\top}_0\})e_j|
\leq  \ssup|e_i^{\top}\{\nabla_{\gamma} \mathbf{S}^{\top}_G(\theta, \gamma) - \mathbf{F}^{\top}
_0\}e_j|\to_p 0.$ This is already verified by Lemma \ref{ULLNh}.
We arrive at the conclusion that for any vector $a\in \mathcal{R}^p, |a|_2 = 1$, $|a^{\top}\{\mathbf{H}_G(\tilde{\theta}, \tilde{\gamma}) - a^{\top}\mathbf{H}_{\infty}(\theta^0, \gamma^0)\}(\hat{\theta}-\theta^0)| \leq |\{a^{\top}\mathbf{H}_G(\tilde{\theta}, \tilde{\gamma}) - a^{\top}\mathbf{H}_{\infty}(\theta^0, \gamma^0)\}|_2|(\hat{\theta}-\theta^0)|_2 = \Co_p(1)\times\CO_p(|(\hat{\theta}-\theta^0)|_2 ) =  \Co_p(|(\hat{\theta}-\theta^0)|_2 )$ and $|a^{\top}\{\nabla_{\gamma} \mathbf{S}^{\top}_G(\tilde{\theta}, \tilde{\gamma})- \mathbf{F}^{\top}_0\}(\hat{\gamma} - \gamma^0)|_2  = \Co_p(|(\hat{\gamma}-\gamma^0)|_2 ) = \Co_p(G^{-1/2})$ by A.8).

Next we look at the invertibility of the matrix $\mathbf{H}_{\infty}(\theta^0,\gamma^0)$.
Taking the expected value of the score function over
the distribution of $\left( \mathbf{x}_{g}\mathbf{,y}%
_{g}\right) $ gives
\begin{eqnarray*}
\E\left[ h_{g}\left( \mathbf{\theta }^{0},%
\mathbf{\gamma }^{0}\right) \right]  &=&\E[\E%
[ \mathbf{h}_{g}\left( \mathbf{w}_{g},\mathbf{\theta }^{0},\mathbf{%
\gamma }^{0}\right) |\mathbf{x}_{g}] ]
\\&=&\E[\{(\mathbf{y}_g-\mathbf{m}_g(\theta^0))^{\top} \otimes \nabla \mathbf{m}^{\top}_g (\theta^0)\}\partial \mbox{Vec}\{\mathbf{{W}}_g\left(\mathbf{\theta }^{0},%
\mathbf{\gamma }^{0}\right)\}/\partial \theta ]
\\&&-\E[\nabla\mathbf{m}%
_{g}^{\top }(\theta^0)\mathbf{{W}}_{g}^{-1}(\theta^0,\gamma^0)\nabla
_{\mathbf{\theta }}\mathbf{m}_{g}(\theta^0) ]
\\&&+\E[[  \{(\mathbf{y}_g-\mathbf{m}_g(\theta^0))^{\top} \mathbf{{W}}_g^{-1}\left(\mathbf{\theta }^{0},\mathbf{\gamma }^{0}\right) \otimes  I_q\}]\partial \mbox{Vec}(\nabla \mathbf{m}^{\top}_g)/\partial \theta  ]\\
&=&\E[ -\nabla \mathbf{m}_{g}(\theta^0)^{\top }\mathbf{W}_{g}^{
-1}\left(\mathbf{\theta }^{0},%
\mathbf{\gamma }^{0}\right)\nabla \mathbf{m}_{g}(\theta^0)] \\&&+ \E[\{\E[(\mathbf{y}_g-\mathbf{m}_g(\theta^0))^{\top}|\mathbf{x}_{g}] \otimes \nabla \mathbf{m}^{\top}_g \}\partial \mbox{Vec}(\mathbf{{W}}_g(\mathbf{\theta }^{0},%
\mathbf{\gamma }^{0}))/\partial \theta ]\\&&
+\E[  \{\E[\{(\mathbf{y}_g-\mathbf{m}_g(\theta^0))^{\top}|\mathbf{x}_{g}\} \mathbf{{W}}_g^{-1}\left( \mathbf{\theta }^{0},%
\mathbf{\gamma }^{0}\right) \otimes  I_q]\}\partial \mbox{Vec}(\nabla \mathbf{m}^{\top}_g)(\theta^0)/\partial \theta ] \\
&=&\E[ -\nabla \mathbf{m}_{g}^{\top }(\theta^0)\mathbf{W}_{g}^{
-1}\left( \mathbf{\theta }^{0},%
\mathbf{\gamma }^{0}\right)\nabla \mathbf{m}_{g}(\theta^0)],
\end{eqnarray*}%
which is negative definite by assumption A.9).

The GEE estimator can be specifically written as%
\begin{equation}\label{eq:gee}
\sqrt{G}\left( \mathbf{\hat{\theta}}\mathbf{-\theta }%
_{0}\right) =\left[\mathbf{H}_{\infty}(\theta^0, \gamma^0) \right] ^{-1}\frac{1}{\sqrt{G%
}}\sum_{g}s_{g}\left(\mathbf{\theta }^{0},\mathbf{{%
\gamma^0}}\right)+ \Co_p(1)+ \Co_p(\sqrt{G}|\hat{\theta} - \theta^0|_2).
\end{equation}

Due to the $L_2$ NED property of $s_g$, $\Var(\sum^G_{g=1}s_{g}) = \CO(G )$, thus


we have $\sqrt{G}|\hat{\theta} - \theta^0|_2 \lesssim|\mathbf{H}_{\infty}(\theta^0, \gamma^0)^{-1}|_2  = \CO_p(CM_G^{2})$, as the order of $\frac{1}{\sqrt{G%
}}\sum_{g}s_{g}\left( \mathbf{w}_{g},\mathbf{\theta }^{0};\mathbf{{%
\gamma^0}}\right)$ under assumption B.3) is $\CO_p(G^{-1/2})$.
This implies that $ \Co_p(\sqrt{G}|\hat{\theta} - \theta^0|_2) = \CO_p(1).$

\subsubsection{Step 2 Central Limit Theorem}
We derive the variance of $s_{g}\left( \mathbf{w}_{g},\mathbf{%
\theta }^{0}\mathbf{,\gamma }^{0 }\right) $ in this subsection.
\begin{eqnarray*}
&&AS_G = \mathrm{Var}\left[ \frac{1}{\sqrt{G}}\sum_{g}s_{g}\left(
\mathbf{w}_{g},\mathbf{\theta }^{0}\mathbf{,\gamma }^{0 }\right) \right]\\
&=&\mathrm{Var}\left \{ \frac{1}{\sqrt{G}}\sum_{g}\nabla \mathbf{m}%
_{g}^{\top }\left( \mathbf{\theta }^{0}\right) \mathbf{W}_{g}^{-1}\left(
 \mathbf{\theta }^{0},\mathbf{\gamma }^{0 }\right) \left[ \mathbf{y}_{g}-\mathbf{m}_{g}\left(
\mathbf{\theta }^{0}\right) \right] \right \}  \\
&=&\mathrm{Var}\left[ \frac{1}{\sqrt{G}}\sum_{g}\nabla \mathbf{m}%
_{g}^{\top}\left( \mathbf{\theta }^{0}\right) \mathbf{W}_{g}^{-1}\left( \mathbf{\theta }^{0},
\mathbf{\gamma }^{0 }\right) \mathbf{u}_{g}\right]  \\
&=&\frac{1}{G}\sum_{g}\E\left[ \nabla \mathbf{m}_{g}^{\top
}\left( \mathbf{\theta }^{0}\right) \mathbf{W}_{g}^{-1}\left( \mathbf{\theta }^{0},\mathbf{\gamma
}^{0 }\right) \mathbf{u}_{g}\mathbf{u}_{g}^{\top }\mathbf{W}%
_{g}^{-1}\left( \mathbf{\theta }^{0},\mathbf{\gamma }^{0 }\right) \nabla \mathbf{m}_{g}\left(
\mathbf{\theta }^{0}\right) \right]  \\
&&+\frac{1}{G}\sum_{g}\sum_{h, h\neq g}\E\left[ \nabla
\mathbf{m}_{g}^{\top }\left( \mathbf{\theta }^{0}\right) \mathbf{W}%
_{g}^{-1}\left( \mathbf{\theta }^{0}, \mathbf{\gamma }^{0 }\right) \mathbf{u}_{g}\mathbf{u}^{\top}_{h}%
\mathbf{W}_{h}^{-1}\left(\mathbf{\theta }^{0}, \mathbf{\gamma }^{0 }\right) \nabla \mathbf{m}%
_{h}\left( \mathbf{\theta }^{0}\right) \right] .
\end{eqnarray*}

The next step is to apply the central limit theorem (Corollary 1 in \cite{Jenish2012}) the element $\mathbf{S}_G = \frac{1}{\sqrt{G}}\sum_{g}%
s_{g}\left( \mathbf{w}_{g},\mathbf{\theta }^{0}\mathbf{,\gamma }%
^{0 }\right),$ and $AS_{\infty} = \lim_{G\to \infty} AS_G$.
For that we need to verify the following conditions:
\begin{itemize}
\item[i)]   $s_g$ is uniform $L_2$ NED on the $\alpha-$ mixing random field $\tilde{\vps}$ with  coefficients $d_gL$ and $\psi(s)$, $\sup_{G,g} d_gL < \infty$  and $\sum^{\infty}_{r=1}r^{d-1}\psi(r)<\infty$. Moreover $\mbox{sup}_{G}\mbox{sup}_{g} \|s_g\|_r,$ where $r > 2+\delta'$, with $\delta'$ as a constant.
\item[ii)]  The input field $\tilde{\vps}$ is $\alpha-$ mixing with coefficient $\sum^{\infty}_{r =1}r^{(d \tau^*+d)-1}L^{\tau^*} \hat{\alpha}^{\delta/(2+\delta')}(r)< \infty.$ ($\tau^* = \delta '\tau /(4+2\delta')$)
\item[iii)] $\inf_G |D_G|^{-1} M_G^{-2}  \lambda_{min}(\mathbf{AS}_{\infty})> 0.$  (suppressed $G$ for the triangular array.)
\end{itemize}
i) is proved in  Lemma \ref{lem2:consistency}, ii) can be inferred by A.11), and iii) can be inferred from A.10).
Therefore under A.1)-A.11)
 \begin{equation}
 \mathbf{AS}_{\infty} ^{-1/2} \mathbf{S}_G \Rightarrow \mathbb{N}(0,I_p).
\end{equation}

So we have  $AV(\hat{\theta}) = \mathbf{H}_{\infty}^{\top} \mathbf{AS}_{\infty} \mathbf{H}_{\infty}$
\begin{equation}
\sqrt{G}AV(\hat{\theta})^{-1/2}(\hat{\theta} - \theta^0) \Rightarrow \mathbb{N}(0,I_p).
\end{equation}

\subsection{Proof of Proposition \ref{pre}}

\begin{itemize}
\item[A.8)'](Identifiability)
$\overline{E}_{G}\left( \mathbf{\theta }, \mathbf{\gamma}\right)\defeq \sum_{g}(\mathbf{e}_{g}(\check{\theta})-\mathbf{z}%
_{g}(\gamma))^{\top }(\mathbf{e}_{g}(\check{\theta})-\mathbf{z}_{g}(\gamma))$.
 And $E_{\infty}(\gamma, \theta)\defeq \lim_{G\rightarrow \infty }\overline{E}_{G}\left( \mathbf{\gamma,
\theta }\right) .$  Assume that $\theta^0, \gamma^0$ are identified unique in a sense that \\ $\liminf_{G\to \infty}\mbox{inf}_{\gamma \in \Gamma: \nu(\gamma, \gamma^0) \geq \vps }\overline{E}_{G}\left( \mathbf{\theta}, \mathbf{\gamma}\right) > c_0> 0$, for a positive constant $c_0$.


\item[A.9)'] The true point $\theta^0, \gamma^0$ lies in the interior point of $\Theta, \Gamma$. $\check{\theta}$ is estimated with $|\check{\theta} - \theta^0|_2 = \CO_p(G^{-1/2}).$\\

\item[A.11)'] $(|D_G|M_G)^{-1}\sum_{g} \sum_l\sum_{m<l} (\mathbf{e}_{glm}(\check{\theta})-\mathbf{z}%
_{glm}(\hat{\gamma})) \partial{\mathbf{z}%
_{glm}(\hat{\gamma})}/\partial \gamma = \Co_p(1)$.
\end{itemize}

In this subsection, we verify the consistency of the preestimator $\hat{\gamma}$.
As we have
\begin{equation}
\mathbf{\hat{\gamma}}=\arg \min_{\gamma} \sum_{g}(\mathbf{e}_{g}(\check{\theta})-\mathbf{z}%
_{g}(\gamma))^{\top }(\mathbf{e}_{g}(\check{\theta})-\mathbf{z}_{g}(\gamma)),
\end{equation}
which leads to $\arg \mathbf{zero}_{\gamma \in \Gamma} \sum_{g} \sum_l\sum_{m<l} (\mathbf{e}_{glm}(\check{\theta})-\mathbf{z}%
_{glm}(\gamma)) \partial{\mathbf{z}%
_{glm}(\gamma)}/\partial \gamma = 0$.

We can proceed with a similar expansion step  as in Section \ref{expansion}.
Therefore $\sum_{g} \sum_l\sum_{m<l} \{\mathbf{e}_{glm}(\check{\theta})-\mathbf{z}%
_{glm}(\hat{\gamma})\} \partial{\mathbf{z}%
_{glm}(\hat{\gamma})}/\partial \gamma \\= \sum_{g} \sum_l\sum_{m<l} \{\mathbf{e}_{glm}(\theta^0)-\mathbf{z}%
_{glm}(\gamma^0)\} \partial{\mathbf{z}%
_{glm}(\gamma^0)}/\partial \gamma\\ +  \sum_{g} \sum_l\sum_{m<l} \{\mathbf{e}_{glm}(\tilde{\theta})-\mathbf{z}%
_{glm}(\tilde{\gamma})\} \partial{\mathbf{z}%
_{glm}(\tilde{\gamma})}/{\partial \gamma\partial \gamma^{\top}}(\hat{\gamma}- \gamma^0)\\- \sum_{g} \sum_l\sum_{m<l} \{\partial \mathbf{z}%
_{glm}(\tilde{\gamma})/\partial \gamma\} \partial{\mathbf{z}%
_{glm}(\tilde{\gamma})}/{\partial \gamma^{\top}}(\hat{\gamma}- \gamma^0) \\ + \sum_{g} \sum_l\sum_{m<l} \{\partial\mathbf{z}%
_{glm}(\tilde{\gamma})/\partial \gamma\} \partial\mathbf{e}_{glm}(\tilde{\theta})/\partial \theta^{\top} (\check{\theta} - \theta^0),$
where $\tilde{\gamma}, \tilde{\theta}$ lies in the line segment between $\theta^0, \gamma^0$ and $\check{\theta}, \hat{\gamma}$.

It is known that under proper NED assumptions a pooled estimation $\check{\theta}$ satisfying $|\check{\theta}- \theta^0|_2  = \CO_p(1/\sqrt{n}).$
The verification step would be similar to the proof in Section \ref{expansion}, where we also need ULLN for the term $G^{-1}\sum_{g} \sum_l\sum_{m<l} \partial \mathbf{z}%
_{glm}(\tilde{\gamma})/\partial \gamma \partial{\mathbf{z}%
_{glm}(\tilde{\gamma})}/{\partial \gamma^{\top}}$, $2 G^{-1} \sum_{g} \sum_l\sum_{m<l}\{ \mathbf{e}_{glm}(\tilde{\theta})-\mathbf{z}%
_{glm}(\tilde{\gamma}) \}\partial\mathbf{e}_{glm}(\tilde{\theta})/\partial \theta$ and $G^{-1}\sum_{g} \sum_l\sum_{m<l} \{\mathbf{e}_{glm}(\tilde{\theta})-\mathbf{z}%
_{glm}(\tilde{\gamma})\} \partial{\mathbf{z}%
_{glm}(\tilde{\gamma})}/{\partial \gamma\partial \gamma^{\top}}$.
This will lead to $\sum_{g} \sum_l\sum_{m<l} \{\mathbf{e}_{glm}(\theta^0)-\mathbf{z}%
_{glm}(\gamma^0)\} \partial{\mathbf{z}%
_{glm}(\gamma^0)}/\partial \gamma = \CO_p(\sqrt{G}).$ (Lemma A.3 in \cite{Jenish2012}).

The desired results now follows from condition A.1) - A.3), A.5), A.6) and A.8)', A9)', A11)'.

\subsection{Proof of Theorem \ref{th:variance}}

We prove that $\mbox{sup}_{(\gamma, \theta) \in (\Gamma, \Theta)} e_i^{\top}\hat{\mathbf{A}}(\theta, \gamma)e_j \to_p e_i^{\top}\mathbf{A}_0 e_j$,\\ and $\mbox{sup}_{(\gamma, \theta) \in (\Gamma, \Theta)} e_i^{\top}\mathbf{\hat{B}}(\theta, \gamma)e_j \to_p e_i^{\top}\mathbf{B}_0 e_j$.
And by the Slutsky's theorem the variance covariance estimation is consistent.
 Firstly we prove that $e_i^{\top}(\mathbf{\hat{A}}-  \mathbf{A}_{0})e_j \to_p 0 $. This is implied by uniform law of large numbers for near-epoch dependent sequences, as mentioned the NED property of the underlying sequence $(\mathbf{x}_g)$ is trivial under condition A.1) - A.5) as it is a measurable function of the input field $\tilde{\vps}$.

$
e_i^{\top}\mathbf{\hat{A}}e_j=\frac{1}{G M_G}\sum_{g}e_i^{\top}\nabla \mathbf{\hat{m}}%
_{g}^{\top}\mathbf{\hat{W}}_{g}^{-1}\nabla \mathbf{\hat{m}}%
_{g}e_j\rightarrow_p \\\lim_{G\rightarrow \infty }\frac{1}{G M_G}\sum_{g}e_i^{\top}\E%
\left( \nabla \mathbf{m}_{g}^{\top}\mathbf{W}_{g}^{-1}\nabla \mathbf{m}%
_{g}\right)e_j
 =e_i^{\top}\mathbf{A}_{0}e_j.$

We still need to prove that $e_i^{\top}\mathbf{\hat{B}}e_j\rightarrow_p e_i^{\top}\mathbf{B}_{0}e_j.$
We denote $\mathbf{W}_g = \mathbf{W}_g(\theta^0,\gamma^0)$ and $\hat{\mathbf{W}}_g = \mathbf{W}_g(\hat{\theta},\hat{\gamma})$.

Recall that $Z_{g}\defeq  \nabla \mathbf{m}^{\top}_{g}\mathbf{W}_{g}^{-1}\mathbf{u}_{g},$
and $\hat{Z}_{g}\defeq \nabla \mathbf{\hat{m}}_{g}^{\top}\mathbf{\hat{W}}%
_{g}^{-1}\mathbf{\hat{u}}_{g}.$ %
\begin{eqnarray*}
\mathbf{B}_{0} &=&\lim_{G\rightarrow \infty }\mathrm{Var}\left[ \frac{1}{%
\sqrt{M_G^2|D_G|}}\sum_{g}s_{g}\left( \mathbf{\theta }^{0},\mathbf{%
\gamma }^{0 }\right) \right] \\
&=&\lim_{G\rightarrow \infty }\frac{1}{M_G^2|D_G|}\sum_{g}\E\left[ \nabla \mathbf{m}_{g}^{\top }%
\mathbf{W}_{g}^{-1}\mathbf{u}_{g}\mathbf{u}_{g}^{\top }\mathbf{W}%
_{g}^{-1}\nabla \mathbf{m}_{g}\right]
\\&&+\frac{1}{M_G^2|D_G|}\sum_{g}\sum_{h(\neq g)}\E\left[ \nabla \mathbf{m}_{g}^{\top }\mathbf{W}_{g}^{-1}%
\mathbf{u}_{g}\mathbf{u}_{h}^{\top }\mathbf{W}_{h}^{-1}\nabla \mathbf{m}%
_{h}\right] \\
&=&\lim_{G\rightarrow \infty }\frac{1}{M_G^2|D_G|}\sum_{g}\E\left[ Z_{g}^{^{\top }}Z_{g}\right] +%
\frac{1}{M_G^2|D_G|}\sum_{g}\sum_{h(\neq g) \in D_G}\E\left[ Z_{g}^{^{\top
}}Z_{h}\right] . \\
\mathbf{\hat{B}} &=&\frac{1}{M_G^2|D_G|}\sum_{g}\sum_{h(\neq g)}k(d_{gh})\nabla
\mathbf{
\hat{m}}_{g}^{\top }\mathbf{\hat{W}}_{g}^{-1}\mathbf{\hat{u}}_{g}%
\mathbf{\hat{u}}_{h}^{\top }\mathbf{\hat{W}}_{h}^{-1}\nabla \mathbf{\hat{m}%
}_{h}, \\
&\mathbf{=}&\frac{1}{M_G^2|D_G|}\sum_{g}\nabla \mathbf{\hat{m}}_{g}^{\top }%
\mathbf{\hat{W}}_{g}^{-1}\mathbf{\hat{u}}_{g}\mathbf{\hat{u}}_{g}^{\top }%
\mathbf{\hat{W}}_{g}^{-1}\nabla \mathbf{\hat{m}}_{g}
\\&&+\frac{1}{M_G^2|D_G|}%
\sum_{g}\sum_{h(\neq g) }k(d_{gh})\nabla \mathbf{\hat{m}}_{h}\mathbf{%
\hat{W}}_{g}^{-1}\mathbf{\hat{u}}_{g}\mathbf{\hat{u}}_{h}^{\top }\mathbf{%
\hat{W}}_{h}^{-1}\nabla \mathbf{\hat{m}}_{h} \\
&=&\frac{1}{M_G^2|D_G|}\sum_{g}\hat{Z}_{g}^{^{\top }}\hat{Z}_{g}+\frac{1}{M_G^2|D_G|}%
\sum_{g}\sum_{h(\neq g)}k(d_{gh})\hat{Z}_{g}^{^{\top }}\hat{Z}_{h}.
\end{eqnarray*}%
Define $\mathbf{B}_{0}^{k}$ and $\mathbf{B}^{k}$ as
\begin{eqnarray*}
\mathbf{B}_{0}^{k} &\mathbf{=}&\frac{1}{M_G^2|D_G|}\sum_{g}\E\left[
\nabla \mathbf{m}_{g}^{\top }\mathbf{W}_{g}^{-1}\mathbf{u}_{g}\mathbf{u}%
_{g}^{\top }\mathbf{W}_{g}^{-1}\nabla \mathbf{m}_{g}\right] \\
&&+\frac{1}{M_G^2|D_G|}%
\sum_{g}\sum_{ h(\neq g)}k(d_{gh})\E\left[ \nabla \mathbf{m}%
_{g}^{\top }\mathbf{W}_{g}^{-1}\mathbf{u}_{g}\mathbf{u}_{h}^{\top }%
\mathbf{W}_{h}^{-1}\nabla \mathbf{m}_{h}\right] \\
&=&\frac{1}{M_G^2|D_G|}\sum_{g}\E\left( Z_{g}^{^{\top }}Z_{g}\right) +%
\frac{1}{M_G^2|D_G|}\sum_{g\in D_Gy}\sum_{h(\neq g)}k(d_{gh})\E\left(
Z_{g}^{^{\top }}Z_{h}\right) . \\
\mathbf{B}^{k} &\mathbf{=}&\frac{1}{M_G^2|D_G|}\sum_{h(\neq g)}\left[ \nabla \mathbf{m}%
_{g}^{\top }\mathbf{W}_{g}^{-1}\mathbf{u}_{g}\mathbf{u}_{g}^{\top }%
\mathbf{W}_{g}^{-1}\nabla \mathbf{m}_{g}\right]\\&& +\frac{1}{M_G^2|D_G|}%
\sum_{g}\sum_{h(\neq g)}k(d_{gh})\left[ \nabla \mathbf{m}%
_{g}^{\top }\mathbf{W}_{g}^{-1}\mathbf{u}_{g}\mathbf{u}_{h}^{\top }%
\mathbf{W}_{h}^{-1}\nabla \mathbf{m}_{h}\right] \\
&=&\frac{1}{M_G^2|D_G|}\sum_{g}Z_{g}^{^{\top }}Z_{g}+\frac{1}{M_G^2|D_G|}%
\sum_{g}\sum_{h(\neq g)}k(d_{gh})Z_{g}^{^{\top }}Z_{h}.
\end{eqnarray*}%
Next write the estimation error for $\mathbf{B}_0$ as in three parts, namely the
part consists of generated errors ($I_1$), the variance ($I_2$) and the bias part ($I_3$).
We need to prove that the generated error term is negligible, the variance term is small induced by the
property NED, and the bias term is also small.
\begin{eqnarray*}
&& \left \vert e_i^{\top}(\mathbf{\hat{B}-B}_{0})e_j \right \vert\\ &=&\left\vert  e_i^{\top}(\mathbf{\hat{B}}%
-\mathbf{B}^{k})e_j +  e_i^{\top}(\mathbf{B}^{k}-\mathbf{B}_{0}^{k})e_j+ e_i^{\top}(\mathbf{B}_{0}^{k}-\mathbf{B}_{0})e_j
\right \vert \\
&\leq &\left \vert e_i^{\top}( \mathbf{\hat{B}}-\mathbf{B}^{k}) e_j\right \vert +\left \vert
e_i^{\top}(\mathbf{\mathbf{B}^{k}-\mathbf{B}_{0}^{k}})e_j\right \vert +\left \vert e_i^{\top}( \mathbf{%
{B}_{0}^{k}-B}_{0} )e_j\right \vert \\
&\defeq&I_1+I_2+I_3\\
\end{eqnarray*}

The following statement are what we need to to prove, and will lead to
$\left \vert  e_i^{\top}(\mathbf{\hat{B}}-\mathbf{B}_{0})e_j\right \vert =\Co_p(1)$.

\begin{eqnarray*}
I_1 & = &| e_i^{\top}(\mathbf{\hat{B}}-\mathbf{B}^{k})e_j|\\
&=&|\frac{1}{ M_G^2|D_G|}\sum_{g} e_i^{\top}\hat{Z}_{g}^{\top }\hat{Z}_{g}e_j+\frac{1}{M_G^2|D_G|}\sum_{g}\sum_{h(\neq g)}k(d_{gh}) e_i^{\top}\hat{Z}_{g}^{\top }\hat{Z}_{h}e_j\\&&
-[ \frac{1}{M_G^2|D_G|}\sum_{g} e_i^{\top}Z_{g}^{^{\top }}Z_{g}e_j+\frac{1}{M_G^2|D_G|}
\sum_{g}\sum_{h (\neq g)}k(d_{gh}) e_i^{\top}Z_{g}^{^{\top }}Z_{h} e_j]|=\smallO_p\left( 1\right)\end{eqnarray*}
\begin{eqnarray*}
I_2 &=& |  e_i^{\top}(\mathbf{B}^{k}-\mathbf{B}_{0}^{k})e_j| \\&=&|\frac{1}{M_G^2|D_G|}\sum_{g} e_i^{\top}\left[ Z_{g}^{^{\top }}Z_{g}-\E%
\left( Z_{g}^{\top }Z_{g}\right) \right]e_j \\
&&+\frac{1}{M_G^2|D_G|}%
\sum_{g}\sum_{h(\neq g)}k(d_{gh}) e_i^{\top}\left[ Z_{g}^{^{\top }}Z_{h}-%
\E\left( Z_{g}^{^{\top }}Z_{h}\right) \right] e_j|\\
&=&\smallO_p\left( 1\right)
\end{eqnarray*}
\begin{eqnarray*}
I_3&=& \left \vert e_i^{\top} (\mathbf{B}_{0}^{k}-\mathbf{B}_{0})e_j\right \vert\\ &=&|
\frac{1}{|D_G| M_G^2}\sum_{g}\sum_{h(\neq g)}k(d_{gh}) e_i^{\top}\E(
Z_{g}^{^{\top }}Z_{h})e_j {-}\frac{1}{G M_G^2}\sum_{g}\sum_{h(\neq g)}e_i^{\top}\E\left[ Z_{g}^{^{\top }}Z_{h}\right]e_j |\\
&=&\frac{1}{|D_G| M_G^2}\sum_{g}\sum_{h(\neq g)}|k(d_{gh})-1
| e_i^{\top} \E\left( Z_{g}^{^{\top }}Z_{h}\right) e_j|\\&=&\smallO_p\left( 1\right)
\end{eqnarray*}
%
To prove each of $I_1, I_2, I_3$ is $\Co_p(1)$, we define
$p_{gh}=Z_{g}^{^{\top }}Z_{h}-\E\left( Z_{g}^{^{\top
}}Z_{h}\right) $.

\textbf{Step 1}
We handle firstly $I_1$,
$I_1 \leq |M_G^{-2}|D_G|^{-1} \sum_g\sum_h e_i^{\top} (\hat{Z}_g - Z_g)^{\top}Z_h e_j K(d_{gh})| \\+  |M_G^{-2}|D_G|^{-1} \sum_g\sum_h e_i^{\top} (\hat{Z}_h- Z_h^{\top})(\hat{Z}_g - Z_g)e_j K(d_{gh})|+ | M_G^{-2}|D_G|^{-1} \sum_g\sum_h e_i^{\top}Z_g^{\top}(\hat{Z}_h - Z_h)e_j K(d_{gh})|
\defeq I_{11} +I_{12}+I_{13}.$
{Assume that $\hat{Z}_g - Z_g = (\nabla \mathbf{m}^{\top}_{g}\mathbf{W}_{g}^{-1}(\hat{\mathbf{u}}_{g}-\mathbf{u}_{g} ) )\\
= ( \nabla\mathbf{m}_{g}^{\top}\mathbf{W}_{g}^{-1}C_ g \Delta_g)$.  $\sum_g|C_g|_2 = \CO_p(L G)$ and $|\Delta_g|_2 = \CO_p(G^{-1/2}),$ where recall that $|.|_2$ defined the Euclidean norm of a matrix.}
Thus we have $I_{11} = M_G^{-2}|D_G|^{-1} \sum_g\sum_h |e_i^{\top}  (\hat{Z}_g - Z_g)^{\top}Z_h  e_j K(d_{gh})| = M_G^{-2}|D_G|^{-1} \sum_g\sum_h  |e_i^{\top} \nabla\mathbf{m}_{g}\mathbf{W}_{g}^{-1}C_g \Delta_g  Z_h  e_j K(d_{gh})| \\
\leq   M_G^{-2}|D_G|^{-1} \sum_g   |e_i^{\top} \nabla\mathbf{m}_{g}\mathbf{W}_{g}^{-1}C_g \Delta_g|_2|  \mbox{max}_{\rho(h,g)\leq h_g }Z_h  e_j|_2
\\ \leq   M_G^{-2}|D_G|^{-1} \sum_g|e_i^{\top} \nabla\mathbf{m}_{g}\mathbf{W}_{g}^{-1}C_g|_2 |\Delta_g|_2|\mbox{max}_{h:\rho(h,g)\leq h_g }Z_h  e_j|_2 =\CO_p(h_g^{d/q'}L^{d/q'} /\sqrt{G}),$ given the fact that the number of observations lying in a $h_g$ ball is $\{\sharp h:  \rho(h,g)\leq h_g\}  \lesssim C h_g^d L^d ,$
$(\E|\max_{h:  \rho(h,g)\leq h_g} Z_h|^2)^{1/2} \leq C h_g^{d/q'}\max_{h:  \rho(h,g)\leq h_g}\|Z_h\|_{q'} L^{d/{q'}}$, where from B.2) we have that $ \max_{h:  \rho(h,g)\leq h_g}\|Z_h\|_{q'} \leq C L^2$.

$I_{12} = M_G^{-2}|D_G|^{-1} \sum_g \sum_h e_i^{\top} (\hat Z_g - Z_g)^{\top} (\hat{Z}_h - Z_h) K(d_{gh})e_j = M_G^{-2}|D_G|^{-1} \sum_g \sum_h e_i^{\top} (\hat Z_g - Z_g)^{\top} (\hat{Z}_h - Z_h) K(d_{gh})e_j\leq M_G^{-2}|D_G|^{-1} \sum_g \sum_h e_i^{\top} \nabla \mm_g^{\top} \mmW_g^{-1}C_g\Delta_g (\nabla \mm_h^{\top}\mmW_h^{-1} C_h\Delta_h)^{\top}  K(d_{gh})e_j \\ \leq |e_i^{\top} \nabla \mm_g^{\top} \mmW_g^{-1}C_g|_2 |\Delta_g |_2|\Delta_h^{\top} |_2  |C_h^{\top} \mmW_h^{-1} \nabla \mm_h^{\top}e_j|_2 = \CO_p( h_g^{d/q'} |D_G|^{-1}L^{d/q'}L^2).$
The rate of $I_{13}$ is similarly derived as $I_{11}$.
Then from B.1) $I_1 = \Co_p(1).$

\textbf{Step 2}
Now we look at the variance case $I_2,$

$I_2 = \frac{1}{|D_G| M_G^2}%
\sum_{g}\sum_{h \neq g  }|k(d_{gh}) e_i^{\top}\left[ Z_{g}^{^{\top }}Z_{h}-%
\E\left( Z_{g}^{^{\top }}Z_{h}\right) \right] e_j \vert
=\smallO_p\left( 1\right).$\\
As we can see that $\E I_2 = 0$ and we need to study \\$\Var(I_2) = |D_G|{-2} M_G^{-4} \sum_{g1} \sum_{h1} \sum_{g2} \sum_{h2} k(d_{g1h1})k(d_{g2h2})\\ \E \{e_i^{\top}\left[ Z_{g1}^{^{\top }}Z_{h1}-%
\E\left( Z_{g1}^{^{\top }}Z_{h1}\right) \right] e_je_i^{\top}\left[ Z_{g2}^{^{\top }}Z_{h2}-%
\E\left( Z_{g2}^{^{\top }}Z_{h2}\right) \right] e_j\}.$

Denote $p_{g1h1,ij} \defeq e_i^{\top}\left[ Z_{g1}^{^{\top }}Z_{h1}-%
\E\left( Z_{g1}^{^{\top }}Z_{h1}\right) \right] e_j.$

{According to assumption A.1)- A.8), the underlying random field $\tilde{\varepsilon}$ with $\alpha-$ mixing $\tilde{\alpha}(u,v,r)\leq (uL+vL)^{\tau}\hat{\alpha}(r)$, with $\tau \geq 0.$
 We need to verify  the NED property of $p_{g1h1,ij}$.}

{From Lemma \ref{lem2:consistency}, the NED property of $Z_g = s_g(\theta^0,\gamma^0)$ with $\psi(m) $ and NED constant bounded by $L^2d_gC'$, where $C'$ is a bound for the $\max_{i,j}\||\nabla m_g(\theta^0) \mathbf{W}_{gij}^{-1}(\theta^0,\gamma^0)|_2\|_4$.}
According to the definition of Bartlett kernel we focus on the pairs with $\rho(h1,g1) \leq h_g$ and $\rho(h2,g2) \leq h_g$,  we see that
$p_{g1h1,ij}$, $\|Z_{h1}^{\top}Z_{g_1} - \E [Z_{h1}^{\top}Z_{g1}|\mathcal{F}_{h1}(s+h_g)]\| \leq (\||Z_{h1}|_2\|_4 d_{g1}\vee \||Z_{g1}|_2\|_4 d_{h1} )\psi(s)$.

Therefore $p_{g1h1,ij}$ would be also $L_2$ NED with ${\psi}(m) = \tilde{\psi}(m+h_g), $ with $m>h_g$.

From the property of the $L_2$ NED, following from Lemma B.3 of \cite{Jenish2012},

$\Cov(p_{g1h1,ij}, p_{g2h2,ij}) =  \E \{e_i^{\top}\left[ Z_{g1}^{\top }Z_{h1}-%
\E\left( Z_{g1}^{^{\top }}Z_{h1}\right) \right] e_je_i^{\top}\left[ Z_{g2}^{^{\top }}Z_{h2}-%
\E\left( Z_{g2}^{^{\top }}Z_{h2}\right) \right] e_j\} \\\leq \|p_{g1h1,ij}\|_{2+\delta}\{C_1 \|p_{g1h1,ij}\|_{2+\delta}[\rho(g1,g2)/3]^{d\tau^*}\hat{\alpha}^{\delta/(2+\delta)}(\rho(g1,g2)/3)+ C_2 \tilde{\psi}([\rho(g1,g2)]/3)\},$
where $\tau^* \defeq \delta \tau/(2+\delta)$.

So $\Var(I_2) = M_G^{-4}|D_G|^{-2} h_g^{2d} L^{2d}  \sum_{g1} \sum_{g2} \mbox{max}_{h1,h2} k(d_{g1h1})k(d_{g2h2}) \E \{e_i^{\top}\left[ Z_{g1}^{^{\top }}Z_{h1}-%
\E\left( Z_{g1}^{^{\top }}Z_{h1}\right) \right] e_j\\e_i^{\top}\left[ Z_{g2}^{^{\top }}Z_{h2}-%
\E\left( Z_{g2}^{^{\top }}Z_{h2}\right) \right] e_j\} $
$\leq M_G^{-4}|D_G|^{-2} h_g^{2d} L^{2d}\mbox{max}_{h1,h2} \sum_{g1,g2}\|p_{g1h1,ij}\|_{2+\delta}\{C_1 \|p_{g1h1,ij}\|_{2+\delta}\\ \{\rho(g1,g2)/3\}^{d\tau^*}\hat{\alpha}^{\delta/(2+\delta)}(\rho(g1,g2)/3)+ C_2 \tilde{\psi}(\rho(g1,g2)/3)\}\\ \leq M_G^{-4}G^{-2} h_g^{2d} \mbox{max}_{h1,h2} \sum_{g1}\sum^{\infty}_{r=1} \sum_{g2 \in \{g_2: \rho_{g1,g2} \in [r, r+1)\}}\|p_{g1h1,ij}\|_{2+\delta}\{C_1 \|p_{g1h1,ij}\|_{2+\delta}[\rho(g1,g2)/3]^{d\tau^*}\\ \hat{\alpha}^{\delta/(2+\delta)}(\rho(g1,g2)/3)+ C_2 \psi(([\rho(g1,g2)]/3-h_g))_{+}\}\\
  \leq |D_G|^{-1} h_g^{2d} L^{2d}\sum^{\infty}_{r=1}  \{C'_1 r^{(d\tau^*+d)-1}\hat{\alpha}^{\delta/(2+\delta)}(r)+ C_2r^{d-1}\psi((r-h_g)_{+})\}$.
 { From B.4) we assume that $h_g^{2d} L^{2d}\sum^{\infty}_{r=1} r^{(d\tau^*+d)-1}\hat{\alpha}^{\delta/(2+\delta)}(r)= \CO(G)$, and $h_g^{2d} \sum^{\infty}_{r=1}  L^{2d} r^{d-1}\psi((r-h_g)_{+}) = \CO(G), $ then we have $\Var(I_2) = \Co(1).$ }

\textbf{Step 3 }

{According to B.4), $|k(d_{gh})-1|\leq C_k |\rho(g,h)/h_g|^{\rho_K}$ for $\rho(g,h)/h_g\leq 1$ for some constant $\rho_k\geq 1$ and $0<C_k<\infty$.}

We handle the bias term $I_3$,\\
$M_G^{-2}|D_G|^{-1}\sum_g \sum_h |e_i^{\top}(k(\rho(g,h)/h_g)-1)\E(Z_g^{\top}Z_h)e_j|
 \\\leq M_G^{-2}|D_G|^{-1}\sum_g \sum_h C_k |\rho(g,h)/h_g|^{\rho_k}e_i^{\top}\E(Z_g^{\top}Z_h)e_j\\ \leq M_G^{-2}|D_G|^{-1}\sum_g \sum_h  |\rho(g,h)/h_g|^{\rho_k}\|e_i^{\top}Z_g^{\top}\|\|Z_he_j\| .$

{Also according B.4), $M_G^{-2}|D_G|^{-1}\sum_g \sum_h  |\rho(g,h)/h_g|^{\rho_k}\|e_i^{\top}Z_g^{\top}\|\|Z_he_j\|$ is $\Co(1)$.}

\subsection{Two special cases}\label{ass2}
To justify the NED assumptions in A.2), we now verify the two $L_2$ NED properties in our example. ($L_4$ NED can be similarly verified.)  In particular we would like to analyze how the underlying assumptions of the data innovation processes would induce the assumption of A.1).
\subsubsection{Poisson Regression/ Negative Binomial}
The focused model is
$y_{n,i}$s are poisson counts observations,
$\E(y_{n,i}| x_{n,i}, v_{n,i}) = \exp(x_{n,i}^{\top}\beta)v_{n,i} $.
We suppose that $v_{i,n} = g(\eta_{i,n})$, where $g(.)$ is twice continuously differentiable function. For example $g(x) = \exp(x)$ and then $\E(y_{n,i}| x_{n,i}, v_{n,i}) = \exp(x_{n,i}^{\top}\beta+ \eta_{n,i} ),$  and $x_{n,i}$ are controls with $p \times 1$ dimension.

We assume that $\eta_{n,i}$ follows a spatial autoregressive model. Namely \\$\eta_{n,i} = \lambda\sum^n_{j=1}w_{n,ij}\eta_{n, j} + \epsilon_{n,i}.$
Suppose $\eta_n = \lambda W \eta_n+ \epsilon_n$, and $\eta_n = (I- \lambda W)^{-1}\epsilon_n$, define $ [a_{ij}] = (I -\lambda W)^{-1}.$

Then we have
$$v_{n,i} = g(\sum_{j=1} a_{ij} \epsilon_{n,j}).$$

For the moment we assume the decomposition:
$y_{n,i} = \E(y_{n,i}| x_{n,i}, v_{n,i}) +\varepsilon_{n,i}.$

Assume that $\{\xi_{n,i} = (x_{n,i}, \epsilon_{n,i}, \varepsilon_{n,i})\}$ are {mixing random field}.

We now establish that $Y = \{y_{n,i},s_i \in D_n , n \geq 1\}$ is uniform $L_2$ NED on $\xi = \{\xi_{n,i}, s_i \in D_n, n \geq 1\}.$ Define $\mathcal{F}_{n,i}(s) =  \sigma(\xi_{n,j}: j \in D_{n}, \rho(i, j) \leq s).$

It can be seen that, for any $i \in D_n,$
\begin{eqnarray*}
\tilde{y}_{n,i} = y_{n,i} - \E(y_{n,i}| \mathcal{F}_{n,i}(s))& = &\exp(x_{n,i}^{\top}\beta)v_{n,i}  +\varepsilon_{n,i}  -  \exp(x_{n,i}^{\top}\beta)\E (v_{n,i}| \mathcal{F}_{n,i}(s)) - \varepsilon_{n,i} \\
&=& [v_{n,i}  - \E \{v_{n,i}| \mathcal{F}_{n,i}(s)\}] \exp(x_{n,i}^{\top}\beta)
\end{eqnarray*}
As $v_{n,i} -  \E (v_{n,i}| \mathcal{F}_{n,i}(s)) =  g(\sum_j a_{ij} \epsilon_{n,j}) - \E \{ g(\sum_j a_{ij} \epsilon_{n,j})| \mathcal{F}_{n,i}(s)\}.$

Taylor expansion to the first order yield,
\begin{eqnarray}
g(\sum_j a_{ij} \epsilon_{n,j}) - \E \{ g(\sum_j a_{ij} \epsilon_{n,j})| \mathcal{F}_{n,i}(s)\} = g'(\tilde{a})\sum_{j \in B^c(s)} a_{ij} \epsilon_{n,j},
\end{eqnarray}
where $\tilde{a}$ is a point between $0$ and $\sum_j a_{ij} \epsilon_{n,j}$, $B^c(s)$ is the set of $j$ with $\rho(i,j)\geq s$.
Thus we have
\begin{eqnarray}
(\E|\tilde{y}_{n,i}|^2)^{1/2} \leq C \sum_{j \in B^c(s)} |a_{ij}|,
\end{eqnarray}
where we assume that $\|g'(\tilde{a})\epsilon_{n,j}\|_2$ is uniformly bounded by $C$.
Also we require that
$\mbox{limsup}_{s \to \infty}  \sup_{i\in D_n} \sum_{j \in B^c(s)}|a_{ij}| \to 0.$ The proof is completed.

\subsubsection{Probit Model}
We now prove the case of probit model,
\begin{eqnarray*}
y_{n,i} = \IF(y_{n,i}^* >0)\\
y^*_{n,i} = x_{n,i}^{\top}\beta + e_{n,i}.
\end{eqnarray*}
And
$e_{n,i} = \lambda \sum_jw_{n,ij} e_{n,j} + v_{n,i}.$
We now establish that $Y = \{y_{n,i},s_i \in D_n , n \geq 1\}$ ($\|y^*_{n,i}\|_2< \infty$) is $L_2$ NED on $\xi = \{(x_{n,i}, e_{n,i}), s_i \in D_n, n \geq 1\}.$ Thus again similar to the previous case we can denote $e_{n,i} = \sum_{j}a_{ij} v_{n,i}, $ where $a_{ij}$ are the matrix entries of $(I- \lambda W)^{-1}$.

\begin{proof}
First of the latent process is $\{y^*_{n,i}\}$ is a special case of the Cliff-Ord type of process, and therefore would be $L_2-$ unform NED if
$\mbox{limsup}_{s \to \infty}  \sup_{i\in D_n} \sum_{j \in B^c(s)}|a_{ij}| \to 0 $, and $\| v_{n,i}\|_{r'}\leq \infty,$ $r'=2$.\\
 For any $\epsilon > 0$, define the event $B = \{ |y_{n, i}^{*}| < \epsilon, |\E[y_{n, i}^{*}|\mathcal{F}_{n, i}(s)]| < \epsilon \}$. Using $|\II(x_{1} \geq 0) - \II(x_{2} \geq 0) \leq \frac{|x_{1} - x_{2} |}{\epsilon} \II(x_{1} > \epsilon \text{ or } x_{2} > \epsilon ) + \II(x_{1} < \epsilon,  x_{2} < \epsilon )$, we have
\begin{align*}
  & \| y_{n, i} - \E[y_{n, i}|\mathcal{F}_{n, i}(s)] \| =  \| \II(y_{n, i}^{*} \geq 0) - \E[ \II(y_{n, i}^{*} \geq 0)|\mathcal{F}_{n, i}(s)] \|\\
   & \leq \| \II(y_{n, i}^{*} \geq 0) - \II\{ \E[y_{n, i}^{*}|\mathcal{F}_{n, i}(s)] \geq 0 \} \| = \left\{ \E\left |  \II(y_{n, i}^{*} \geq 0) - \II\{ \E[y_{n, i}^{*}|\mathcal{F}_{n, i}(s)] \geq 0 \}  \right |^{2} \right \}^{\frac{1}{2}} \\
   & \leq  \left\{ \frac{1}{\epsilon^{2}} \int_{B^{c}}\left |  y_{n, i}^{*} - \E[y_{n, i}^{*}|\mathcal{F}_{n, i}(s)]  \right |^{2} d\P + \int_{B} d\P \right \}^{\frac{1}{2}} \\
   & \leq  \left\{ \frac{1}{\epsilon^{2}} \int_{B^{c}}\left |  y_{n, i}^{*} - \E[y_{n, i}^{*}|\mathcal{F}_{n, i}(s)]  \right |^{2} d\P \right \}^{\frac{1}{2}} + \left\{ \int_{B} d\P \right \}^{\frac{1}{2}} \\
   & \leq  \frac{1}{\epsilon} \| y_{n, i}^{*} - \E[y_{n, i}^{*}|\mathcal{F}_{n, i}(s)] \|_{2}  + \pi_{4}\epsilon^{1/2}
   , \quad \text{for some constant $\pi_{4}>0$},
\end{align*}
where the first inequality is based on Therorem 10.12 of \cite{Davidson1994} by taking $\II\{ \E[y_{n, i}^{*}|\mathcal{F}_{n, i}(s)] \geq 0 \}$ as an approximation of $\II(y_{n, i}^{*} \geq 0)$ with measure $\mathcal{F}_{n, i}(s)$. When taking $\epsilon = \| y_{n, i}^{*} - \E[y_{n, i}^{*}|\mathcal{F}_{n, i}(s)] \|^{q}, 0 < q < 1 $, when $\epsilon$ converges to 0, both terms converge to 0 at a slower rate than $\| u_{n, i}^{*} - \E[y_{n, i}^{*}|\mathcal{F}_{n, i}(s)] \|$, therefore, the process $\{ (y_{n, i}) \}_{i=1}^{n} $ is uniform $L_2$ NED.
\end{proof}

\subsection{Exponential family}\label{exp}
For parameter $\theta \in \mathbf{R}^p$, and a random variable $X$.
$f(x,\theta) = h(x)\mbox{exp} \{\theta^{\top}T(x)- A(\theta)\},$
where $A(\theta) = \log \int h(x)\mbox{exp} \{\theta^{\top}T(x)\}d F(x)$ is the \textit{cumulant function}, and $T(x)$ is referred to as the \textit{\textit{sufficient statistics}}.
In particular, we know that  $\partial A(\theta)/\partial \theta = \E(T(X)) $ and $\partial A(\theta)/\partial \theta \partial \theta^{\top}= \Var(T(X)) = I (\theta)$ are regarded as the Fisher information matrix.

{
Suppose $y_i$ is following an exponential family condition on $x_i$, then the conditional mean and conditional variance function will be both expressed as known function, which is the first and the second derivative of the cumulants generating function $A(\mu_i)$. In particular $\E(T(y_i)) = \partial A(\mu_i)/\partial \mu_i|_{\mu_i = v(x_i^{\top}\theta)},$ and the variance covariance  $\Var(T(y_i)) = \partial A(\mu_i)/\partial \mu_i \partial \mu_i^{\top}|_{\mu_i = v(x_i^{\top}\theta)},$ where $v(\cdot)$ is a link function. Notably the variance covariance function is thus treated as a known function related to the conditional mean in this case as they are both related to $A(\cdot)$.}

\bibliography{Literature}

\end{document}